\documentclass[aihp]{imsart}

\RequirePackage{amsthm,amsmath,amsfonts,amssymb}
\RequirePackage[numbers]{natbib}
\RequirePackage[colorlinks,citecolor=blue,urlcolor=blue]{hyperref}
\RequirePackage{graphicx}
\graphicspath{ {./figures/} }
\DeclareGraphicsExtensions{.pdf}

\startlocaldefs

\usepackage[utf8]{inputenc}
\usepackage[T1]{fontenc}
\usepackage{ucs}
\usepackage{xcolor}
\usepackage{dsfont}
\usepackage{physics}
\usepackage{mathtools}
\usepackage[inline]{enumitem}
\usepackage[ruled]{algorithm2e}
\SetKwInput{KwInput}{Input}
\SetKwInput{KwOutput}{Output}
\usepackage{tikz,pgfplots}
\usepackage{pbox}

\usepackage{constants}

\newconstantfamily{ctrlcst}{
	symbol=\kappa
}
\newconstantfamily{csts}{
	symbol=C
}
\renewconstantfamily{normal}{
	symbol=c
}

\makeatletter
\newcommand{\pushright}[1]{\ifmeasuring@#1\else\omit\hfill\(\displaystyle#1\)\fi\ignorespaces}
\newcommand{\pushleft}[1]{\ifmeasuring@#1\else\omit\(\displaystyle#1\)\hfill\fi\ignorespaces}
\makeatother

\renewcommand{\norm}[1]{\|#1\|}

\newcommand{\normsup}[1]{\left\|#1\right\|_{\scriptscriptstyle\infty}}

\renewcommand{\emptyset}{\varnothing}
\newcommand{\setof}[2]{\{#1\,:\,#2\}}
\newcommand{\bsetof}[2]{\bigl\{#1\,:\,#2\bigr\}}
\newcommand{\Bsetof}[2]{\Bigl\{#1\,:\,#2\Bigr\}}

\newcommand{\lrangle}[1]{\langle #1 \rangle}

\newcommand{\Z}{\mathbb{Z}}
\newcommand{\R}{\mathbb{R}}
\newcommand{\Rd}{\mathbb{R}^d}
\newcommand{\Zd}{\mathbb{Z}^d}

\newcommand{\pathSet}{\mathrm{Path}}

\newcommand{\sfe}{\mathsf{e}}
\newcommand{\sfo}{\mathsf{o}}
\newcommand{\sfO}{\mathsf{O}}

\newcommand{\bbB}{\mathbb{B}}
\newcommand{\bbC}{\mathbb{C}}

\newcommand{\bbE}{\mathbb{E}}

\newcommand{\bbG}{\mathbb{G}}

\newcommand{\bbN}{\mathbb{N}}

\newcommand{\bbQ}{\mathbb{Q}}
\newcommand{\bbR}{\mathbb{R}}
\newcommand{\bbS}{\mathbb{S}}

\newcommand{\bbZ}{\mathbb{Z}}

\newcommand{\calE}{\mathcal{E}}

\newcommand{\calM}{\mathcal{M}}

\newcommand{\calQ}{\mathcal{Q}}

\newcommand{\calT}{\mathcal{T}}
\newcommand{\calU}{\mathcal{U}}

\newcommand{\calW}{\mathcal{W}}

\newcommand{\rmc}{\mathrm{c}}

\newcommand{\rme}{\mathrm{e}}

\newcommand{\rmm}{\mathrm{m}}

\newcommand{\rmp}{\mathrm{p}}

\newcommand{\rmr}{\mathrm{r}}

\newcommand{\given}{\,|\,}

\newcommand{\fcone}{\mathcal{Y}^\blacktriangleleft}
\newcommand{\bcone}{\mathcal{Y}^\blacktriangleright}
\newcommand{\bend}{\mathbf{b}}
\newcommand{\fend}{\mathbf{f}}
\newcommand{\diam}{D}

\newcommand{\CPts}{\textnormal{CPts}}

\newcommand{\slab}{\mathrm{slab}}

\newcommand{\concatenate}{\circ}
\newcommand{\displace}{X}

\newcommand{\SetRootMarkBackCont}{\mathfrak{B}_L}
\newcommand{\SetRootMarkForwCont}{\mathfrak{B}_R}
\newcommand{\SetRootDiaCont}{\mathfrak{A}}
\newcommand{\irreducible}{\mathrm{irr}}

\newcommand{\surcharge}{\mathfrak{s}}

\newcommand{\Skeleton}{\mathrm{Sk}}
\newcommand{\SkeletonSet}{\mathcal{S}}
\newcommand{\longEdge}{\mathrm{long}}
\newcommand{\shortEdge}{\mathrm{short}}

\newcommand{\Bad}{\mathrm{Bad}}

\newcommand{\qco}{q^{\mathrm{c}}}

\newcommand{\betac}{\beta_{\mathrm c}}
\newcommand{\betasat}{\beta_{\rm sat}}

\newcommand{\mix}{\mathrm{mix}}

\newcommand{\RWP}[1]{\operatorname{\mathrm{P}^{\mathrm{RW}}_{\mathit{#1}}}}

\newcommand{\RWG}{G^{\mathrm{RW}}}

\theoremstyle{plain}
\newtheorem{theorem}{Theorem}[section]
\newtheorem{lemma}[theorem]{Lemma}

\newtheorem{corollary}[theorem]{Corollary}
\newtheorem{conjecture}[theorem]{Conjecture}
\newtheorem{remark}{Remark}[section]

\newtheorem{claim}{Claim}
\newtheorem{definition}{Definition}[section]
\newtheorem{obs}{Observation}

\endlocaldefs

\begin{document}

\begin{frontmatter}

\title{Ornstein--Zernike behavior for Ising models\\with infinite-range interactions}
\runtitle{OZ behavior for Ising models with infinite-range interactions}

\begin{aug}
	\author[A]{\inits{Y.}\fnms{Yacine}~\snm{Aoun}\ead[label=e1]{Yacine.Aoun@unige.ch}},
	\author[B]{\inits{S.}\fnms{Sébastien}~\snm{Ott}\ead[label=e2]{ott.sebast@gmail.com}}
	\and
	\author[A]{\inits{Y.}\fnms{Yvan}~\snm{Velenik}\ead[label=e3]{Yvan.Velenik@unige.ch}}
	%%%%%%%%%%%%%%%%%%%%%%%%%%%%%%%%%%%%%%%%%%%%%%
	%% Addresses                                %%
	%%%%%%%%%%%%%%%%%%%%%%%%%%%%%%%%%%%%%%%%%%%%%%
	\address[A]{Section de Mathématiques,
		Université de Genève, 1205 Genève, Switzerland\printead[presep={,\ }]{e1,e3}}
	
	\address[B]{Département de Mathématiques,
		Université de Fribourg, 1700 Fribourg, Switzerland\printead[presep={,\ }]{e2}}
		
\end{aug}

\begin{abstract}
	We prove Ornstein--Zernike behavior for the large-distance asymptotics of the two-point function of the Ising model above the critical temperature under essentially optimal assumptions on the interaction. The main contribution of this work is that the interactions are not assumed to be of finite range. To the best of our knowledge, this is the first proof of OZ asymptotics for a nontrivial model with infinite-range interactions.
	
	Our results actually apply to the Green function of a large class of ``self-repulsive in average'' models, including a natural family of self-repulsive polymer models that contains, in particular, the self-avoiding walk, the Domb--Joyce model and the killed random walk.
	
	We aimed at a pedagogical and self-contained presentation.
\end{abstract}

\begin{abstract}[language=french]
	Nous prouvons, sous des hypothèses essentiellement optimales sur l'interaction, que le comportement asymptotique de la fonction à 2-point du modèle d'Ising au-dessus de sa température critique prend la forme prédite par Ornstein et Zernike. La contribution principale de ce travail est que nous ne supposons pas l'interaction de portée finie. À notre connaissance, il s'agit de la première preuve du comportement Ornstein-Zernike pour un modèle non trivial avec des interactions de portée infinie.
	
	Nos résultats s'appliquent plus généralement à la fonction de Green d'une grande classe de modèle «auto-répulsifs en moyenne», incluant une famille naturelle de modèles de polymère auto-répulsif à laquelle appartiennent, en particulier, la marche aléatoire auto-évitante, le modèle de Domb-Joyce et la marche aléatoire tuée.
	
	Nous nous sommes efforcés de rendre notre présentation aussi pédagogique et complète que possible.
\end{abstract}



\end{frontmatter}

\vspace*{5mm}
\hfill\pbox{7cm}{\textit{In memory of Dima Ioffe}\\\textit{a brilliant mathematician }\\\textit{and, above all, a great friend}}

%%%%%%%%%%%%%%%%%%%%%%%%%%%%%%%%%%%%%%%%%%%%%%%%%%%%%%%%%%%%%%%%%%%%%%%%%%%%%%%%%%%%%%%%%%%%%%%%%%%

%%%%%%%%%%%%%%%%%%%%%%%%%%%%%%%%%%%%%%%%%%%%%%%%%%%%%%%%%%%%%%%%%%%%%%%%%%%%%%%%%%%%%%%%%%%%%%%%%%%%%%%%
\section{Introduction and results}
%%%%%%%%%%%%%%%%%%%%%%%%%%%%%%%%%%%%%%%%%%%%%%%%%%%%%%%%%%%%%%%%%%%%%%%%%%%%%%%%%%%%%%%%%%%%%%%%%%%%%%%%

\vspace*{5mm}

\subsection{Introduction}

The central objects in this work are Green functions of the form
\[
G(x) = \sum_{\gamma:0\to x} q(\gamma),
\]
where the sum is over paths in \(\Zd\) starting at \(0\) and ending at \(x\), that is, sequences of vertices \(\gamma = (\gamma_0=0, \gamma_1, \dots, \gamma_{n-1}, \gamma_n=x)\) of arbitrary finite length. The weight \(q(\cdot)\) is assumed to satisfy a number of conditions,
which will be stated precisely in Section~\ref{ssec:ConditionsOnPaths}. For this introduction, it suffices to say that they are general enough to enable the analysis of a variety of important quantities in statistical mechanics, among which the two-point function of the ferromagnetic Ising model on \(\Zd\) above the critical temperature,
or the Green functions of the large class of self-repulsive polymer models considered in~\cite{Ioffe+Velenik-2008}, which includes the self-avoiding walk, the Domb--Joyce model, as well as the killed random walk (or equivalently, the covariances of the massive Gaussian Free Field).

\medskip
Our main result on the Green function \(G\) is the proof that, under suitable assumptions, the latter displays classical Ornstein--Zernike (OZ) behavior at large distances, that is,
\[
	G(x) = \frac{\Psi(\hat x)}{\norm{x}^{(d-1)/2}}\, \sfe^{-\nu(\hat x) \norm{x}}\, (1+\sfo(1)),
\]
where \(\norm{x}\) denotes the Euclidean norm of \(x\), \(\hat x = x/\norm{x}\), and \(\Psi\) and the inverse correlation length \(\nu\) are positive functions on the unit sphere.
In fact, our result expresses \(G\) in terms of the Green function of a directed random walk on \(\Zd\), thus also providing higher-order corrections to this leading asymptotics once the latter are established for this (much simpler) directed random walk.

\medskip
Results of this type have a long history, which will be briefly recalled below. The main contribution of the present study is that, for the first time to the best of our knowledge, the steps in \(\gamma\) are not assumed to be of bounded size. This allows us in particular to prove the above OZ behavior for Ising models on \(\Zd\) with interactions of infinite range. More precisely, it becomes possible to consider Ising models with (formal) Hamiltonian of the form
\[
-\sum_{i,j\in\Zd} J_{i,j} \sigma_i\sigma_j,
\]
where the coupling constants are assumed to be nonnegative, translation invariant and centrally symmetric (\(J_{i,j} = J_{i-j} = J_{j-i}\)) and to satisfy some very mild irreducibility condition (see Section~\ref{sec:definitions}).
Of course, in order for \(\nu\) to be (strictly) positive, it is also necessary that the coupling constants decay at least exponentially fast with the distance. 
As we discovered recently, the latter condition is however not always sufficient to ensure OZ behavior. For simplicity, let us consider coupling constants of the form \(J_{i} = \psi(i)\sfe^{-\rho(i)}\), where \(\rho\) is a norm on \(\Rd\) and \(\psi\) a subexponential function. Then, the following \emph{saturation} phenomenon was established (for a large class of models) in~\cite{Aoun+Ioffe+Ott+Velenik-CMP-2021, Aoun+Ioffe+Ott+Velenik-PRE-2021}. To each direction \(s\in\bbS^{d-1}\), we can associate an inverse temperature \(\betasat(s) < \beta_{\rm c}\) (where \(\betasat>0\) if and only if the coupling constants satisfy some explicit condition) such that \(\nu_\beta(s) = \rho(s)\) for all \(\beta<\betasat(s)\), but \(\nu_\beta(s) < \rho(s)\) for all
\(\beta\in(\betasat(s),\betac)\); see Figure~\ref{fig:saturation}.
We show in a companion paper~\cite{Aoun+Ott+Velenik-2021-inPreparation} that OZ behavior is violated when \(\beta < \betasat(s)\) (for general Potts models).
\begin{figure}
	\includegraphics{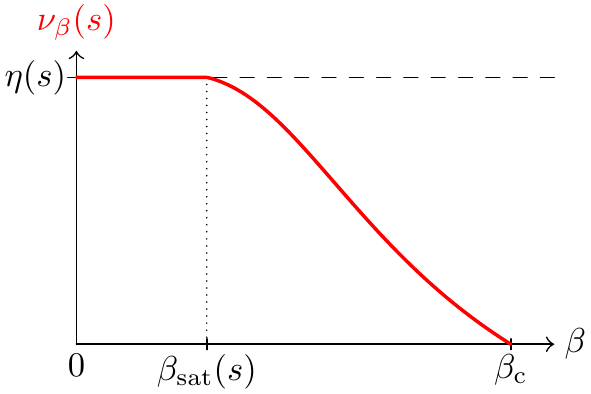}
	\caption{
			The results in this paper apply to the regime \(\beta\in(\betasat(s),\betac)\). The asymptotic behavior of the Green function \(G\) is actually not of Ornstein--Zernike type in the regime \(\beta\in (0,\betasat(s))\).}
	\label{fig:saturation}
\end{figure}

More precisely, the properties of typical paths contributing to the Green function \(G(x)\) change dramatically when \(\beta\) crosses \(\betasat\). Namely, when \(\beta<\betasat\), typical paths are composed of \(\sfo(\norm{x})\) edges (in many cases, a single giant edge connecting a vertex in the vicinity of \(0\) to a vertex in the vicinity of \(x\), with a finite random number of additional edges connecting it to \(0\) and \(x\); this is the so-called ``one-jump'' regime, which is observed in other areas of probability theory and statistical physics as well~\cite{Godreche-2019}). In contrast, when
\(\beta\in(\betasat(s),\betac)\), typical paths contain a number of edges growing linearly with \(\norm{x}\), all of them being microscopic (their size having an exponential upper tail, the longest edge has a length of order \(\log\norm{x}\)); this is required for OZ behavior to occur.

In the present paper, we prove that OZ holds whenever \(\beta\in(\betasat(s),\betac)\) provided some additional assumption (described later in Section~\ref{sec:Saturation}) is satisfied. In particular, this implies OZ behavior whenever the coupling constants decay
superexponentially with the distance (that is, \(\lim_{\norm{x}\to\infty}J_x\sfe^{c\norm{x}} < \infty\) for all \(c\in\bbR\)).%
We also explain how the usual by-products can be extracted, including analyticity and uniform strict convexity properties of the inverse correlations length. 

Following the approach developed in previous works, our basic strategy is to construct a coupling between the paths contributing to the Green function \(G\) and the trajectories of a directed random walk on \(\Zd\). This is rather remarkable, as the path weight \(q\) does not factorize in general, so that this provides a coupling between a process with infinite memory and a Markov process. Although, as mentioned above, our goal in this paper is the derivation of the OZ asymptotics, the existence of this coupling allows considerably more sophisticated applications; see the end of the historical discussion just below for a partial list of applications of such coupling results to a variety of problems in equilibrium statistical physics.%

\subsection*{A brief history.}
As mentioned above, the derivation of OZ behavior has a long history. In this section, we briefly sketch the latter, without even trying to be comprehensive, emphasizing in particular the major contributions Dima Ioffe made to this topic. 

The analysis of the large-distance asymptotics of pair correlation functions goes back more than a century to the seminal works of Ornstein and Zernike~\cite{Ornstein+Zernike-1914, Zernike-1916}. While not mathematically rigorous, their work provided the first derivation of OZ asymptotics; moreover, the approach they introduced has had a major influence on most of the rigorous works that followed.

The first rigorous derivations of OZ behavior were obtained for the planar Ising model~\cite{Wu-1966, Wu+McCoy+Tracy+Barouch-1976} and relied on the explicit computation of the two-point function. These works also showed that OZ behavior does not hold in this setting when \(\beta\geq\beta_{\rm c}\).

The first robust approaches, applicable to more general classes of models and in any dimension, were introduced in~\cite{Abraham+Kunz-1977, Paes-Leme-1978}. They are by nature restricted to perturbative regimes (very high temperatures, very low densities, etc.). Other approaches of this type, relying on a variety of different techniques, were introduced later
(see, for instance, \cite{Bricmont+Frohlich-1985a, Minlos+Zhizhina-1996}) and new works along these lines still regularly appear.

In parallel, an alternative \textit{nonperturbative} approach was developed. It originated with the work~\cite{Chayes+Chayes-1986} on the Green function of the self-avoiding walk, valid in the whole subcritical regime.
This approach was later extended to the two-point connectivity function in Bernoulli percolation~\cite{Campanino+Chayes+Chayes-1991}. The latter two works were breakthroughs, but suffered from several crippling limitations: they were essentially restricted to on-axis directions; the way they established the crucial separation of masses estimate (namely, that the rate of exponential decay of the ``direct'' correlation function is strictly larger than the rate of exponential decay of the Green function) was extremely model-dependent and unlikely to be extendable to more complex situations; finally, factorization properties of the weights played a crucial role.

Dima Ioffe entered this line research in 1998~\cite{Ioffe-1998}. Combining the approach in~\cite{Chayes+Chayes-1986} with techniques from large deviation theory and convex analysis, he was able to treat the Green function of the subcritical self-avoiding walk in arbitrary directions,
as well as obtain new information on the directional dependence of various quantities, including the analyticity of the correlation length as a function of the direction.
A few years later, in a joint work with Massimo Campanino~\cite{Campanino+Ioffe-2002}, they devised a robust coarse-graining method to establish separation of masses, which allowed them to extend the results of~\cite{Campanino+Chayes+Chayes-1991} to arbitrary directions.
Finally, in joint work with the third author~\cite{Campanino+Ioffe+Velenik-2003}, it was shown how exponential mixing properties could be leveraged in conjunction with an adaptation of the coarse-graining of~\cite{Campanino+Ioffe-2002} to prove OZ behavior for finite-range Ising model on \(\Zd\) for any temperature above critical (this was later extended to general Potts models~\cite{Campanino+Ioffe+Velenik-2008}).
An inconvenient feature of~\cite{Campanino+Ioffe+Velenik-2003} was the fact that it did not relate the two-point function of the Ising model to the Green function of a random walk, but rather to the corresponding quantity for a more complicated process with an infinite memory (formulated in terms of a suitable Ruelle operator).
This flaw was corrected more recently by the second and third authors~\cite{Ott+Velenik-2018}, who showed how one could recover independence using techniques originating in the field of perfect simulations~\cite{Comets+Fernandez+Ferrari-2002}; an alternative derivation, more combinatorial in nature, can be found in~\cite{Ioffe+Ott+Shlosman+Velenik-2021} (a version of which can also be found in Section~\ref{sec:facto_meas} of the present paper).
Finally, these methods were adapted to the random-current representation of the Ising model, allowing the first proof of OZ behavior for the Ising model in the presence of a magnetic field~\cite{Ott-2020}.

Let us mention that the approach developed by Dima and his collaborators led to a great variety of applications to important problems in equilibrium statistical mechanics, related to the properties of interfaces in planar lattice systems~\cite{Greenberg+Ioffe-2005, Campanino+Ioffe+Velenik-2008, Hammond-2012, Coquille+Velenik-2012, Coquille+Duminil-Copin+Ioffe+Velenik-2014, Ioffe+Shlosman+Toninelli-2015, Ott+Velenik-2018, Ioffe+Ott+Velenik+Wachtel-2020, Ioffe+Ott+Shlosman+Velenik-2021}, the effect of a stretching force on self-interacting polymers with or without disorder~\cite{Ioffe+Velenik-2008, Ioffe+Velenik-2010, Ioffe+Velenik-2012, Ioffe+Velenik-2012b, Ioffe+Velenik-2012c, Ioffe+Velenik-2013, Ioffe-2015}, the asymptotic behavior of more general correlation functions~\cite{Campanino+Ioffe+Velenik-2004, Campanino+Gianfelice-2009, Campanino+Ioffe+Louidor-2010, Campanino+Gianfelice-2015, Ott+Velenik-2018(2)}, etc.

\subsection{Results}
In this section, we describe the main results of the present work in an informal manner. The precise statements require many concepts and notations that will only be introduced in Section~\ref{sec:definitions} and can be found in Section~\ref{sec:StatementResults}.

\medskip
All the results presented below rely on an assumption that we call NSA and which is discussed in detail in Section~\ref{sec:Saturation}.
This assumption has really two parts. The first one is necessary: the validity of our construction (and of the resulting OZ asymptotics) requires that \(\beta \in (\betasat(s),\betac)\). Indeed, we prove in~\cite{Aoun+Ott+Velenik-2021-inPreparation} that the 2-point function of the Potts model on \(\Zd\) \textit{does not} display OZ behavior when \(\beta \in (0,\betasat(s))\) (under which conditions OZ behavior occurs \textit{at} \(\betasat(s)\) is not yet fully elucidated). However, we also require a stronger assumption, that we believe is always satisfied. We refer the interested reader to Section~\ref{sec:Saturation} for more details. It suffices here to say that this stronger assumption is not needed when any one of the following conditions is fulfilled:
\begin{enumerate*}
	\item the coupling constants decay superexponentially fast with distance, or
	\item the direction \(s\) lies along one of the coordinate axis and the coupling constants are invariant under lattice reflections, or
	\item \(\beta\in (\sup_{s'\in\bbS^{d-1}}\betasat(s'),\betac)\).
\end{enumerate*}

\medskip
Let \(s\in\bbS^{d-1}\), let \(n\in\bbN\) be a large integer. To lighten notation, we assume that \(ns\in\Zd\) (otherwise, once can choose arbitrarily a closest point to \(ns\) on the lattice).
The weight \(q\) induces a probability measure \(q(\gamma)/G(0,ns)\) on paths from \(0\) to \(ns\).
The weight \(q\) does not, in general, factorize: if \(\gamma\) is the concatenation of two paths \(\gamma_1\) and \(\gamma_2\), then \(q(\gamma) \neq q(\gamma_1) q(\gamma_2)\), even when \(\gamma_1\) and \(\gamma_2\) share a single vertex. This is in particular the case in the Ising model, where the weight \(q\) results from the resummation over the loop soup with which the path interacts: this resummation induces an effective self-interaction for the path that is of unbounded range.

It is thus rather remarkable that one can construct a coupling between these paths \(\gamma\) and a directed random walk. This is the content of Theorem~\ref{thm:main_coupling_with_RW}, which we explain informally now.

Let \(f\) be a function from the set of all paths from \(0\) to \(ns\) to the real numbers. We are interested in evaluating expressions of the type
\[
	\sum_{\gamma:\, 0\to ns} q(\gamma) f(\gamma) .
\]
Of course, the Green function corresponds simply to choosing \(f(\gamma)=1\) for all \(\gamma\).

The goal of the construction is to show that, up to a negligible error, one can replace the sum over \(\gamma\) by a sum over strings of subpaths \(\gamma_L, \gamma_1, \dots, \gamma_M, \gamma_R\) (where \(M\geq 0\) is not fixed), the concatenation of which yields a path from \(0\) to \(ns\). Although we will not define precisely the relevant properties of these subpaths here (this will be done in Section~\ref{sec:definitions}), it suffices to say that \(\gamma_L\) and \(\gamma_R\) are constrained to lie in cones, while \(\gamma_1,\dots\gamma_M\) lie in \textit{diamonds} given by the intersection of two cones, as in the following pictures:
\begin{center}
	\includegraphics{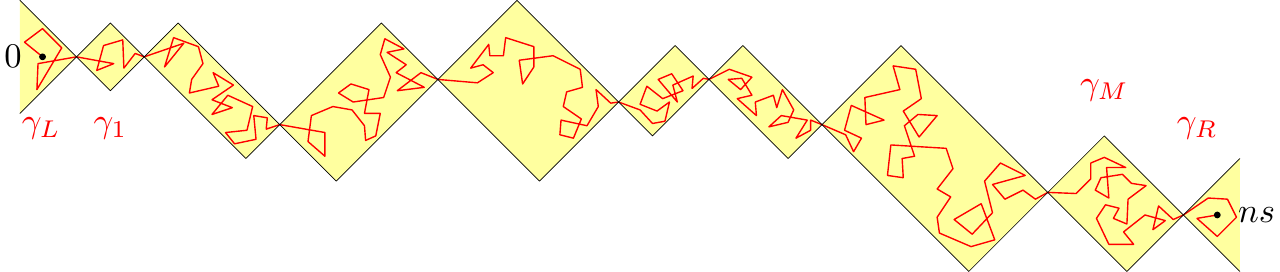}
\end{center} 
Theorem~\ref{thm:main_coupling_with_RW} then states the existence of three probability measures \(\rmp, \rmp_L, \rmp_R\) and two constants \(C_{LR}\in\bbR_{>0}\) and \(c>0\) such that
\[
	\sfe^{n\nu(s)} \sum_{\gamma:\, 0\to ns} q(\gamma) f(\gamma)
	=
	C_{LR} \sum_{\gamma_L,\,\gamma_R} \sum_{M\geq 0} \sum_{\gamma_1, \dots, \gamma_M} f(\bar{\gamma}) \mathds{1}_{\{\bar\gamma:\,0\to ns\}} \rmp_L(\gamma_L) \rmp_R(\gamma_R) \prod_{i=1}^n \rmp(\gamma_i)
	+
	\sfO(\sfe^{-cn}),
\]
where we have written \(\bar\gamma = \gamma_L\circ\gamma_1\circ\dots\circ\gamma_M\circ\gamma_R\), with \(\circ\) denoting concatenation, and the indicator function guarantees that the resulting path \(\bar\gamma\) indeed connects \(0\) to \(ns\). Observe that the subpaths are sampled \textit{independently} (apart from the overall constraint that \(\bar\gamma\) must join \(0\) to \(ns\)). Moreover, notice that the left-hand side has been multiplied by \(\sfe^{n\nu(s)}\), so that the error term resulting from this factorization procedure is really of order \(\sfe^{-n\nu(s)-cn}\). 

Theorem~\ref{thm:main_coupling_with_RW} also guarantees that \(\rmp, \rmp_L, \rmp_R\) have exponential tails, in the sense that they associate a probability exponentially decaying in \(\ell\) to paths of diameter at least \(\ell\).

Let \((S_k)_{k\geq 0}\) be the (directed) random walk on \(\Zd\) with transition probabilities given by the push-forward \(\rmp(y) = \sum_{\gamma:\,0\to y} \rmp(\gamma)\). Denote by \(\RWP{u}\) the distribution of \((S_k)\) conditionally on \(\{S_0=u\}\) and by \(\RWG(u,v) = \sum_{k\geq 0} \RWP{u}(S_k=v)\) the corresponding Green function.

Applying Theorem~\ref{thm:main_coupling_with_RW} with \(f\equiv 1\) allows one to express the asymptotic behavior of the Green function in terms of an average of the Green function \(\RWG(u,v)\): for any \(n\geq 0\),
\[
	\sfe^{n\nu(s)}\, G(0,ns) = (1+\sfO(\sfe^{-cn}))\, C_{LR} \sum_{u,v\in\Zd} \rmp_L(u) \rmp_R(v) \RWG(u,ns-v),
\]
where the probability measures \(\rmp_L, \rmp_R\) on \(\Zd\) are the push-forwards of the corresponding measures on paths and have exponential tails.
\begin{center}
	\includegraphics{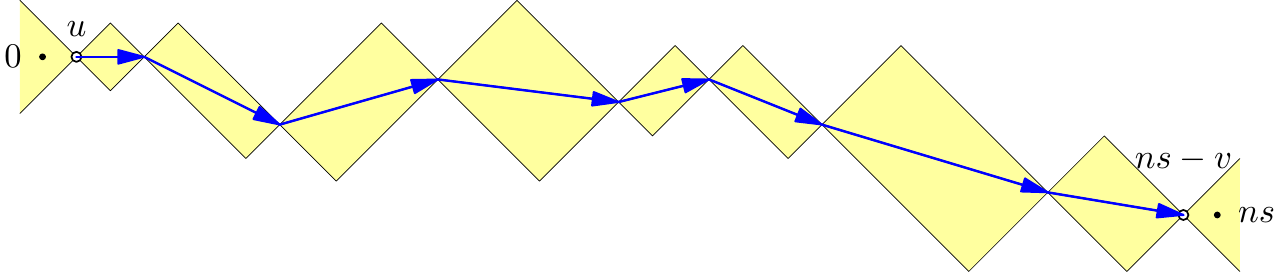}
\end{center} 

At this stage, the Ornstein--Zernike asymptotics for \(G(0,ns)\) follow from the asymptotic behavior of the Green function \(\RWG(u,ns-v)\). Rather straightforward computations yield the existence of \(c_s>0\) such that, for every \(\varepsilon>0\),
\[
	G(0,ns) = \frac{c_s}{n^{(d-1)/2}}\, \sfe^{-n\nu(s)}(1+\sfo(n^{\varepsilon-1/2})).
\]
The precise statements of these results are gathered as Theorem~\ref{thm:main_OZ_asymp}.
\begin{remark}
In the zero-mean case, the asymptotic expansion of \(\RWG(0,ns)\) in powers of \(n^{-1/2}\) have been established in~\cite{Uchiyama_1998}. In our case, the associated random walk has non-zero mean, and we were not able find a similar expansion of \(\RWG(0,ns)\), except for a first-order expansion of \(\RWG(0,ns)\) obtained in~\cite{Doney1966,Stam1969}. In Section~\ref{section:OZ_asymptotics}, we use standard large deviation bounds and a refinement of the local limit theorem to study \(\RWG(0,ns)\) to obtain the claim stated above. With additional work, one could obtain an expansion in powers of \(n^{-1/2}\) of \(\RWG(0,ns)\), and therefore of \(G(0,ns)\), of arbitrarily high order. In particular, in order to obtain arbitrarily high power-law corrections to the exponential decay for the two-point function of the Ising model, it is sufficient to derive the (much simpler) expansion for the Green function of the associated random walk. 
\end{remark}

Another interesting consequence of Theorem~\ref{thm:main_coupling_with_RW} is that typical paths \(\gamma\) contributing to \(G(0,ns)\) contain approximately \(A n\) steps (for some \(A=A(s)>0\)).
This behavior is in sharp contrast with what happens in the saturation regime \(\beta\in (0,\betasat(s))\) for the Potts models, for which it is proved in~\cite{Aoun+Ott+Velenik-2021-inPreparation} that typical paths have a length of order \(1\).
We complement this analysis with a local limit theorem establishing Gaussian fluctuations.
All these results are collected in Corollary~\ref{cor:NumberSteps}.

\medskip
Finally, as a last consequence of Theorem~\ref{thm:main_coupling_with_RW}, we prove that the inverse correlation length \(\nu\) depends analytically on the direction \(s\). We emphasize that such analyticity results are notoriously hard to prove nonperturbatively. The precise statement can be found in Theorem~\ref{thm:main_Wulff_analytic}.

\subsection{Illustration}
In this section, we present a brief illustration of the picture provided by this and our previous works, in the specific case of the two-point function \(\langle \sigma_0\sigma_{ne_1} \rangle_\beta\) for the Ising model with no external field at \(\beta<\betac\). Let us thus consider coupling constants of the form
\(J_i = \rho(i)^{-\alpha}\sfe^{-\rho(i)}\), with \(\rho\) a norm on \(\Rd\) and \(\alpha > 2d\). In this case, there exists \(\betasat(e_1)\in (0,\betac)\) such that
\begin{itemize}
	\item For all \(\beta \in (0,\betasat(e_1))\), \(\nu_\beta(e_1)=\rho(e_1)\) and there exists \(C_\beta>0\) such that, as \(n\to\infty\),
	\[
		\langle \sigma_0\sigma_{ne_1} \rangle_\beta = C_\beta n^{-\alpha}\sfe^{-n\rho(e_1)} (1+\sfo(1)).
	\]
	\item For all \(\beta \in (\betasat(e_1),\betac)\), \(0<\nu_\beta(e_1)<\rho(e_1)\) and there exists \(c_\beta>0\) such that, as \(n\to\infty\),
	\[
		\langle \sigma_0\sigma_{ne_1} \rangle_\beta = c_\beta n^{-(d-1)/2}\sfe^{-n\nu_\beta(e_1)} (1+\sfo(1)).
	\]
\end{itemize}
The first claim is proved in~\cite{Aoun+Ioffe+Ott+Velenik-CMP-2021,Aoun+Ott+Velenik-2021-inPreparation}, while the second is one of the results of the present paper.
In particular, the only case in which the sharp asymptotic is not yet known is \(\beta=\betasat(e_1)\). The same holds for coupling constants of the form \(J_i = \sfe^{-c'\rho(i)^{\eta}}\sfe^{-\rho(i)}\) for arbitrary \(c'>0\) and \(\eta \in (0,1)\).

%%%%%%%%%%%%%%%%%%%%%%%%%%%%%%%%%%%%%%%%%%%%%%%%%%%%%%%%%%%%%%%%%%%%%%%%%%%%%%%%%%%%%%%%%%%%%%%%%%%%%%%%
\section{Definitions and conventions}
\label{sec:definitions}
%%%%%%%%%%%%%%%%%%%%%%%%%%%%%%%%%%%%%%%%%%%%%%%%%%%%%%%%%%%%%%%%%%%%%%%%%%%%%%%%%%%%%%%%%%%%%%%%%%%%%%%%

In this section, we introduce the definitions and notations that will be used throughout the paper. We encourage the reader to refer to figures whenever possible. We also give our assumptions on the weights \(q(\gamma)\) and explain how the two-point function of the Ising model can be rewritten as sum over paths \(\gamma\) with some weights \(q\).

\subsection{General}

We will work on \(\Zd\) canonically embedded in \(\Rd\). We will consider the following graph structure: vertices are sites of \(\Zd\) while the edge set, denoted \(\bbE_d\), will be induced by a symmetric weight function \(J:\Zd\times\Zd\to\R_{\geq 0}\): the edge \(\{i, j\}\) is in \(\bbE_d\) if \(J_{ij}>0\).

For \(A\subset\Zd\), we denote
\begin{gather*}
	E_A = \bsetof{\{i,j\}\in \bbE_d}{i, j\in A},
\end{gather*}

A \emph{path} \(\gamma = (\gamma_0, \dots, \gamma_n)\) is a sequence of vertices such that \(\{\gamma_{k}, \gamma_{k+1}\}\in\bbE_d\) for \(k\in\{0, \dots, n-1\}\). \(n\) is the \emph{length} of \(\gamma\), \(|\gamma| =n\). Equivalently, a path is given by a starting point and a sequence of edges, \(\gamma=(\gamma_0; e_1, \dots, e_n)\), satisfying \(\gamma_0\in e_1\) and \(e_k\cap e_{k+1}\neq \emptyset\) for \(1\leq k < n\). We write \(\gamma: x\to y\) if \(\gamma_0=x\) and \(\gamma_{|\gamma|} = y\). For \(0\leq a\leq b\leq |\gamma|\), let \(\gamma_a^b = (\gamma_a, \gamma_{a+1}, \dots, \gamma_b)\). The path \(\gamma_a^b\) is a path in its own right, that is, the ``time'' labels are not preserved: \((\gamma_a^b)_k = \gamma_{a+k}\) for \(0\leq k\leq|\gamma|\). The set of paths from \(x\) to \(y\) that use only edges in \(E\subset \bbE_d\) is denoted \(\pathSet_{E}(x,y)\). We set \(\pathSet(x,y)\equiv \pathSet_{\bbE_d}(x,y)\), \(\pathSet(x) = \bigcup_{y}\pathSet(x,y)\) and \(\pathSet= \bigcup_{x}\pathSet(x)\). Summation over \(\gamma\in\pathSet(x,y)\) will be denoted \(\sum_{\gamma:\, x\to y}\). We shall also use the notations
\[
	\fend(\gamma) = \gamma_0,\quad \bend(\gamma) = \gamma_{|\gamma|}.
\]
We say that a path is \emph{edge-self-avoiding} if \(e_k=e_l \implies k=l\).

For two paths \(\gamma: x\to y\), and \(\gamma': y\to z\), we denote by \(\gamma\concatenate\gamma'\) their concatenation: \(\gamma\concatenate \gamma' = (\gamma_0, \dots, \gamma_{|\gamma|}, \gamma'_1, \dots, \gamma'_{|\gamma'|} )\).

To an edge \(e=\{i,j\}\in\bbE_d\), we associate the closed line segment \([e] = \setof{t i + {(1-t) j}}{t\in[0, 1]}\). To a subgraph \((V,E)\subset (\Zd,\bbE_d)\), we associate the subset of \(\Rd\) \([(V,E)] = [E] = \bigcup_{e\in E} [e]\).

We use \(\norm{\cdot}\) to denote the Euclidean norm on \(\Rd\). \(\bbS^{d-1}\) denotes the unit sphere in \(\Rd\): \(\bbS^{d-1} = \setof{x\in\Rd}{\norm{x} = 1}\). We also denote by \(\bbB_r(x)\) the closed ball of radius \(r\) centered at \(x\). For \(x\in\Rd\setminus\{0\}\), we will write \(\hat{x} = x/\norm{x}\). We also define the slabs
\[
	\slab_{a,b}(y) = \setof{x\in\Rd}{a\leq x\cdot y < b},
\]
with \(a\leq b\in\R, y\in\Rd\). We write \(\slab_{b}(y)\equiv \slab_{0,b}(y)\) for \(b\geq 0\) and \(\slab_{a}(y)\equiv \slab_{a,0}(y)\) for \(a\leq 0\).

\subsection{Discrete and continuous}

We will often define sets as subsets of \(\Rd\) as it is sometimes convenient. To lighten notations, when \(A\subset\Rd\), we write \(\sum_{x\in A}\) to mean \(\sum_{x\in A\cap\Zd}\). In the same spirit, integer parts are systematically omitted. We hope that these few liberties will not confuse the reader.

\subsection{Constants}

\(c, C, c', C', \dots\) will denote constants whose value can change from line to line. Unless explicitly stated otherwise, they depend only on \(d, \beta, J\).

\subsection{Little \(\sfo\) notation}

We will write \(f(a) = \sfo_a(1)\) when \(\lim_{a\to\infty} f(a) = 0\). We will also use the more standard \(g=\sfo(f(a))\) if \(\lim_{a\to\infty} \frac{g(a)}{f(a)}=0\). Note that \(\sfo(\cdot)\) can be positive or negative. Finally, we say that \(g=\sfO(f(a))\) if there exists \(C\geq 0, a_0\geq 0\) such that, for any \(a\geq a_0\), \(\abs{g(a)} \leq C \abs{f(a)}\).

\subsection{Interaction}
\label{subsec:interaction}

The interactions (edge weights) \(J_{ij}\) are non-negative coupling constants. We shall assume
\begin{itemize}
	\item Ferromagnetism: \(J_{ij}\geq 0\).
	\item Central symmetry and translation invariance: \(J_{ij} = J_{i-j} = J_{j-i}\).
	\item Normalization: \(J_0=0\), \(\sum_{x\in\Zd} J_x = 1\).
	\item Exponential decay: \(J_{x}\leq \sfe^{-c\norm{x}}\) for some \(c>0\).
	\item Irreducibility: for any \(x\in\Zd\), there exists \(\gamma: 0\to x\) with \(\prod_{k=1}^{|\gamma|} J_{\gamma_{k-1}, \gamma_k} >0\).
\end{itemize}

\subsection{General path models to be considered}
\label{ssec:ConditionsOnPaths}

While our main goal is to prove OZ-type asymptotics for the Ising model, the proof also applies to a class of self-repulsive path models including the self-avoiding walk, the killed random walk and, more generally, the large class of self-repulsive polymer models studied in~\cite{Ioffe+Velenik-2008}. We define here the relevant family of path models.

\begin{definition}[Path models]
	\label{def:path_models}
	A translation-invariant \emph{path model} is a weight function \(q: \pathSet \to \R_{\geq 0}\) such that \(q(\gamma) = q(x+\gamma)\) for any \(\gamma\in\pathSet\) and \(x\in\Zd\).
\end{definition}
In the sequel, we will always assume translation invariance and not mention it systematically.

We say that the paths \(\gamma_1, \dots, \gamma_n\) are \emph{compatible} if the path \(\gamma_1\concatenate\cdots\concatenate\gamma_n\) has positive weight: \(q(\gamma_1\concatenate\cdots\concatenate\gamma_n) > 0\). Moreover, we say that path \(\gamma\) is \emph{admissible} if \(q(\gamma)>0\).
\begin{definition}[Conditional weights]
	\label{def:conditional_weight}
	Given compatible paths \(\gamma\) and \(\gamma'\), we define the \emph{conditional weight} of \(\gamma'\) given \(\gamma\) by
	\[
		q(\gamma'\given \gamma) = \frac{q(\gamma\concatenate \gamma')}{q(\gamma)}.
	\]
	It will sometimes be notationally convenient to allow \(\gamma\) to be empty; in this case we define \(q(\gamma'\given\gamma) = q(\gamma')\).
\end{definition}
The central quantity in our study is the two-point function associated to some of these models.

\begin{definition}[Two-point function]
	The \emph{two-point function} of the path model is
	\[
		G(x,y) = \sum_{\gamma:\, x\to y} q(\gamma).
	\]
\end{definition}

\begin{definition}[Inverse correlation length]
	\label{def:ICL}
	When the limit exists, the \emph{inverse correlation length} of the path model in direction \(s\in\bbS^{d-1}\) is defined by
	\begin{equation}
		\nu(s) = -\lim_{n\to\infty} \frac{1}{n}\log G(0,ns).
	\end{equation}
\end{definition}

When discussing the Ising model, we shall write \(\nu_\beta\) for the inverse correlation length, in order to make its temperature dependence explicit.

\medskip
We will work with path models satisfying the following conditions:
\begin{enumerate}[label=\textbf{P\arabic*}]
	\item\label{weight_property:ICL} Existence and positivity of \(\nu\): the limit in Definition~\ref{def:ICL} exists and is positive for all \(s\in\bbS^{d-1}\); moreover, \(\nu\) extends to a norm on \(\Rd\) by positive homogeneity of order one.
	\item\label{weight_property:repulsion} Self-repulsion: for any admissible \(\gamma': z\to x\),
	\[
		\sum_{\gamma:\, x\to y} q(\gamma\given \gamma')\leq \sum_{\gamma:\, x\to y} q(\gamma).
	\]
	\item\label{weight_property:lower_bound} Lower bound: there exists a constant \(C=C(\beta)>0\) such that, for every compatible paths \(\gamma'\) and \(\gamma\) (\(\gamma'\) can be empty),
	\[
		q(\gamma\given\gamma') \geq \prod_{k=1}^{|\gamma|}CJ_{\gamma_{k}-\gamma_{k-1}}.
	\]
	\item\label{weight_property:weight_growth} Controlled growth: there exists \(C=C(\beta)\) such that, for any compatible paths \(\gamma', \gamma\),
	\[
		q(\gamma\given \gamma')\leq \prod_{k=1}^{|\gamma|}CJ_{\gamma_{k}-\gamma_{k-1}}.
	\]
	\item\label{weight_property:finite_energy} Finite energy: Let \(R>0\), \(a>0\) and \(\hat t\in\bbS^{d-1}\). There exists \(C=C(\beta,R,a)\) such that, for any compatible paths \(\gamma\) and \(\gamma'\) satisfying \(\fend(\gamma) = \bend(\gamma') = 0\), \(\gamma' \subset \setof{z\in\Zd}{z\cdot \hat{t} \leq 0}\) and \(\gamma \subset \setof{z\in\Zd}{\norm{z}\leq R \text{ or } z\cdot\hat{t} > a\norm{z}-R}\),
	\[
		q(\gamma\given\gamma') \leq C q(\gamma)
		\quad\text{ and }\quad
		q(\gamma'\given\gamma) \leq C q(\gamma').		
	\]
	\item\label{weight_property:monotonicity} Monotonicity of conditional weights: for any compatible paths \(\gamma'\), \(\gamma''\) and \(\gamma\) such that \(\gamma\cap(\gamma'\concatenate\gamma'')=\bend(\gamma)\) (we allow \(\gamma''\) to be empty),
	\[
		q(\gamma\given \gamma'\concatenate\gamma'') \geq q(\gamma\given \gamma'').
	\]
	\item\label{weight_property:mixing} Exponential ratio mixing: there exist \(C, c_{\mix}>0\) such that, for any compatible paths \(\gamma', \gamma'', \gamma\),
	\[
		1 - C\sum_{i\in \gamma, j\in\gamma''} \sfe^{-c_{\mix}\norm{i-j}}
		\leq
		\frac{q(\gamma\given \gamma'\concatenate \gamma'')}{q(\gamma\given \gamma'')}
		\leq
		1 + C\sum_{i\in \gamma, j\in\gamma''} \sfe^{-c_{\mix}\norm{i-j}}.
	\]
\end{enumerate}

We briefly discuss these assumptions. \ref{weight_property:ICL} is classical. It holds in the whole high-temperature regime of various models. By an adaptation of the arguments in~\cite{Ott-2020b}, it follows from the exponential decay of \(G(x,y)\) and the other properties. Note that it implies that
\[
	G(0,x) = \sfe^{-\nu(x)(1+\sfo_{\norm{x}}(1))}.
\]
Properties~\ref{weight_property:repulsion} and~\ref{weight_property:monotonicity} are the least robust. They hold, for instance, for the self-repulsive polymer models of~\cite{Ioffe+Velenik-2008} (which include the self-avoiding walk, the Domb--Joyce model and the killed random walk) and for the high-temperature representation of the Ising model.
Note that the hypotheses can be weakened by asking \(\sum_{\gamma:\, x\to y} q(\gamma\given \gamma')\leq C\sum_{\gamma:\, x\to y} q(\gamma)\) for some \(C\geq 1\), since one can then work with the weight \(\tilde{q} = Cq\) which satisfies the inequality with \(C=1\).
Properties~\ref{weight_property:lower_bound}, \ref{weight_property:weight_growth} and~\ref{weight_property:finite_energy} enable our local surgery arguments and prevent some pathological behaviors. Property~\ref{weight_property:mixing} is a high-temperature assumption. It is often a consequence of the exponential decay of \(G\).

We can now define the class of path measures that are considered in the present work.
\begin{definition}
	\label{def:path_model_set}
	We define \(\calQ\) to be the set of translation-invariant path models such that Properties~\ref{weight_property:ICL}, \ref{weight_property:repulsion}, \ref{weight_property:lower_bound}, \ref{weight_property:weight_growth}, \ref{weight_property:finite_energy}, \ref{weight_property:monotonicity} and \ref{weight_property:mixing} are satisfied.
\end{definition}

\subsection{Ising model}

We consider the ferromagnetic Ising model on \(\Zd\) with two-body interaction. Namely, we consider the finite-volume measures on \(\{-1, 1\}^{\Lambda}\)
\[
	\mu_{\Lambda;\beta}^{\eta}(\sigma) \propto \sfe^{\beta \sum_{\{i,j\}\subset \Lambda} J_{ij} \sigma_i\sigma_j + \beta \sum_{i\in\Lambda,j\notin\Lambda} J_{ij}\sigma_i \eta_j}
\]
where \(\beta\geq 0\) is the inverse temperature and \(\eta\) is a boundary condition. Expectation under \(\mu_{\Lambda;\beta}^{\eta}\) is denoted \(\lrangle{\cdot}_{\Lambda;\beta}^{\eta}\).

%It is well known that the model undergoes a phase transition: there exists \(\betac=\betac(J,d)\) such that
%\begin{itemize}
%	\item \(\lim_{\Lambda\to\Zd} \lrangle{\sigma_0}^1_{\Lambda;\beta} >0\) for \(\beta >\betac\),
%	\item \(\lim_{\Lambda\to\Zd} \lrangle{\sigma_0}^1_{\Lambda;\beta} =0\) for \(\beta <\betac\).
%\end{itemize}
%In the second case, one also has a unique Gibbs measure which can be obtained as the limit \(\lim_{\Lambda\to\Zd} \mu_{\Lambda;\beta}^{0}\). Moreover, one has the following stronger statement: for \(\beta<\betac\), there exist \(C<\infty, c>0\) such that
%\[
%	\lrangle{\sigma_{0}}_{\Lambda_N;\beta}^1\leq C\sfe^{-cN}.
%\]

It is well known that this model undergoes a phase transition. Namely, let us denote by \(\mu_{\beta}^{0}\) the weak limit of \(\mu_{\Lambda;\beta}^{0}\) as \(\Lambda\to\Zd\) and let us consider the 2-point function \(G_\beta(x) = \lrangle{\sigma_0\sigma_x}^0_{\beta}\). Then, there exists \(\betac=\betac(J,d) \in (0, \infty]\) such that
\begin{itemize}
	\item there exists \(c=c(\beta)>0\) such that \(G_\beta(x) \leq \sfe^{-c\norm{x}}\) for all \(\beta < \betac\),
	\item \(\inf_{x\in\Zd} G_\beta(x) > 0\) for all \(\beta > \betac\).
\end{itemize}
In the first case, \(\mu_{\beta}^{0}\) is the only infinite-volume Gibbs measure.

\subsection{Path representation of correlations}

We briefly discuss how the Ising model falls into the framework of path models defined above. Everything in this section is standard~\cite{Pfister+Velenik-1999}. Let \(\Lambda\subset\Zd\) be finite. The \emph{high-temperature representation} of the correlation function \(\langle\sigma_A\rangle_{\Lambda;\beta}^0\), \(A\Subset\Lambda\), of the Ising model is based on the identity
\[
	\sum_{\sigma\in\{-1,1\}^{\Lambda}}\sigma_A\, \sfe^{\beta \sum_{\{i,j\}\subset \Lambda} J_{ij} \sigma_i\sigma_j } = \prod_{\{i,j\}\in E_{\Lambda}} \cosh(\beta J_{ij})\sum_{\substack{\omega\subset E_{\Lambda}\\ \partial \omega = A}} \prod_{\{i,j\}\in\omega} \tanh(\beta J_{ij}),
\]
where \(\partial\omega\) is the set of vertices having odd degree in \((\Zd,\omega)\).

Fix some total order on \(\Zd\).
Given \(\omega\subset \bbE_d\) with \(\partial\omega = \{x,y\}\), one can extract an edge-self-avoiding path \(\Gamma(\omega): x\to y\) using the following algorithm:
\begin{algorithm}
	\caption{Path extraction}
	\KwInput{\(\omega\)}
	Set \(F = \omega\), \(\gamma_0 = x\), \(\bar{\gamma} = \emptyset\), \(n = 0\);\\
	\While{\(\gamma_n \neq y\)}{
		Let \(i\) be the smallest vertex such that \(\{\gamma_n, i\}\in F\);\\
		Let \(A = \bsetof{\{\gamma_n,j\}\in\bbE_d}{j\leq i}\);\\
		Set \(\gamma_{n+1} = i\);\\
		Update \(F = F\setminus A\), \(\bar{\gamma} = \bar{\gamma}\cup A\), \(n = n+1\);\\
	}
	\KwOutput{\(\gamma\), \(n\), \(\bar\gamma\)}
	\label{algo:path_extract}
\end{algorithm}

\noindent
The algorithm yields \(\Gamma(\omega) = \gamma\in \pathSet(x,y)\), \(n = |\gamma|\) and \(\bar{\gamma}\) is the set of edges one has to check whether \(\Gamma(\omega) = \gamma\) or not (note that \(\bar{\gamma}\) depends on \(\gamma\) only, not on \(\omega\)).

The \emph{path representation} of the two-point function is given by
\begin{equation*}
	\lrangle{\sigma_x\sigma_y}_{\Lambda}^0
	=
	\sum_{\gamma\in\pathSet(x,y)} q_{E_{\Lambda}}(\gamma),
\end{equation*}
where, for any finite \(E\subset\bbE_d\),
\begin{equation*}
		q_{E}(\gamma) = \frac{\sum_{\substack{\omega \subset E:\, \partial\omega= \{x,y\}}} \mathds{1}_{\{\Gamma(\omega) = \gamma\}} \prod_{\{i, j\}\in\omega} \tanh(\beta J_{ij})}{\sum_{\substack{\omega \subset E:\, \partial\omega=\emptyset}} \prod_{\{i, j\}\in\omega} \tanh(\beta J_{ij})}.
\end{equation*}
For all \(\beta<\betac\), we can then define the infinite-volume weights by \(q(\gamma) = \lim_{E\to\bbE_d} q_E(\gamma)\). They satisfy
\[
	\lrangle{\sigma_x\sigma_y} = \sum_{\gamma\in\pathSet(x,y)} q(\gamma).
\]

\begin{lemma}\label{lem:IsingQ}
	Suppose \(\beta<\betac\). Then, \(q\in \calQ\).
\end{lemma}
The proof of this lemma, which is an adaptation of previous results, is sketched in Appendix~\ref{sec:IsingHTProperties}.

\subsection{Convexity, cones and cone-points}

\begin{definition}[Unit ball, Wulff shape]
	For a norm \(\rho\) on \(\Rd\), we denote by
	\begin{gather*}
		\calU_{\rho} = \setof{x\in\Rd}{\rho(x)\leq 1},\quad \calW_{\rho} = \bigcap_{s\in\bbS^{d-1}}\setof{x\in\Rd}{x\cdot s\leq \rho(s)},
	\end{gather*}
	the corresponding \emph{unit ball} and \emph{Wulff shape} (or \emph{polar set}). When \(\rho\) is omitted, it is set to be \(\nu\): \(\calU \equiv \calU_{\nu}\) and \(\calW \equiv \calW_{\nu}\).
\end{definition}

\begin{definition}[Extremal radius]
	\label{def:extremal_radius}
	For \(\rho\) a norm, denote by
	\[
		r_*(\rho) = \max_{u\in\calU_{\rho}} \norm{u}.
	\]
	When \(\rho\) is omitted, it is set to \(\nu\): \(r_* = r_*(\nu)\).
\end{definition}

\begin{definition}[Dual vectors; see Fig.~\ref{fig:UnitBallWulff}]
	Let \(\rho\) be a norm. Let \(s\in\bbS^{d-1}\). We say that \(t\) and \(s\) are \emph{\(\rho\)-dual} if \(t\in\partial\calW_{\rho}\) and \(t\cdot s = \rho(s)\). We denote by \(T_{s}^{\rho}\) the set of all vectors \(\rho\)-dual to \(s\). In the same line of ideas, we denote by \(S_{t}^{\rho}\subset \bbS^{d-1}\) the set of directions \(s\) such that \(t\) and \(s\) are \(\rho\)-dual. When \(\rho\) is omitted, it is set to be \(\nu\).
\end{definition}
\begin{figure}[ht]
	\includegraphics{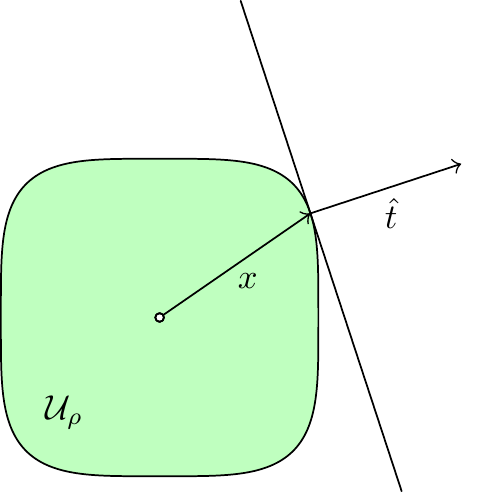}
	\hspace*{1cm}
	\includegraphics{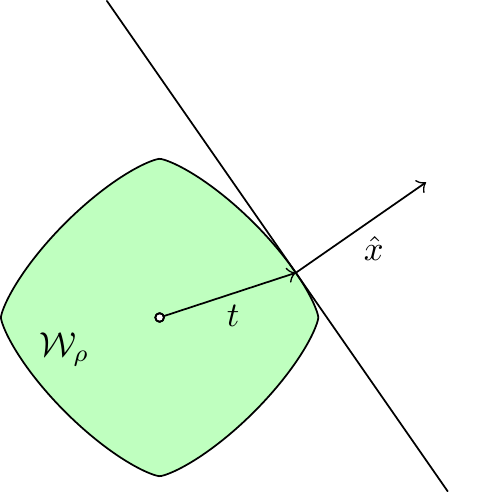}
	\caption{The unit ball and the Wulff shape associated to the \(L^4\)-norm \(\rho(x) = \norm{x}_4\). A pair of \(\rho\)-dual vectors \(x\) and \(t\) are also represented.
	We have used the notation \(\hat x = x/\norm{x}\) and \(\hat t = t/\norm{t}\).}
	\label{fig:UnitBallWulff}
\end{figure}
\begin{definition}[Cones; see Fig.~\ref{fig:ForwardCone}]
	\label{def:cones}
	Let \(\rho\) be a norm. Let \(t\in\partial\calW_{\rho}\) and \(1>\delta>0\). We define the cones
	\[
		\fcone_{\rho,t,\delta} = \setof{x\in\Rd}{t\cdot x \geq (1-\delta)\rho(x)}.
	\]
	We also define
	\[
		\fcone_{\rho,t,\delta}(x) = x+\fcone_{\rho,t,\delta}, \quad \bcone_{\rho,t,\delta}(x) = x-\fcone_{\rho,t,\delta}.
	\]
	When \(\rho\) is omitted, it is set to be \(\nu\):
	\[
		\fcone_{t,\delta} \equiv \fcone_{\nu,t,\delta}.
	\]
\end{definition}
\begin{figure}[ht]
	\includegraphics{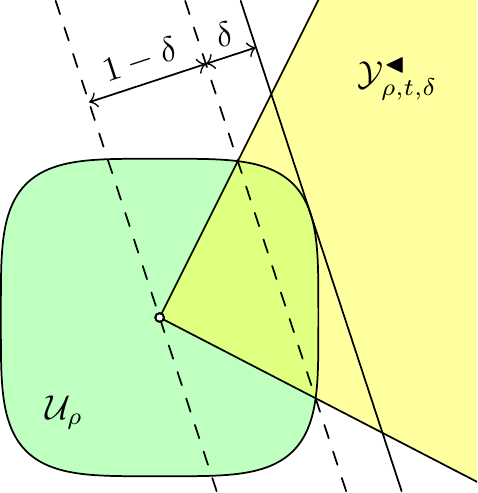}
	\caption{The forward cone \(\fcone_{\rho,t,\delta}\) for the same \(t\) as in Fig.~\ref{fig:UnitBallWulff}.}
	\label{fig:ForwardCone}
\end{figure}
\begin{definition}[Cone-points of paths; see Fig.~\ref{fig:ConePoint}]
	\label{def:cone_points}
%	Let \(A\subset \Rd\) be connected. \(x\in A\cap \Zd\) is a \((t,\delta)\)\emph{-cone-point} of \(A\) if:
%	\[
%		A\subset \fcone_{t,\delta}(x)\cup \bcone_{t,\delta}(x).
%	\]Denote \(\CPts_{t,\delta}(A)\) the set of \((t,\delta)\)-cone-points of \(A\).
%	For \(F\subset \bbE_d\) such that \((V_F,F)\) is connected, denote \(\CPts_{t,\delta}(F) \equiv \CPts_{t,\delta}([F]) \).
	Let \(\gamma\in\pathSet\). \(v\in\gamma\) is a \((t,\delta)\)\emph{-cone-point} of \(\gamma\) if there exists \(0\leq k\leq |\gamma|\) such that \(\gamma_k=v\), \(k'<k \Rightarrow \gamma_{k'}\in\bcone_{t,\delta}(v)\) and \(k'>k \Rightarrow \gamma_{k'}\in\fcone_{t,\delta}(v)\). Denote \(\CPts_{t,\delta}(\gamma)\) the set of \((t,\delta)\)-cone-points of \(\gamma\).
\end{definition}
\begin{figure}[ht]
	\includegraphics{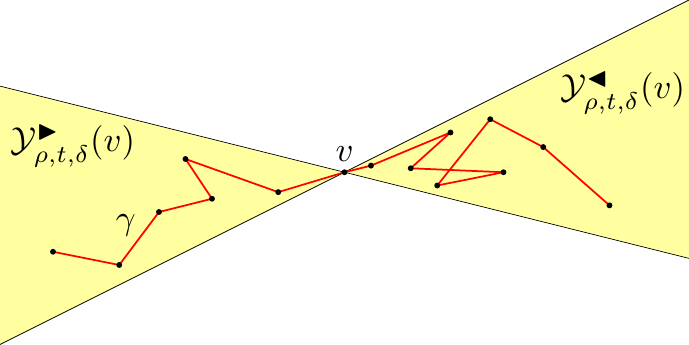}
	\caption{\(v\) is a \((t,\delta)\)-cone-point of the path \(\gamma\).}
	\label{fig:ConePoint}
\end{figure}

\begin{definition}[Diamonds; see Fig.~\ref{fig:Diamond}]
	\label{def:diamonds}
	The \emph{diamond} associated to \(x\in\Rd\) and \(y\in\fcone_{t,\delta}(x)\) is defined by
	\[
		\diam_{t,\delta}(x,y) = \fcone_{t,\delta}(x)\cap \bcone_{t,\delta}(y).
	\]
\end{definition}
\begin{figure}[ht]
	\includegraphics{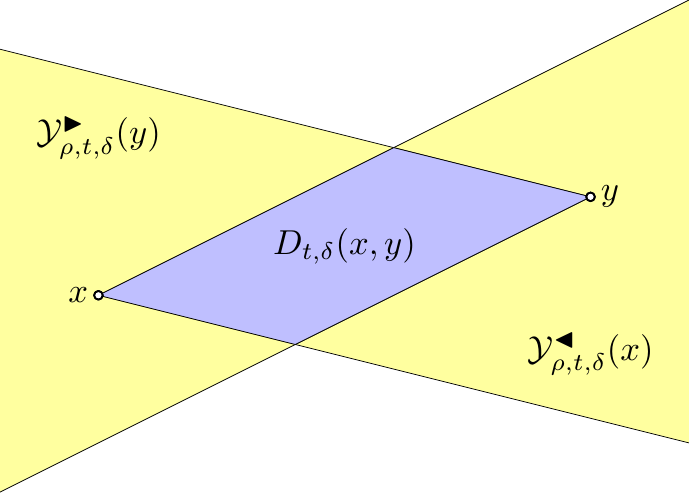}
	\caption{The Diamond \(D_{t,\delta}(x,y)\).}
	\label{fig:Diamond}	
\end{figure}
\begin{definition}[Diamond-contained paths; see Fig.~\ref{fig:DiamondContained}]
	\label{def:contained_paths}
	A path \(\gamma\in\pathSet\) is
	\begin{itemize}
		\item \((t,\delta)\)-\emph{forward-contained} if \(\gamma\subset \fcone_{t,\delta}(\fend(\gamma))\);
		\item \((t,\delta)\)-\emph{backward-contained} if \(\gamma\subset \bcone_{t,\delta}(\bend(\gamma))\);
		\item \((t,\delta)\)-\emph{diamond-contained} if it is forward- and backward-contained.
	\end{itemize}
\end{definition}
\begin{figure}[ht]
	\includegraphics{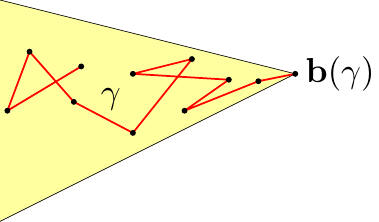}
	\hspace{-6mm}
	\includegraphics{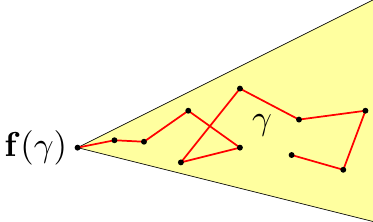}
	\hspace{2mm}
	\includegraphics{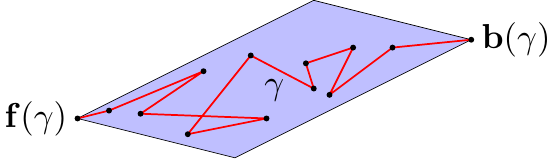}
	\caption{Backward-contained, forward-contained and diamond-contained paths.}
	\label{fig:DiamondContained}
\end{figure}

We will omit the dependency on \((t,\delta)\) when it is clear from the context.

\pagebreak % <<<<<<<<<<<<<<<<<<<<<<<<<<<<<<<<<<< PAGE BREAK

\begin{definition}[Irreducible paths]
	\label{def:irreducible_paths}
	A path \(\gamma\in\pathSet\) is 
	\begin{itemize}
		\item \((t,\delta)\)-\emph{irreducible} if it is diamond-contained and it does not contain \((t,\delta)\)-cone-points other than \(\fend(\gamma),\bend(\gamma)\);
		\item \((t,\delta)\)-\emph{forward irreducible} if it is forward-contained and \(\gamma\cap \diam(\fend(\gamma), \bend(\gamma))\) does not contain \((t,\delta)\)-cone-points except for \(\fend(\gamma)\) and, possibly, \(\bend(\gamma)\);
		\item \((t,\delta)\)-\emph{backward irreducible} if it is backward-contained and \(\gamma \cap \diam(\fend(\gamma),\bend(\gamma))\) does not contain \((t,\delta)\)-cone-points except for \(\bend(\gamma)\) and, possibly, \(\fend(\gamma)\).
	\end{itemize}
\end{definition}

We will use the notations (see Fig.~\ref{fig:SetsOfPathsDisplacement})
\begin{gather*}
	\SetRootDiaCont(t,\delta) = \setof{\gamma}{\gamma\text{ is }(t,\delta)\text{-diamond-contained},\ \gamma_0 = 0},\\
	\SetRootMarkBackCont(t,\delta) = \setof{\gamma}{\gamma\text{ is }(t,\delta)\text{-backward-contained},\  \gamma_0 = 0},\\
	\SetRootMarkForwCont(t,\delta) = \setof{\gamma}{\gamma\text{ is }(t,\delta)\text{-forward-contained},\ \gamma_0 = 0}.
\end{gather*}
\begin{figure}[ht]
	\includegraphics{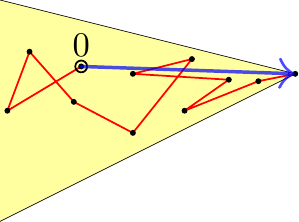}
	\hspace{0mm}
	\includegraphics{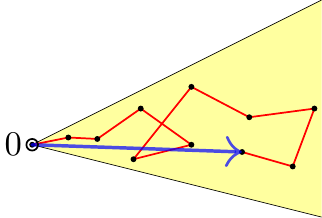}
	\hspace{10mm}
	\includegraphics{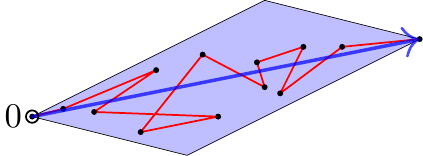}	
	\caption{The paths in \(\SetRootMarkBackCont(t,\delta)\), \(\SetRootMarkForwCont(t,\delta)\) and \(\SetRootDiaCont(t,\delta)\) corresponding to the paths in Fig.~\ref{fig:DiamondContained}. The displacement \(\displace(\gamma)\) (in blue) and the origin (circled dot) are also indicated in each case.}
	\label{fig:SetsOfPathsDisplacement}
\end{figure}

The same sets with the superscript \(\irreducible\) are the respective restrictions of the set to irreducible paths. The set \(\SetRootDiaCont\) can be seen as a subset of \(\SetRootMarkForwCont\). We will use the following concatenation rule: for \(\gamma\in\SetRootMarkBackCont\) and \(\gamma'\in\SetRootMarkForwCont\),
\[
	\gamma\concatenate\gamma' \equiv \gamma\concatenate (\fend(\gamma)+\gamma').
\]
Higher number of concatenations are treated from right to left.

%For a fixed given pattern, we shall use the same notations for the pattern versions of the objects.

\pagebreak % <<<<<<<<<<<<<<<<<<<<<<<<<<<<<<<<<<< PAGE BREAK

Finally, we define the displacement of a path.
\begin{definition}[Displacement; see Fig.~\ref{fig:Displacement} and~\ref{fig:SetsOfPathsDisplacement}]
	\label{def:displacement}
	The \emph{displacement} of a path \(\gamma\) is defined by \(\displace(\gamma) = \gamma_{|\gamma|} - \gamma_0\). For diamond-contained paths, the displacement is a function of the diamond.
\end{definition}
\begin{figure}[ht]
	\includegraphics{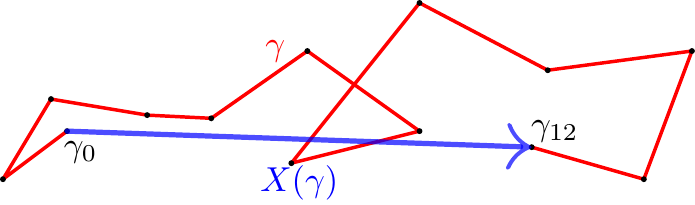}
	\caption{The displacement \(X(\gamma)\) of the path \(\gamma\).}
	\label{fig:Displacement}
\end{figure}

\subsection{Skeleton}

We finish this section by defining a family of objects that will be the output of our coarse-graining procedure.

\begin{definition}[Cells]
	\label{def:cells}
	Let \(K\geq 0\). Let \(\rho\) be a norm on \(\Rd\). We define the cells
	\[
		\Delta_{\rho,K} = K\calU_{\rho}\cap \Zd,\quad \Delta_{\rho,K}(x) = x+ \Delta_{\rho,K}.
	\]
	When \(\rho\) is omitted, it is set to be \(\nu\): \(\Delta_{K}\equiv \Delta_{\nu,K}\).
\end{definition}

\begin{definition}[Skeleton]
	\label{def:skeleton}
	Let \(K,\rho\) be as in Definition~\ref{def:cells}. Let \(R\geq 0\). A \emph{\((\rho,K)\)-skeleton rooted at \(v_0\)} is a sequence \((v_0, \dots, v_n) \subset \Zd\), satisfying \(v_{k+1}\notin \Delta_{\rho,K}(v_k)\) for \(k=0, \dots, n-1\). The \emph{length} of a skeleton \(V=(v_0,\dots, v_n)\) is \(|V| = n\).
	
	A pair (step) \((v_k,v_{k+1})\) is called \((\rho,K,R)\)-\emph{short} if \(v_{k+1}\in \Delta_{\rho,K+R}(v_k)\). Otherwise, it is \((\rho,K,R)\)-\emph{long}.
	We write
	\begin{gather*}
		\shortEdge(v_0, \dots, v_n) = \bsetof{ (v_{k},v_{k+1})}{0\leq k < n, v_{k+1}\in \Delta_{\rho,K+R}(v_k)},\\
		\longEdge(v_0, \dots, v_n) = \bsetof{ (v_{k},v_{k+1})}{0\leq k < n, v_{k+1}\notin \Delta_{\rho,K+R}(v_k)}.
	\end{gather*}
	When \(\rho\) is omitted, it is set to \(\nu\).
\end{definition}
In the course of the proof, we will work at fixed \(K,R,\rho\) and leave them implicit in the notation.

\begin{definition}[Skeleton Sets]
	We denote by \(\SkeletonSet_{\rho,M,I,(r_i)_{i\in I}}(x)\), \(I\subset \{1, \dots, M\}\), \(r_i\in\bbZ_{>0}\), the set of skeletons with \(v_0=x\), \(M+1\) vertices (length \(M\)), whose set of long steps is exactly \(I\) and such that \(v_{i}\in\Delta_{\rho,K+R+r_i}(v_{i-1})\setminus \Delta_{\rho,K+R+r_i-1}(v_{i-1})\) for each \(i\in I\).
	
	We write \(\SkeletonSet_{\rho}(x) = \bigcup_{M\geq 1} \bigcup_{I\subset \{1,\dots, M\}} \bigcup_{(r_i)_{i\in I}} \SkeletonSet_{M,I,(r_i)_{i\in I}}(x)\).
	When \(\rho\) is omitted, it is set to \(\nu\).
\end{definition}

\begin{definition}[Cone-points of skeletons]
	Let \(V=(v_0, \dots, v_M)\) be a skeleton. \(v\) is a \emph{\((t,\delta)\)-cone-point} of \(V\) if there is \(0\leq k\leq M\) such that \(v=v_k\in V\), and \(k'>k \Rightarrow v_{k'}\in\fcone_{t,\delta}(v)\) and \(k'<k \Rightarrow v_{k'}\in\bcone_{t,\delta}(v)\). Denote \(\CPts_{t,\delta}(V)\) the set of \((t,\delta)\)-cone-points of \(V\).
\end{definition}

%%%%%%%%%%%%%%%%%%%%%%%%%%%%%%%%%%%%%%%%%%%%%%%%%%%%%%%%%%%%%%%%%%%%%%%%%%%%%%%%%%%%%%%%%%%%%%%%%%%%%%%%
\section{Main theorems and assumptions}
\label{sec:Saturation}
%%%%%%%%%%%%%%%%%%%%%%%%%%%%%%%%%%%%%%%%%%%%%%%%%%%%%%%%%%%%%%%%%%%%%%%%%%%%%%%%%%%%%%%%%%%%%%%%%%%%%%%%

\subsection{Saturation phenomenon, conditions on the interaction}\label{sec:SaturationConditions}

\subsubsection{The saturation phenomenon}
The assumption that \(J\) decays exponentially is not sufficient to guarantee OZ decay. Indeed, it was observed (and proved) in~\cite{Aoun+Ioffe+Ott+Velenik-CMP-2021, Aoun+Ioffe+Ott+Velenik-PRE-2021} that a key assumption (the existence of a \emph{mass gap}) to have such asymptotics was not fulfilled in certain situations, and that OZ asymptotics were not valid in that case.
In this section, we present a condition (NSA) ensuring that this problem does not occur. Let us however emphasize that there are several situations in which the assumption is trivially satisfied (see Lemma~\ref{lem:NSA}) and we conjecture that it always holds. The reader who is not interested in the technical details or willing to restrict to these cases can skip Section~\ref{sec:SaturationConditions} (and replace, in our statements, the requirement that NSA holds by the simpler condition that the direction \(s\) is not saturated).

To state our additional assumption, we need to introduce
\[
	\eta(s) = - \liminf_{n\to\infty} \frac{1}{n} \log J_{ns},
\]
which we extend to \(\Rd\) by positive homogeneity of order \(1\).
From Property~\ref{weight_property:finite_energy}, one always has \(\nu(s) \leq \eta(s)\) (see~\cite{Aoun+Ioffe+Ott+Velenik-CMP-2021}). We say that there is \emph{saturation in direction \(s\)} if \(\nu(s) = \eta(s)\). Note that saturation never occurs when \(J\) decays superexponentially fast (in which case \(\eta(s) = \infty\)).

In the case of the Ising model, using monotonicity (see again~\cite{Aoun+Ioffe+Ott+Velenik-CMP-2021}), we can then define
\[
	\betasat(s) = \inf\setof{\beta\in\R}{\nu_{\beta}(s)<\eta(s)}.
\]

\subsubsection{The main assumption (NSA)}\label{ssec:NSA}
It is convenient to work using convex duality. Define the (convex, centrally symmetric, possibly unbounded) set
\[
\calW_{\eta} = \bigcap_{s\in\bbS^{d-1}} \setof{x\in\Rd}{s\cdot x \leq \eta(s)}.
\]
We use the same notation as for the polar set associated to a norm, since one has the inequality (which is an equality when \(\eta\) is a norm)
\[
\eta(x) \geq \sup_{t\in\calW_{\eta}} t\cdot x = \sup_{t\in\partial\calW_{\eta}} t\cdot x.
\]
\begin{definition}
	The set of \emph{non-saturated dual vectors} is defined as the open set
	\[
	\calT = \partial\calW_{\nu}\setminus \partial\calW_{\eta}.
	\]
\end{definition}

\medskip
Our analysis below will apply to directions satisfying the following assumption.
\begin{definition}[No saturation assumption]
	Fix \(s\in\bbS^{d-1}\). Recall that \(T_s\) is the set of all vectors \(\nu\)-dual to \(s\). We say that the \emph{no saturation assumption} (NSA) is fulfilled if \(T_s\cap \partial\calW_{\eta}\neq T_{s}\).%
\end{definition}
\begin{remark}
	A typical situation in which NSA is violated is illustrated in Fig.~\ref{fig:noNSA}, assuming, for simplicity, that \(\eta\) is a norm. The direction \(s\) points towards an affine piece of \(\partial\calU_\nu\) with unique normal \(t\in T_s\). While saturation does not occur in the direction \(s\) itself (\(\nu(s) > \eta(s)\)), there is another direction \(s'\), pointing toward the same facet \(s\) belongs to, at which saturation occurs: \(\nu(s') = t\cdot s' = \eta(s')\). By the inclusion \(\calU_\eta \subset \calU_\nu\), this means that \(t\in \partial\calW_\eta\), which shows that NSA is indeed violated.
	The reason failure of NSA is problematic for our analysis is explained in Section~\ref{subsec:not_many_edges}.
\end{remark}
\begin{figure}[ht]\label{figure:noNSA}
	\includegraphics{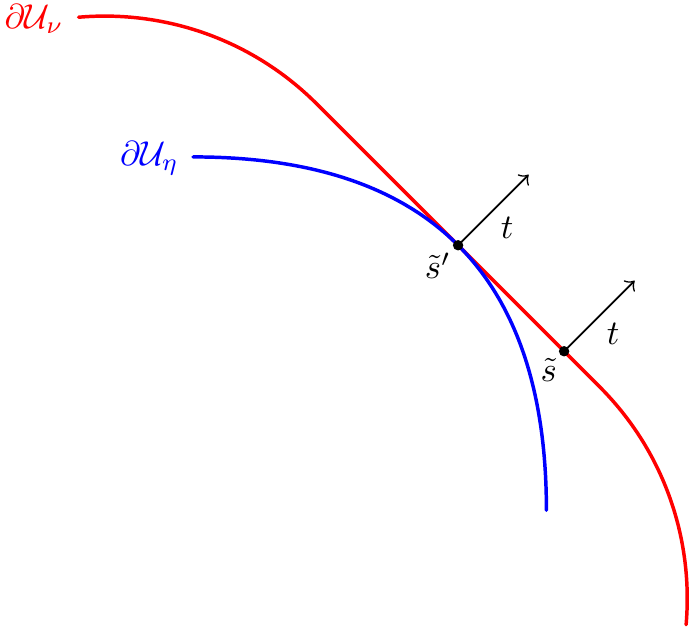}
	\caption{An example in which NSA does not hold: the direction \(s\in\bbS^{d-1}\) is such that \(\tilde{s}=s/\nu(s)\) belongs to a facet of \(\partial\calU_\nu\), while there is a direction \(s'\in\bbS^{d-1}\) in which saturation occurs and \(\tilde{s}' = s'/\nu(s')\) belongs to the same facet of \(\partial\calU_\nu\).}
	\label{fig:noNSA}
\end{figure}

The next lemma lists several cases in which it is known that NSA is fulfilled.
\begin{lemma}\label{lem:NSA}
NSA is satisfied whenever at least one of the following assumptions is verified:
\begin{itemize}
	\item \(J\) decays superexponentially;
	\item \(J_x = \psi(x)\sfe^{-\rho(x)}\), with \(\rho\) a norm and \(\sum_{n\geq 0}\psi(ns) = \infty\) for all \(s\in\bbS^{d-1}\);
	\item the interaction is exponentially-decaying and
	there are no saturated directions (that is, for the Ising model, \(\sup_{s\in\bbS^{d-1}} \betasat(s)< \beta<\betac\));%
	\item \(J\) has all symmetries of \(\Zd\), is exponentially decaying, \(s= \pm \rme_i\) with \(\rme_i\) the \(i\)-th canonical coordinate vector, and the direction \(s\) is not saturated (for Ising, \(\betasat(s)<\beta<\betac\)).%
\end{itemize}
\end{lemma}
\begin{proof}
The first two cases are particular instances of the third one: indeed, in both cases, \(\sup_{s\in\bbS^{d-1}} \betasat(s) =0\). This is a consequence of the condition derived in~\cite{Aoun+Ioffe+Ott+Velenik-CMP-2021,Aoun+Ioffe+Ott+Velenik-PRE-2021,Aoun+Ott+Velenik-2021-inPreparation}: \(\sup_{s\in\bbS^{d-1}} \betasat(s) =0\) if and only if \(\sum_{x\in\Zd} J_x\sfe^{t\cdot x} = \infty\) for any \(t\in \partial\calW_{\eta}\).

Let us briefly sketch the proof of the last case. Without loss of generality, fix $s=e_1$. Let us show that one can always find a vector dual to $s$ which is not saturated. 
There are only two cases where there could exist saturated dual vectors: either there is a unique dual vector to $s$, or there are infinitely many of them.
	
In the first case, let \(t=\nu(e_1)e_1\) be the unique dual vector to \(s\). Assume by contradiction that \(t\) is saturated: this means that there exists a direction \(s'\in\bbS^{d-1}\) such that \(\nu(s') = t\cdot s' = \eta(s')\). By symmetry, the same is true of \(s'' = (s'_1,-s'_\perp)\), where we have written \(s'_1 = s'\cdot e_1\) and \(s'_\perp = s' - s'_1\). By construction, \(\nu(s'') = t\cdot s'' = \eta(s'')\).
Since \(s'_1/\nu(s') = s''_1/\nu(s'') = 1/\nu(s)\), it follows that the line segment \([s'/\nu(s'), s''/\nu(s'')]\) contains \(s/\nu(s)\) and is thus included in \(\partial\calU_\nu\). Since \(\calU_\eta\subset\calU_\nu\) and \(s'_1/\nu(s'), s''_1/\nu(s'')\in\partial\calU_\eta\), it follows from the convexity of \(\calU_\eta\) that \([s'/\nu(s'), s''/\nu(s'')]\) is also a subset of \(\partial\calU_\eta\). This implies that \(s/\nu(s)\in\partial\calU_\eta\) and thus that \(\nu(s)=\eta(s)\), which is a contradiction.

In the second case, the dual vector \(t=\nu(e_1)e_1\) is never saturated, so there is nothing to do for this particular choice.
\end{proof}

We expect that NSA always holds, but we lack a proof at moment of writing.
\begin{conjecture}
	NSA holds for any \(\beta\in (\betasat(s),\beta_{c})\).
\end{conjecture}

\subsubsection{Consequences of NSA}

%Si jamais : [Theorem~15.2] dans Rockafellar-1970
A useful fact~\cite{Rockafellar-1970} is that any nonempty, compact, convex and centrally symmetric set \(W\subset (\calW_{\eta}\setminus\partial\calW_{\eta})\) is the polar set of a norm \(\tilde{\eta}\), \(W=\calW_{\tilde{\eta}}\), and this norm satisfies: there exists \(C_{\tilde{\eta}}<\infty\) such that
\[
	\forall x\in\Zd,\quad J_x \leq C_{\tilde{\eta}}\, \sfe^{-\tilde{\eta}(x)}.
\]

\begin{definition}
	We say that \(U\subset\calT\) is \emph{nice} if it is the closure of an open subset of \(\calT\).
\end{definition}

The main reason for working with nice subsets of \(\calT\) is the following mass-gap Lemma.
\begin{lemma}
	\label{lem:mass_gap_long_edges}
	Let \(U\subset\calT\) be nice. Then, there exist \(m>0\) and \(C\) such that, for any \(t\in U\) and \(x\in\Zd\),
	\[
		\sfe^{t\cdot x}J_x \leq C\, \sfe^{-m \nu(x)}.
	\]
\end{lemma}
\begin{proof}
	As \(\calT\) is open and \(U\subset\calT\) is closed (in \(\partial\calW_{\nu}\), so it is compact), there exists \(\epsilon>0\) such that
	\[
		U_{\epsilon} = \bigcup_{t\in U} \bbB_{\epsilon}(t) \subset \calW_{\eta}\setminus \partial\calW_{\eta}.
	\]
	Let \(\bar{U}_{\epsilon}\) be the convex hull of \(U_{\epsilon} \cup -U_{\epsilon}\) (both are sets in \(\R^d\)). It is a compact, convex, and centrally symmetric set. Therefore, there is a norm \(\tilde{\eta}\) such that \(\bar{U}_{\epsilon}=\calW_{\tilde{\eta}}\) and, by construction, \(J_x\leq C_{\tilde{\eta}}\sfe^{-\tilde{\eta}(x)}\).
	\begin{figure}[ht]
		\includegraphics{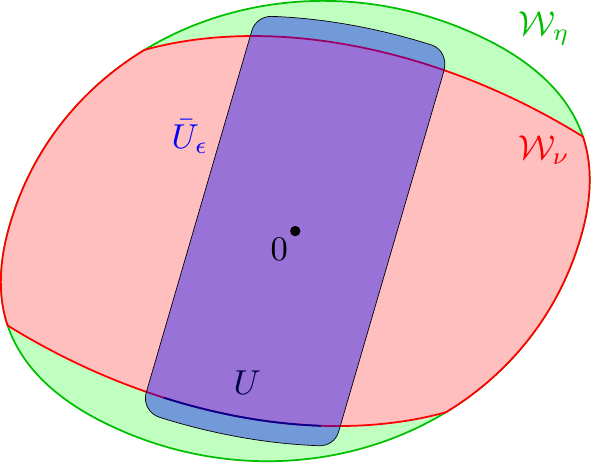}
		\caption{The construction in Lemma~\ref{lem:mass_gap_long_edges}. Notice that the introduction of the set \(\bar{U}_{\epsilon}\) is not necessary when \(\eta\) is a norm.}
		\label{fig:Lemma3_1}
	\end{figure}

	We then consider two cases. First, if \(x\cdot t \leq \nu(x)/2\), we have
	\[
		\sfe^{t\cdot x }J_x\leq C\sfe^{t\cdot x-3\nu(x)/4}\leq C\, \sfe^{-\nu(x)/4},
	\]
	since \(J_{x}\leq C\,\sfe^{-3\nu(x)/4}\) (by definition of \(\eta\) and since \(\eta(x)\geq \nu(x)\)). The case \(x\cdot t > \nu(x)/2\) follows from compactness of \(U\), the fact that \(\bar{U}_{\epsilon}=\calW_{\tilde{\eta}}\) and the fact that, for any \(t\in U\),
	\[
		t\cdot x - \tilde{\eta}(x)
		=
		t\cdot x - \sup_{t'\in \calW_{\tilde{\eta}}} t'\cdot x
		\leq
		t\cdot x - (t+\epsilon\hat{t})\cdot x
		=
		- \epsilon \hat{t}\cdot x
		<
		-\frac{\epsilon}{2\norm{t}} \nu(x).\qedhere
	\]
\end{proof}
The heuristic reason behind our use of the word ``massgap'' above is that Lemma~\ref{lem:mass_gap_long_edges} shows that, for any \(t\in U\) and any \(s\in\bbS^{d-1}\) \(\nu\)-dual to \(t\), one has \(J_{ns} \leq C \sfe^{-(m+1)\nu(s) n}\), which implies that \(\eta(s) > \nu(s)\). There is thus indeed a gap between the rates of exponential decay in direction \(s\) of the interaction and of the 2-point function.

In particular, (NSA) implies the existence of a non-saturated vector \(t\) dual to \(s\). The set of non-saturated vectors being open, it follows that there exists a nice set \(U\subset\calT\) containing \(t\) in its interior.%

Note that by compactness of \(\bbS^{d-1}\) and of \(T_s\), NSA at \(s\) implies NSA in the closure of an open neighborhood of \(s\). In particular, for any of those directions \(s'\), \(T_{s'}\subset \calT\). The union of these sets of dual vectors will provide a suitable nice subset of \(\calT\). This whole procedure is necessary, since we need uniformity of the mass gap (and of other statements) over suitable neighborhoods of a given direction and of the associated dual vectors. NSA might not hold if \(\calU_{\nu}\) is not strictly convex, see Figure~\ref{fig:noNSA}. Although our analysis implies that \(\calU_{\nu}\) is strictly convex in directions in which NSA holds, we don't know how to establish strict convexity a priori. Of course, everything discussed above becomes trivial if there is no saturated direction.

\subsection{Precise statements of the results}\label{sec:StatementResults}
In this section, we provide precise statement of the main results of the present work.

\begin{theorem}[Random Walk Coupling]
	\label{thm:main_coupling_with_RW}
	Let \(q\in \calQ\). Let \(s\in\bbS^{d-1}\) be such that NSA holds. Then, there exist probability measures \(\rmp, \rmp_L, \rmp_R\) on \(\SetRootDiaCont, \SetRootMarkBackCont,\SetRootMarkForwCont\) resp., a normalization constant \(C_{LR}\) and \(C, c > 0\) such that, for any \(n\geq 0\) and any \(f:\pathSet\to \bbC\),
	\[
		\Bigl|\sfe^{n\nu(s)} \sum_{\gamma:\, 0\to ns} q(\gamma) f(\gamma) -
		C_{LR} \sum_{\substack{\gamma_L\in\SetRootMarkBackCont\\\gamma_R\in\SetRootMarkForwCont}} \sum_{n\geq 0} \sum_{\gamma_1, \dots, \gamma_n\in\SetRootDiaCont} f(\bar{\gamma}) \mathds{1}_{\{\displace(\bar{\gamma}) = ns\}} \rmp_L(\gamma_L) \rmp_R(\gamma_R) \prod_{i=1}^n \rmp(\gamma_i) \Bigr| 
		\leq C \normsup{f} \sfe^{-cn},
	\]
	where \(\bar{\gamma} = \gamma_L\concatenate\gamma_1\concatenate\cdots\concatenate\gamma_n\concatenate \gamma_R\). Moreover, the following two properties hold:
	\begin{itemize}
		\item Exponential decay: there exist \(C', c'>0\) such that
		\[
		\sum_{\gamma:\, \norm{\displace(\gamma)}\geq l } p(\gamma) \leq C'\sfe^{-c'l},
		\]for \(p\in \{\rmp, \rmp_L, \rmp_R\}\);
		\item Finite energy: there exists \(c''>0\) such that
		\[
		p(\gamma)\geq c''\sfe^{t_0\cdot \displace(\gamma)} q(\gamma),
		\]for \(p\in \{\rmp, \rmp_L, \rmp_R\}\), \(\gamma\) irreducible, and \(t_0\) \(\nu\)-dual to \(s\).
	\end{itemize}
\end{theorem}
Theorem~\ref{thm:main_coupling_with_RW} is proved in Section~\ref{subsec:coupling_with_RW}.

\medskip
Let \((S_k)_{k\geq 0}\) be the (directed) random walk on \(\Zd\) with transition probabilities given by the push-forward \(\rmp(y) = \sum_{\gamma\in\SetRootDiaCont:\,\displace(\gamma)=y} \rmp(\gamma)\). Denote by \(\RWP{u}\) the distribution of \((S_k)\) conditionally on \(\{S_0=u\}\) and by \(\RWG(u,v) = \sum_{k\geq 0} \RWP{u}(S_k=v)\) the corresponding Green function. 
\begin{theorem}[OZ asymptotics]
	\label{thm:main_OZ_asymp}
	Let \(q\in\calQ\). Let \(s\in\bbS^{d-1}\) be such that NSA holds. Then, there exist probability measures \(\rmp_L, \rmp_R\) on \(\Zd\) with exponential tails, a normalization constant \(C_{LR}\) and \(c > 0\) such that, for any \(n\geq 0\),
	\[
	\sfe^{n\nu(s)}\, G(0,ns) = (1+\sfO(\sfe^{-cn}))\, C_{LR} \sum_{u,v\in\Zd} \rmp_L(u) \rmp_R(v) \RWG(u,ns-v).
	\]
	In particular, there exists \(c_s>0\) such that, for every \(\varepsilon>0\),
	\[
	G(0,ns) = \frac{c_s}{n^{(d-1)/2}}\, \sfe^{-n\nu(s)}(1+\sfo(n^{\varepsilon-1/2})).
	\]
\end{theorem}
Theorem~\ref{thm:main_OZ_asymp} is proved in Section~\ref{section:OZ_asymptotics}.
\begin{corollary}[Typical number of steps in \(\gamma\)]
	\label{cor:NumberSteps}
	Let \(q\in\calQ\). Let \(s\in\bbS^{d-1}\) be such that NSA holds. Then, there exists \(A>0\) such that, for any \(\epsilon>0\), there exists \(c> 0\) such that
	\begin{equation*}
		G(0,ns)=(1+\sfO(\sfe^{-cn}))\sum_{\substack{\gamma: 0\rightarrow ns \\ (1-\epsilon) An\leq\abs{\gamma}\leq (1+\epsilon) An}}q(\gamma).
	\end{equation*}
	Moreover, there exists \(\epsilon > 0\) and a rate function \(I_s\) on \((A-\epsilon, A+\epsilon)\) with a quadratic minimum at \(A\) such that, for all \(\alpha \in (A-\epsilon, A+\epsilon)\),
	\[
	\sum_{\substack{\gamma: 0\rightarrow ns \\ \abs{\gamma} = \lfloor\alpha n\rfloor}}q(\gamma) = (1+\sfo_n(1))\, \frac{C_s(\alpha)}{\sqrt{n}} \sfe^{-nI_s(\alpha)}\, G(0,ns),
	\]
	where \(C_s(\cdot)\) is a positive real analytic function on \((A-\epsilon, A+\epsilon)\).
\end{corollary}
Corollary~\ref{cor:NumberSteps} is proved along the lines of Theorem~5.2 in~\cite{Ioffe+Velenik-2008}.
\begin{theorem}[Local Analyticity of \(\nu\)]
	\label{thm:main_Wulff_analytic}
	Let \(s\in\bbS^{d-1}\) be such that NSA holds. Then,
	\begin{itemize}
		\item \(s\mapsto \nu(s)\) is analytic in a neighborhood of \(s\);
		\item there exists a unique \(t\in\partial \calW_{\nu}\) dual to \(s\);
		\item \(\partial\calW_{\nu}\) (resp.\ \(\partial \calU_{\nu}\)) is analytic in a neighborhood of \(t\) (resp.\ \(s/\nu(s)\)).
	\end{itemize}
\end{theorem}
Theorem~\ref{thm:main_Wulff_analytic} is proved in Section~\ref{sec:ProofAnalyticity}.
\begin{remark}
	Whenever our assumption NSA holds in every direction simultaneously, one can deduce the uniform strict convexity of \(\calU_\nu\) and \(\calW_\nu\), exactly as was done in earlier works (for instance, \cite{Campanino+Ioffe-2002}).
\end{remark}

\subsection{Sketch of the proof and organization of the paper}

The proof follows the same key steps as the procedure developed in~\cite{Campanino+Ioffe-2002,Campanino+Ioffe+Velenik-2003} with the technical and conceptual refinements of~\cite{Campanino+Ioffe+Velenik-2008,Ott+Velenik-2018,Ott-2020}.
%: one uses a graphical representation of correlations, then a coarse-graining of the graphs appearing in this representation, followed by an analysis of the fine geometry of the graphs to construct a suitable renewal structure.
The output of the construction is a coupling of the graphs with a directed random walk. The main novelty of the present paper is the presence of infinite-range interactions (arbitrarily long edges in the graphical representation) which complicates both the coarse-graining procedure and the study of the coarse-grained object. These complications are not only of a technical nature: the infinite-range of interactions can lead to failure of OZ decay, as was shown in~\cite{Aoun+Ioffe+Ott+Velenik-CMP-2021,Aoun+Ioffe+Ott+Velenik-PRE-2021,Aoun+Ott+Velenik-2021-inPreparation}.

We also tried to make this paper as self-contained and pedagogical as possible, gathering pieces scattered in many earlier papers. In particular, we aim at presenting the arguments in a way which makes the different steps of the proof as independent as possible. Our hope is that this paper might offer a reasonable introduction to this topic for interested people.

\smallskip
Let us now describe briefly the 5 main steps of the proof.

\smallskip	
\paragraph{\textbf{Step~1: Coarse-graining of paths}}

The reason the nonperturbative analysis of the paths contributing to the Green function is difficult is that, once \(\beta\) is close to \(\betac\), typical paths exhibit a very complicated structure at the lattice scale. However, once observed at a scale that is a large multiple of the correlation length, the geometry of such paths becomes once more very simple.

In order to make this precise, we make use of a coarse graining of the microscopic path \(\gamma\), approximating the latter by a polygonal line defined on scale given by a (fixed) multiple \(K\) of the correlation length, which we call its \textit{skeleton}.
This construction is described in detail in Section~\ref{sec:Coarse-graining}.

\smallskip
\paragraph{\textbf{Step~2: Geometry of typical skeletons}}

The relevance of this construction can be seen from simple energy/entropy considerations. Namely, the weight of a given skeleton (that is, the total weight of the paths associated to this particular skeleton) decreases exponentially with \(K\) times the number of vertices of the skeleton, while the number of skeletons with a given number of vertices only grows exponentially with \(\log K\) times the number of vertices. This means that, once \(K\) is chosen large enough, energy dominates entropy and we are essentially back in an effective perturbative regime, with \(K\) playing the role of the perturbation parameter.

This enables one to use energy/entropy arguments in order to describe the geometry of typical skeletons.  Typicality of a skeleton is conveniently quantified by its \textit{surcharge}, a quantity of purely geometric nature.

This construction is detailed in Section~\ref{sec:skeleton_analysis} and the output is a precise geometric description of typical skeletons, showing that the latter are ballistic in a strong sense and contain only few long edges. It is this last requirement that makes the construction in the present paper substantially more complicated than in earlier works.

\smallskip
\paragraph{\textbf{Step~3: Geometry of typical paths}}

In the next step of the analysis, we deduce from the description of typical skeletons (defined at the scale of the correlation length) the description of typical paths (defined at the scale of the lattice).
Namely, using the strong ballisticity of typical skeletons, the fact that paths remain close to their skeleton and probabilistic arguments, we show that the strong ballisticity statement extends to typical paths.
This is described in Section~\ref{sec:path_analysis}. The output is a decomposition of typical paths into strings of irreducible subpaths, \(\gamma=\gamma_{L}\concatenate\gamma_{1}\concatenate\cdots\concatenate\gamma_{M}\concatenate\gamma_{R}\), each subpath being confined in ``diamonds'' (intersection of two cones), as depicted in Figure~\ref{fig:IrreducibleDecomposition}.

\smallskip
\paragraph{\textbf{Step~4: Factorization of the weight and coupling to a directed random walk}}

Unfortunately, since the path-weight \(q\) does not factorize in general, that is,
\[
	q(\gamma) \neq q(\gamma_{L})q(\gamma_{1})\cdots q(\gamma_{M})q(\gamma_{R}),
\]
the decomposition into irreducible pieces does not induce a renewal structure on which to build the coupling to a directed random walk.
In Section~\ref{sec:facto_meas}, we explain how the introduction of additional artificial degrees of freedom and the disintegration of the weight \(q(\gamma)\) with respect to the latter yield joint weights that enjoy good factorization properties. Using this, we can finally construct the desired coupling between typical paths \(\gamma\) contributing to the Green function and a directed random walk on \(\Zd\) with nice properties (increments with exponential tails, etc.).

\smallskip
\paragraph{\textbf{Step~5: Proof of the main results}}

Equipped with this powerful coupling, we complete the proofs of its various by-products in Section~\ref{sec:proof_main_thm}.

%%%%%%%%%%%%%%%%%%%%%%%%%%%%%%%%%%%%%%%%%%%%%%%%%%%%%%%%%%%%%%%%%%%%%%%%%%%%%%%%%%%%%%%%%%%%%%%%%%%%%%%%
\section{Coarse-graining of paths}\label{sec:Coarse-graining}
%%%%%%%%%%%%%%%%%%%%%%%%%%%%%%%%%%%%%%%%%%%%%%%%%%%%%%%%%%%%%%%%%%%%%%%%%%%%%%%%%%%%%%%%%%%%%%%%%%%%%%%%

\subsection{Coarse-graining algorithm}

We now present the coarse-graining procedure which we will apply to the paths. The coarse graining depends on a scale parameter \(K\).

\begin{algorithm}[H]
	\label{algo:skeleton}
	\caption{Skeleton algorithm; see Fig.~\ref{fig:ExtractionSkeleton}}
	\KwInput{a path \(\gamma\) with \(\gamma_0=v_0\) and \(|\gamma| = n\)}
	Set \(V = (v_0)\), \(a_0=0\), \(k=0\);\\
	\While{\(\gamma_{a_k}^{n} \not\subset \Delta_K(v_k)\)}{
		Let \(a_{k+1} = \min\setof{a > a_k}{\gamma_{a} \notin \Delta_K(v_k) }\);\\
		Set \(v_{k+1} = \gamma_{a_{k+1}}\);\\
		Update \(V = (v_0, \dots, v_{k+1})\), \(k=k+1\);\\
	}
	\KwOutput{\(\Skeleton(\gamma)=V\)}	
\end{algorithm}
\begin{figure}[ht]
	\includegraphics{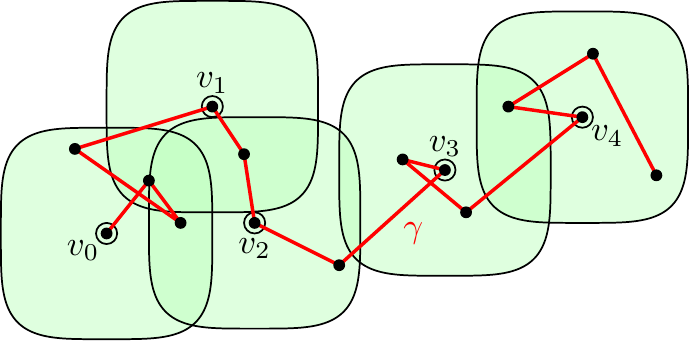}
	\caption{The skeleton \((v_0, \dots, v_4)\) (circled vertices) associated to the path \(\gamma\). The translates of the set \(\Delta_K\) used in the construction are indicated.}
	\label{fig:ExtractionSkeleton}	
\end{figure}

\pagebreak % <<<<<<<<<<<<<<<<<<<<<<<<<<<<<<<<<<< PAGE BREAK

\subsection{Energy extraction}

The key tool to study the typicality of skeletons is the following energy estimate.
\begin{lemma}[Energy bound]
	\label{lem:skeleton_energy_bound}
	Let \(q\in \calQ\). Then, for any \(K\geq 0\) and \(R\geq 0\), any \(x\in\Zd\) with \(\norm{x}\) large, any admissible path \(\tilde{\gamma}\) with \(\tilde{\gamma}_{|\tilde{\gamma}|} = 0\), and any \(K\)-skeleton \(V=(0=v_0, v_1, \dots, v_M)\) with \(x\in\Delta_K(v_M)\),
	\[
		\sum_{\gamma:\, 0\to x} q(\gamma\given \tilde{\gamma})\mathds{1}_{\{\Skeleton(\gamma) = V\}}
		\leq
		\prod_{(u,v)\in\shortEdge(V)} \sfe^{-\nu(v-u)(1+\sfo_K(1))} \prod_{(u,v)\in\longEdge(V)}F(v-u),
	\]
	where
	\[
		F(y) = \min\Bigl( \sum_{z\in\Delta_K} \sfe^{-\nu(z)(1+\sfo_{\norm{z}}(1))} C J_{y-z}, \sfe^{-\nu(y)(1+\sfo_{\norm{y}}(1))}\Bigr).
	\]
\end{lemma}

The distinction between short and long edges enables a control of the frequency of abnormally long edges and is the main additional difficulty compared to~\cite{Campanino+Ioffe+Velenik-2003}.
\begin{proof}
	First,
	\begin{align*}
		\sum_{\gamma:\, 0\to x} q(\gamma\given \tilde{\gamma})\mathds{1}_{\{\Skeleton(\gamma) = V\}}
		&\leq \sum_{\gamma:\, 0\to v_M} \sum_{\gamma':\, v_M\to x} q(\gamma'\given \tilde{\gamma}\concatenate\gamma) q(\gamma\given \tilde{\gamma}) \mathds{1}_{\{\Skeleton(\gamma\concatenate\gamma') = V\}}\\
		&\leq \sfe^{-\nu(x-v_M) (1+\sfo_{\norm{x-v_M}}(1))} \sum_{\gamma:\, 0\to v_M} q(\gamma) \mathds{1}_{\{\Skeleton(\gamma) = V\}},
	\end{align*}
	where we used properties~\ref{weight_property:repulsion} and~\ref{weight_property:ICL}.
	Let \(V_k=(v_0, \dots, v_k)\). We can now proceed inductively:
	\[
		\sum_{\gamma:\, 0\to v_{k}} q(\gamma\given \tilde{\gamma}) \mathds{1}_{\{\Skeleton(\gamma) = V_k\}}
		\leq
		\sum_{\gamma:\, 0\to v_{k-1}} q(\gamma\given \tilde{\gamma}) \sum_{\gamma':\, v_{k-1}\to v_k} q(\gamma' \given \tilde{\gamma}\concatenate \gamma) \mathds{1}_{\{\Skeleton(\gamma\concatenate \gamma') = V_k\}}.
	\]
	
	To get the small steps part of the bound, we can simply use properties~\ref{weight_property:repulsion} and~\ref{weight_property:ICL}, as well as \(\mathds{1}_{\{\Skeleton(\gamma\concatenate \gamma') = V_k\}}\leq \mathds{1}_{\{\Skeleton(\gamma) = V_{k-1}\}}\) to bound the above by
	\begin{equation}\label{eq:intermediate_estimate_skeleton}
		\sfe^{-\nu(v_{k}-v_{k-1}) (1+\sfo_{\norm{v_{k}-v_{k-1}}}(1) )} \sum_{\gamma:\, 0\to v_{k-1}} q(\gamma\given \tilde{\gamma}) \mathds{1}_{\{\Skeleton(\gamma) = V_{k-1}\}},
	\end{equation}
	which implies the desired result as \(v_{k}\notin \Delta_K(v_{k-1})\) implies \(\nu(v_{k}-v_{k-1})\geq K\).
	
	To improve the bound in the case of long steps, we use~\eqref{eq:intermediate_estimate_skeleton} and the observation that, whenever \((v_{k-1},v_k)\in\longEdge(V)\), the last step of \(\gamma'\) must be of the form \(\{z, v_{k}\}\) with \(z\in \Delta_K(v_{k-1})\) and \(v_{k}\notin \Delta_{K+R}(v_{k-1})\). Summing over the possible values of \(x\), we obtain
	\begin{align*}
		\sum_{\gamma':\, v_{k-1}\to v_k} q(\gamma' \given \tilde{\gamma}\concatenate\gamma)  \mathds{1}_{\{\Skeleton(\gamma\concatenate \gamma') = V_k\}}
			&\leq\mathds{1}_{\{\Skeleton(\gamma) = V_{k-1}\}} \sum_{z\in\Delta_K(v_{k-1})} \sum_{\gamma':\,v_{k-1}\to z} q(\gamma'\given \tilde{\gamma}\concatenate\gamma) q\bigl((z,v_k) \given \tilde{\gamma}\concatenate\gamma\concatenate \gamma'\bigr) \\
			&\leq \mathds{1}_{\{\Skeleton(\gamma) = V_{k-1}\}} \sum_{z\in\Delta_K(v_{k-1})} \sfe^{-\nu(z-v_{k-1})(1+\sfo_{\norm{z-v_{k-1}}}(1))} C J_{v_k-z},
	\end{align*}
	where we used properties~\ref{weight_property:repulsion}, \ref{weight_property:ICL} and~\ref{weight_property:weight_growth}.
\end{proof}

%%%%%%%%%%%%%%%%%%%%%%%%%%%%%%%%%%%%%%%%%%%%%%%%%%%%%%%%%%%%%%%%%%%%%%%%%%%%%%%%%%%%%%%%%%%%%%%%%%%%%%%%
\section{Skeleton analysis}\label{sec:skeleton_analysis}
%%%%%%%%%%%%%%%%%%%%%%%%%%%%%%%%%%%%%%%%%%%%%%%%%%%%%%%%%%%%%%%%%%%%%%%%%%%%%%%%%%%%%%%%%%%%%%%%%%%%%%%%

We now study ``typical'' skeletons produced by the coarse-graining procedure. The idea is to compare the weight of certain sets of skeletons to the a priori energy cost \(\sfe^{-\nu(ns)}\) of having a path from \(0\) to \(ns\), in order to show that with high probability these skeletons do not occur as output of the coarse-graining procedure. In order to make this section functionally independent from the rest, we shall introduce the notion of skeleton measure.
\begin{definition}[Skeleton measures]
	A \(K\)-\emph{skeleton measure} is a weight function 
	\[
		Q_{K}: \bigcup_{x\in\Zd} \SkeletonSet(x) \to \R_{\geq 0}.
	\]
	We will say that the family \((Q_{K})_{K\geq 0}\) satisfies a \emph{\(K\)-energy bound} if there exists \(C<\infty\) (independent of \(K\)) such that, for any \(R\geq 0\) and any \(K\)-skeleton \(V=(v_0, v_1, \dots, v_M)\),
	\begin{equation}
		\label{eq:skeleton_energy_bound}
		Q_{K}(V) \leq \prod_{(u,v)\in \shortEdge(V)} \sfe^{-\nu(v-u)(1+\sfo_K(1))} \prod_{(u,v)\in \longEdge(V)} F(v-u),
	\end{equation}
	where
	\[
		F(y) = \min \Bigl(\sfe^{-\nu(y)(1+\sfo_{\norm{y}}(1))}, \sum_{x\in\Delta_{K}} \sfe^{-\nu(x)(1+\sfo_{\norm{x}}(1))}CJ_{y-x}\Bigr).
	\]
	The \(R\)-dependence is in the definition of \(\shortEdge(V)\) and \(\longEdge(V)\).
\end{definition}
The parameter \(R\) serves as a cut-off for the interaction: for a given \(R\), small steps behave in the same way as for models with range \(R\).

In the sequel, we are going to show that various families of skeletons are atypical, by establishing bounds of the type
\[
	\sfe^{t\cdot x} \sum_{V\in \SkeletonSet(0)} Q_K(V) \mathds{1}_{\{V\in\mathcal{V}\}} \leq \sfe^{-a \nu(x)},
\]
where \(\mathcal{V}\) is a collection of skeletons associated to paths \(\gamma:\, 0\to x\) (in particular, \(v_M\in \Delta_K(x)\)), \(t\) belongs to a nice subset of \(\partial\calW_\nu\) and \(a>0\). Such a bound indeed proves that paths associated to skeletons in \(\mathcal{V}\) yield an exponentially negligible contribution to \(G(0,x)\), since, taking \(t\) dual to \(x\) and using the a priori bound \(G(0,x) \geq \sfe^{-(1-\sfo_{\norm{x}}(1))\nu(x)}\), it implies that
\[
	\sum_{\substack{\gamma:\,0\to x:\\\Skeleton(\gamma) \in \mathcal{V}}} q(\gamma)
	\leq
	\sfe^{-\nu(x)} \sfe^{t\cdot x} \sum_{V\in \SkeletonSet(0)} Q_K(V) \mathds{1}_{\{V\in\mathcal{V}\}}
	\leq
	\sfe^{-a\nu(x) - \nu(x)} \leq \sfe^{-\frac a2\nu(x)} G(0,x), 
\]
when \(\norm{x}\) is large enough.

\medskip
The goal of this section is the proof of the following result.
\begin{theorem}[Typical skeletons have many cone-points]
	\label{thm:skeleton_CP}
	Let \(Q_K\) be a family of skeleton measures satisfying the energy bound~\eqref{eq:skeleton_energy_bound}. Let \(U\subset \calT\) be nice. Then, for any \(\delta>0, \epsilon>0, p\in(0,1)\), there exist \(K_0\geq 0\) and \(a>0\) such that, for any \(t\in U\), \(K\geq K_0\), \(
%	y, 
	x\in \Zd\) with \(\norm{x}\) large enough,
	\[
		\sfe^{t\cdot x} \sum_{V\in \SkeletonSet(0)} Q_K(V) \mathds{1}_{\{v_M\in \Delta_K(x)\}} \mathds{1}_{\{|\CPts_{t,\delta}(V)\cap \slab_{\norm{x}}(\hat{x})|\leq p \nu(x)/K\}}
		\leq \sfe^{-a \nu(x)},
	\]
	where \(V=(v_0=0, v_1, \dots, v_M)\).
\end{theorem}
We shall also need a corollary of the analysis (which is a corollary of Theorem~\ref{thm:skeleton_CP} in the finite-range case, but requires additional care in the present setup).
\begin{theorem}[Cone-points are well distributed]
	\label{thm:skeleton_CP_well_distributed}
	Let \(Q_K\) be a family of skeleton measures satisfying the energy bound~\eqref{eq:skeleton_energy_bound}. Let \(U\subset \calT\) be nice. Then, for any \(\delta>0, \epsilon>0, l>1\), there exist \(K_0\geq 0\) and \(a>0\) such that, for any \(t\in U\), \(K\geq K_0\), \(
%	y,
	x\in \Zd\) with \(\norm{x}\) large enough,
	\[
		\sfe^{t\cdot x} \sum_{V\in \SkeletonSet(0)} Q_K(V) \mathds{1}_{\{v_M\in \Delta_K(x)\}} \mathds{1}_{\left\{\sum_{k=1}^{\nu(x)/lr_*K} \mathds{1}_{\{\CPts_{t,\delta}(V)\cap H_k \neq \emptyset\}}	 \leq (1-\epsilon)\nu(x)/lK\right\}}
		\leq \sfe^{-a \nu(x)},
	\]
	where \(V=(v_0=0, \dots, v_M)\), \(H_k = \slab_{(k-1)lr_*K, klr_*K}(\hat{x})\) and \(r_*=r_*(\nu)\) is the extremal radius of \(\nu\).
\end{theorem}

\begin{remark}
	By translation invariance, the same results hold when \(v_0\) is taken arbitrary and \(x\) is replaced with \(x-v_0\).
\end{remark}

\subsection{Entropy of skeletons}

The main feature of skeletons satisfying an energy bound is that their entropy is almost trivial and their large-scale geometry is therefore governed by their energy.
\begin{lemma}[Entropy bound]
	\label{lem:skeleton_entropy_bound}
	Let \(K, R\geq 0\). Let \(M\geq 1\), \(I\subset \{1, \dots, M\}\) and \(r_i\in\Z_{>0}\), \(i\in I\). Then,
	\begin{equation}\label{eq:EntropyBound}
		|\SkeletonSet_{M,I,(r_i)_{i\in I}}| \leq (c R(K+R)^{d-1})^{M-|I|} \prod_{i\in I} c (K+R+r_i)^{d-1},
	\end{equation}
	where \(c\) depends on \(d\) and \(\nu\) only.
\end{lemma}
\begin{proof}
	The bound is an easy consequence of the following observations. First, if \((v_k, v_{k+1})\) is a short step, then \(v_{k+1}\in \Delta_{K+R}(v_k)\setminus \Delta_{K}(v_k)\) and the cardinality of the latter set is bounded above by \(cR(R+K)^{d-1}\). Second, if \((v_{k}, v_{k+1})\) is a long step with \(v_{k+1}\in \Delta_{K+R+r}(v_k)\setminus \Delta_{K+R+r-1}(v_k)\), the number of  possibilities for \(v_{k+1}\) given \(v_k\) is bounded by \(c(K+R+r)^{d-1}\).
\end{proof}

\subsection{Surcharge function}

A key object in the study of skeletons is the surcharge function. Informally, it quantifies how far a path is from being a minimal path for \(\nu\).

\begin{definition}[Surcharge function]
	Let \(\rho\) be a norm on \(\Rd\). Let \(t\in\partial\calW_{\rho}\). The \emph{surcharge function} is the function on \(\Rd\) defined by
	\begin{equation}
		\label{eq:def:surcharge}
		\surcharge_{\rho,t}(z) = \rho(z) - t\cdot z.
	\end{equation}
	When omitted from the notation, \(\rho\) is set to be \(\nu\). Note that \(\surcharge_{\rho,t}\) inherits the triangular inequality from \(\rho\).
\end{definition}
From the definition of \(\calW_{\rho}\), one immediately deduces that \(\surcharge_{\rho,t}\geq 0\). We will regularly use the observation that, by definition of the cones, \(\surcharge_t(z) \geq \delta \nu(z)\) for any \(z\notin\fcone_{t,\delta}\).
\begin{definition}[Surcharge of a skeleton]
	The \emph{surcharge of a skeleton} \(V  =(v_0, \dots, v_M)\) is defined by
	\[
		\surcharge_t(V) = \sum_{k=1}^{M} \surcharge_t(v_{k}-v_{k-1}).
	\]
\end{definition}

A simple observation is that, for any \(t\in\partial\calW_{\nu}\) and any \(y_1, \dots, y_m\in\Zd\),
\begin{equation}
	\label{eq:sum_of_surcharges_LB}
	\sum_{k=1}^m \surcharge_t(y_k) \geq \sum_{k=1}^{m} \nu(y_k) -\nu\Bigl(\sum_{k=1}^m y_k\Bigr).
\end{equation}
Another observation we will often use is that 
\begin{equation}\label{eq:SurchargeIncrements}
\sfe^{t\cdot \sum_{k=1}^m y_k}\prod_{i=1}^{m}\sfe^{-\nu(y_{i}-y_{i-1})}
=
\prod_{i=1}^{m}\sfe^{-\surcharge_{t}(y_{i}-y_{i-1})}
\end{equation}
by the linearity of the dot product.

\subsection{Skeletons do not contain too many vertices}

Before turning to the finer skeleton analysis, let us make a first simple observation (which will be used throughout the remainder of the proof).

\begin{lemma}
	\label{lem:skeleton_size_upper_bound}
	There exists \(K_0\geq 0\) such that, for any \(a\geq 0, K\geq K_0\) and \(n\) large enough,
	\[
		\sum_{V\in\SkeletonSet} Q_K(V) \mathds{1}_{\{|V|\geq a n/K\}} \leq \sfe^{-a(1+\sfo_K(1))n}.
	\]
\end{lemma}
\begin{proof}
	We will use the energy bound~\eqref{eq:skeleton_energy_bound} and the entropy bound~\eqref{eq:EntropyBound} with \(R=0\) (that is, every step is now a long step). The latter imply
	\begin{align*}
		\sum_{V\in\SkeletonSet} Q_K(V) \mathds{1}_{\{|V|\geq a n/K\}}
			&\leq \sum_{M\geq a n/K} \sum_{r_i\geq 1, i\in \{1, \dots, M\}} |\SkeletonSet_{M,\{1, \dots, M\}, (r_i)_{i=1}^M}|\, \sfe^{-\sum_{i=1}^M (K+r_i)(1+\sfo_{K+r_i}(1))}\\
			&\leq \sum_{M\geq a n/K} \Bigl(\sum_{r=1}^{\infty} c(K +r)^{d-1} \sfe^{-(K+r)(1+\sfo_{K+r}(1))}\Bigr)^{\!\!M}\\
			&\leq \sum_{M\geq a n/K} \sfe^{-MK(1+\sfo_K(1))}
			\leq \sfe^{-a(1+\sfo_K(1)) n}. 
	\end{align*}
\end{proof}

\subsection{Skeletons have small surcharge}

The first key property of skeleton measures satisfying an energy bound is that their surcharge is controlled. In the \emph{finite-range case}, the combination of this control with a \emph{deterministic} result on skeletons which says that skeletons \emph{containing only small steps} with few cone-points have large surcharge is at the heart of the analysis in~\cite{Campanino+Ioffe+Velenik-2003}. In the present study, we have to add to this control of the surcharge a control over the global contribution of long edges.
\begin{lemma}
	\label{lem:skeleton_small_surcharge}
	For any \(\epsilon>0\), there exists \(K_0\geq 0\) such that, for any \(K\geq K_0\), \(t\in \partial \calW_{\nu}\) and \(x\in\Zd\) with \(\norm{x}\) large enough,
	\[
		\sfe^{t\cdot x} \sum_{V\in\SkeletonSet} Q_K(V) \mathds{1}_{\{\surcharge_t(V)\geq \epsilon\nu(x)\}} \mathds{1}_{\{v_M\in \Delta_K(x)\}} \leq \sfe^{-\epsilon\nu(x)/2}.
	\]
\end{lemma}
\begin{proof}
	We use the energy bound~\eqref{eq:skeleton_energy_bound} and the entropy bound~\eqref{eq:EntropyBound} with \(R=0\) to conclude that, for any \(0<a<1\),
	\begin{multline*}
		\sfe^{t\cdot x}\sum_{V\in\SkeletonSet} Q_K(V) \mathds{1}_{\{\surcharge_t(V)\geq \epsilon \nu(x)\}} \mathds{1}_{\{x\in \Delta_K(v_M)\}} \\
		\begin{split}
			&\leq \sfe^{K}\sum_{V\in\SkeletonSet} \sfe^{-\surcharge_t(V)} \prod_{k=1}^{M} \sfe^{\sfo_K(1) \nu(v_{k}-v_{k-1})} \mathds{1}_{\{\surcharge_t(V)\geq \epsilon \nu(x)\}} \mathds{1}_{\{x\in \Delta_K(v_M)\}}\\
			&\leq \sfe^{K}\sfe^{-\epsilon(1-a)\nu(x)} \sum_{M\geq 1} \sum_{\substack{y_1, \dots, y_M\in \Delta_K^c\\ \sum y_i \in \Delta_K(x)}} \prod_{k=1}^{M} \sfe^{\sfo_K(1) \nu(y_k)- a\surcharge_t(y_k)}\\
			&\leq \sfe^{K}\sfe^{-\epsilon(1-a)\nu(x)} \sum_{M\geq 1} \sum_{\substack{y_1, \dots, y_M\in \Delta_K^c\\ \sum y_i \in \Delta_K(x)}} \exp\biggl(-(a+\sfo_K(1)) \sum_{k=1}^{M} \nu(y_k) + a\nu\Bigl(\sum_{k=1}^M y_k\Bigr)\biggr)\\
			&\leq \sfe^{(a+1)K}\sfe^{-(\epsilon(1-a)-a)\nu(x)} \sum_{M\geq 1} \sum_{\substack{y_1, \dots, y_M\in \Delta_K^c\\ \sum y_i \in \Delta_K(x)}} \exp\Bigl(-(a+\sfo_K(1)) \sum_{k=1}^{M} \nu(y_k) \Bigr),
		\end{split}
	\end{multline*}
	where in the first line we used \(\sfe^{t\cdot (x-v_M)} \leq \sfe^K\) (since \(x\in \Delta_K(v_M)\)) and~\eqref{eq:SurchargeIncrements} and we used ~\eqref{eq:sum_of_surcharges_LB} in the fourth line. 
	Choosing \(a=\frac{\epsilon}{3+3\epsilon}\) and \(K_0(\epsilon)\) large enough gives the desired claim, since
	\[
		\sum_{y\notin\Delta_K}\sfe^{-(a+\sfo_K(1))\nu(y)} \leq \sum_{y\notin \Delta_K} \sfe^{-\epsilon \nu(y)/(4+4\epsilon)}\leq \sfe^{-c_{\epsilon} K},
	\]
	provided \(K\) is large enough.
\end{proof}

\pagebreak % <<<<<<<<<<<<<<<<<<<<<<<<<<<<<<<<<<< PAGE BREAK

\subsection{Skeletons do not contain many long edges}\label{subsec:not_many_edges}

Up to now, the analysis was almost the same as in~\cite{Campanino+Ioffe+Velenik-2003}. But we only proved facts which are compatible with the ``one-jump'' scenario of the regime where the massgap assumption fails, as mentioned in the introduction. Our assumption NSA becomes crucial here to ensure that a typical skeleton does not possess a giant edge (or several long edges, see~\eqref{eq:lem_long_edge_contrib:eq1} below). Without NSA, our analysis does not exclude the possibility of a typical skeleton containing one (or several) giant edges in a direction where saturation occurs (for example, in the direction \(s'\) in Fig.~\ref{fig:noNSA}). We turn to the main novelty of the skeleton analysis.

For \(y\notin\Delta_{K}\), define
\begin{gather*}
	\rmr(y) = \min\setof{ r\in \Z_{>0}}{y\in \Delta_{K+r}\setminus\Delta_{K+r-1}},\\
	\rmr(V) = \sum_{k=1}^{M} \rmr(v_k-v_{k-1}).
\end{gather*}

\begin{lemma}
	\label{lem:long_edges_have_small_contribution}
	For any nice \(U\subset\calT\) and \(\epsilon>0\), there exist \(a>0\) and \(K_0 \geq 0\) such that, for any \(K\geq K_0, t\in U\) and \(x\in\Zd\) with \(\norm{x}\) large enough,
	\[
		\sfe^{t\cdot x} \sum_{V\in\SkeletonSet} Q_K(V) \mathds{1}_{\{\rmr(V)\geq \epsilon\nu(x)\}} \mathds{1}_{\{v_M\in \Delta_K(x)\}} \leq \sfe^{-a\nu(x)}.
	\]
\end{lemma}
\begin{proof}
	Fix \(\epsilon>0\). We use \(R= \epsilon K/6\). We will treat separately the contribution to \(\rmr(V)\) of small steps and of large steps. Let us introduce
	\[
		\rmr_-(V) = \sum_{(u,v)\in\shortEdge(V)} \rmr(v-u),\qquad \rmr_+(V) = \sum_{(u,v)\in\longEdge(V)} \rmr(v-u).
	\]
	As \(\rmr(V) = \rmr_-(V)+\rmr_+(V)\), \(\rmr(V)\geq \epsilon \nu(x)\) implies that \(\rmr_*(V)\geq \epsilon \nu(x)/2\) for at least one of \(*\in\{-,+\}\). Moreover, we have
	\[
		\rmr_-(V)\leq \frac{\epsilon}{6} K |\shortEdge(V)| \leq \frac{\epsilon}{6} K |V|.
	\]
	Therefore, from our control on \(|V|\) (Lemma~\ref{lem:skeleton_size_upper_bound}), for any \(t\in\partial\calW_{\nu}\) and \(x\in \Zd\) with \(\norm{x}\) large,
	\[
		\sfe^{t\cdot x} \sum_{V\in\SkeletonSet} Q_K(V) \mathds{1}_{\{\rmr_-(V)\geq \epsilon \nu(x)/2\}}
		\leq
		\sfe^{\nu(x)} \sum_{V\in\SkeletonSet} Q_K(V) \mathds{1}_{\{|V|\geq 3 \nu(x)/K\}}
		\leq
		\sfe^{-\nu(x)},
	\]
	once \(K\geq K_0\) for some \(K_0\) large enough.
	
	We now turn to the case where the contribution comes from \(\rmr_+\). Again from our control on \(|V|\), we can suppose \(|V|\leq 2\nu(x)/K\) up to an error \(\sfe^{-\nu(x)/2}\). We will use an improved bound on \(F\): we claim that, for any \(t\in U\), \(K\geq K_0(\epsilon,U)\) and \(y\notin \Delta_{(1+\epsilon/6)K}\), one has the bound,
	\begin{equation}
		\label{eq:lem_long_edge_contrib:eq1}
		\sfe^{t\cdot y} F(y) \leq \sfe^{-m \rmr(y)/2 } \sfe^{-m\epsilon K/18},
	\end{equation}
	where \(m=m(U)\) is the mass gap given by Lemma~\ref{lem:mass_gap_long_edges}. Indeed, using the definition of \(F\) (see below~\eqref{eq:skeleton_energy_bound}),
	\begin{align*}
		\sfe^{t\cdot y} F(y)
		&\leq \sum_{x\in\Delta_K} \sfe^{-\surcharge_t(x) + \sfo(\nu(x))} CJ_{y-x}\sfe^{t\cdot (y-x)}\\
		&\leq C\sum_{x\in\Delta_K} \sfe^{-\surcharge_t(x) + \sfo(\nu(x))} \sfe^{-m\nu(y-x)}\\
		&\leq C \sfe^{\sfo(K)} \sfe^{-m\epsilon K/12} \sfe^{-m\rmr(y)/2}
		\leq \sfe^{-m\epsilon K/18} \sfe^{-m\rmr(y)/2},
	\end{align*}
	where we used Lemma~\ref{lem:mass_gap_long_edges} in the second line and \(\surcharge_t\geq 0\), \(y\notin \Delta_{(1+\epsilon/6)K}\) (recall \(R=\epsilon K/6\)) and \(K\geq K_0\) in the third. Equipped with the bound on \(F\) and using the energy bound~\eqref{eq:skeleton_energy_bound}, one obtains
\pagebreak % <<<<<<<<<<<<<<<<<<<<<<<<<<<<<<<<<<< PAGE BREAK
	\begin{multline*}
		\sfe^{t\cdot x}\sum_{V\in\SkeletonSet} Q_K(V) \mathds{1}_{\{|V|\leq 2\nu(x)/K\}} \mathds{1}_{\{\rmr_+(V)\geq \epsilon\nu(x)/2\}} \mathds{1}_{\{v_M\in \Delta_K(x)\}} \\
		\begin{split}
			&\leq \sum_{M=1}^{2 \nu(x)/K} \sum_{I\subset \{1, \dots, M\}} \sum_{y_i, i\notin I} \sum_{z_i, i\in I} \prod_{i\notin I} \sfe^{-\surcharge_t(y_i)+ \sfo(K)} \prod_{i\in I} \sfe^{t\cdot z_i}F(z_i) \mathds{1}_{\left\{\sum_{i\in I} \rmr(z_i) \geq \epsilon\nu(x)/2\right\}} \\
			&\leq \sum_{M=1}^{2 \nu(x)/K} \sum_{I\subset \{1, \dots, M\}} \sum_{y_i, i\notin I} \sum_{z_i, i\in I} \prod_{i\notin I} \sfe^{-\surcharge_t(y_i)+ \sfo(K)} \prod_{i\in I} \sfe^{-m\varepsilon K/18 - m\rmr(z_{i})/2} \mathds{1}_{\left\{\sum_{i\in I} \rmr(z_i) \geq \epsilon\nu(x)/2\right\}} \\
			&\leq \sfe^{-m \epsilon\nu(x)/8 } \sum_{M=1}^{2 \nu(x)/K} \sum_{I\subset \{1, \dots, M\}} \sum_{z_i, i\in I} \sfe^{\sfo(K)(M-|I|)} \sfe^{-m\epsilon K|I|/18} \prod_{i\in I} \sfe^{-m\rmr(z_i)/4}\\
			&\leq \sfe^{-m \epsilon\nu(x)/8 } \sum_{M=1}^{2 \nu(x)/K} \sum_{I\subset \{1, \dots, M\}} \sfe^{\sfo(K)M} \sfe^{-m\epsilon K|I|/18},
		\end{split}
	\end{multline*}
	where the sums are over \(y_i\in \Delta_{K+R}\setminus \Delta_K\) and \(z_i\in\Zd\setminus\Delta_{K+R}\), we used~\eqref{eq:SurchargeIncrements} in the second line and we used~\eqref{eq:lem_long_edge_contrib:eq1} in the third line and \(\sum_{i\in I} \rmr(z_i) \geq \epsilon\nu(x)/2\) in the fourth line. Taking \(K_0(\epsilon,m)\) large enough implies the desired claim with \(a=\min(1/2,m\epsilon/9)\).
\end{proof}
From Lemma~\ref{lem:long_edges_have_small_contribution}, one can deduce a converse to Lemma~\ref{lem:skeleton_size_upper_bound}: a lower bound on the ``typical'' value of \(|V|\).
\begin{lemma}
	\label{lem:skeleton_size_lower_bound}
	For any nice \(U\subset \calT\) and \(1>\epsilon>0\), there exist \(a>0\) and \(K_0 \geq 0\) such that, for any \(K\geq K_0, t\in U\) and \(x\in\Zd\) with \(\norm{x}\) large enough,
	\[
		\sfe^{t\cdot x} \sum_{V\in\SkeletonSet} Q_K(V) \mathds{1}_{\{|V|\leq (1- \epsilon)\nu(x)/K\}} \mathds{1}_{\{v_M\in \Delta_K(x)\}} \leq \sfe^{-a\nu(x)}.
	\]
\end{lemma}
\begin{proof}
	For any skeleton \(V = (0, \dots, v_M)\) with \(v_M\in\Delta_K(x)\), one has
	\[
		MK + \sum_{k=1}^{M} \rmr(v_k-v_{k-1}) \geq \nu(x)-K.
	\]
	So, if \(M < (1-\epsilon)\nu(x)/K\), \(\rmr(V)\geq \epsilon \nu(x)-K\).
\end{proof}

\subsection{Skeletons contain many cone-points}

We can now turn to the proof of Theorem~\ref{thm:skeleton_CP}. As in~\cite{Campanino+Ioffe+Velenik-2003}, the idea is to prove a \emph{deterministic} statement about skeletons. Namely, that if a skeleton \(V\) has few cone-points, then it must have either a large surcharge or small size \(|V|\).

Let \(\delta>0\) and \(t\in\partial\calW_{\nu}\). We say that \(v_k\in V\) is \emph{\((t,\delta)\)-forward-bad} if there is \(v_{l}\in V\) with \(l>k\) and \(v_l\notin\fcone_{t,\delta}(v_k)\). Similarly, \(v_k\) is \emph{\((t,\delta)\)-backward-bad} if there is \(v_{l}\in V\) with \(l<k\) and \(v_l\notin \bcone_{t,\delta}(v_k)\). Define
\begin{gather*}
	\Bad_{t,\delta}^{\blacktriangleleft}(V) = \bsetof{k\in\{0, 1, \dots, M}{v_k \text{ is }(t,\delta)\text{-forward-bad}},\\
	\Bad_{t,\delta}^{\blacktriangleright}(V) = \bsetof{k\in\{0, 1, \dots, M}{v_k \text{ is }(t,\delta)\text{-backward-bad}}.
\end{gather*}
For \(v_k\in V\), one has \(v_k\notin\CPts_{t,\delta}(V)\) if and only if \(k\in\Bad_{t,\delta}^{\blacktriangleleft}(V) \cup \Bad_{t,\delta}^{\blacktriangleright}(V) = \Bad_{t,\delta}(V)\).
\begin{lemma}
	\label{lem:skeleton_CP_controlled_by_surcharge}
	Let \(1/2>\delta>0\), \(\epsilon\in(0,1)\) and \(t\in\partial\calW_{\nu}\). Let \(K\geq 0\). Let \(x\in \fcone_{t,\delta}\) with \(\norm{x}\) large. Let \(V=(v_0=0, v_1, \dots, v_M)\) be a \(K\)-skeleton with \(v_M\in\Delta_K(x)\). Then, if \(|\CPts_{t,\delta}(V)|\leq (1-\epsilon)\nu(x)/K\), at least one of the following holds:
	\begin{itemize}
		\item \(|V|< (1-\frac{\epsilon}{2}) \nu(x)/K\),
		\item \(\surcharge_{t}(V)\geq \frac{\delta \epsilon}{32} \nu(x)\).
	\end{itemize}
\end{lemma}
\begin{proof}
	We first extract the ``bad excursions'' from a skeleton \(V\) containing bad points. For a skeleton \(V=(v_0, \dots, v_M)\), the sequence \(\bigl((i_1^{\blacktriangleleft}, j_1^{\blacktriangleleft}), \dots, (i_m^{\blacktriangleleft}, j_m^{\blacktriangleleft})\bigr)\) is constructed using the following algorithm.
	
	\begin{algorithm}[H]
		\label{algo:skeleton_forward_bad_extraction}
		\caption{Forward bad extraction algorithm; see Fig.~\ref{fig:ExtractionForwardBad}}
		\KwInput{a skeleton \(V=(v_0, \dots, v_M)\)}
		Set \(k = 1\);\\
		Define \(i_1^{\blacktriangleleft} = \min\setof{i\geq 0}{i\in \Bad_{t,\delta}^{\blacktriangleleft}(V)}\);\\
		Define \(j_1^{\blacktriangleleft} = \min\setof{j>i_1^{\blacktriangleleft}}{v_j \notin \fcone_{t,\delta}(v_{i_1^{\blacktriangleleft}})}\);\\
		\While{\(\exists i\geq j_{k}^{\blacktriangleleft}\) such that \(i\in \Bad_{t,\delta}^{\blacktriangleleft}(V)\)}{
			Update \(k=k+1\);\\
			Define \(i_k^{\blacktriangleleft} = \min\setof{i\geq j_{k-1}^{\blacktriangleleft}}{i\in \Bad_{t,\delta}^{\blacktriangleleft}(V)}\);\\
			Define \(j_k^{\blacktriangleleft} = \min\setof{j>i_k^{\blacktriangleleft}}{v_j \notin \fcone_{t,\delta}(v_{i_k^{\blacktriangleleft}})}\);\\
		}
		\KwOutput{the sequence \(\bigl((i_1^{\blacktriangleleft}, j_1^{\blacktriangleleft}), \dots, (i_m^{\blacktriangleleft}, j_m^{\blacktriangleleft})\bigr)\)}
	\end{algorithm}
	\begin{figure}[ht]
		\includegraphics{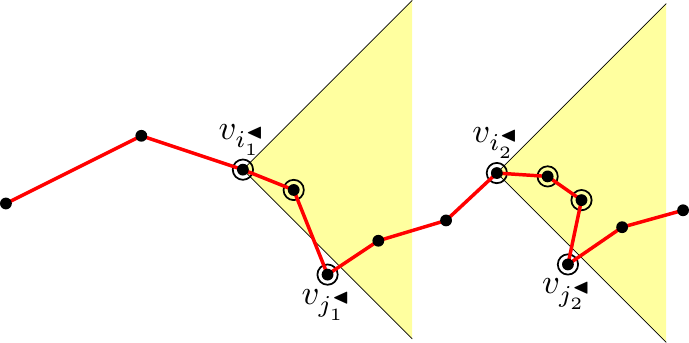}
		\caption{The forward-bad vertices of the skeleton are indicated by circled vertices. Two particularly relevant translates of the forward cone \(\fcone_{t,\delta}\) are also shown.}
		\label{fig:ExtractionForwardBad}
	\end{figure}

	Similarly, the sequence \(\bigl((i_1^{\blacktriangleright}, j_1^{\blacktriangleright}), \dots, (i_{m'}^{\blacktriangleright}, j_{m'}^{\blacktriangleright})\bigr)\) is constructed using the following algorithm.

	\begin{algorithm}[H]
		\label{algo:skeleton_backward_bad_extraction}
		\caption{Backward bad extraction algorithm}
		\KwInput{a skeleton \(V=(v_0, \dots, v_M)\)}
		Set \(k = 1\);\\
		Define \(i_1^{\blacktriangleright} = \max\setof{i\geq 0}{i\in \Bad_{t,\delta}^{\blacktriangleright}(V)}\);\\
		Define \(j_1^{\blacktriangleright} = \max\setof{0\leq j<i_1^{\blacktriangleright}}{v_j \notin \bcone_{t,\delta}(v_{i_1^{\blacktriangleright}})}\);\\
		\While{\(\exists 0\leq i\leq j_{k}^{\blacktriangleright}\) such that \(i\in \Bad_{t,\delta}^{\blacktriangleright}(V)\)}{
			Update \(k=k+1\);\\
			Define \(i_k^{\blacktriangleright} = \max\setof{i\leq j_{k-1}^{\blacktriangleright}}{i\in \Bad_{t,\delta}^{\blacktriangleright}(V)}\);\\
			Define \(j_k^{\blacktriangleright} = \max\setof{j<i_k^{\blacktriangleright}}{v_j \notin \bcone_{t,\delta}(v_{i_k^{\blacktriangleright}})}\);\\
		}
		\KwOutput{the sequence \(\bigl((i_1^{\blacktriangleright}, j_1^{\blacktriangleright}), \dots, (i_{m'}^{\blacktriangleright}, j_{m'}^{\blacktriangleright})\bigr)\)}
	\end{algorithm}

	Note that the sequence \((i_1^{\blacktriangleleft}, j_1^{\blacktriangleleft}, \dots, i_m^{\blacktriangleleft},j_m^{\blacktriangleleft})\) is non-decreasing while the sequence \((i_1^{\blacktriangleright}, j_1^{\blacktriangleright},\dots, i_{m'}^{\blacktriangleright},j_{m'}^{\blacktriangleright})\) is non-increasing.
	We will write \(\calM^{\blacktriangleleft} = \bigcup_{k=1}^{m} \{i_k^{\blacktriangleleft}, i_k^{\blacktriangleleft}+1, \dots, j_k^{\blacktriangleleft}\}\) and \(\calM^{\blacktriangleright} = \bigcup_{k=1}^{m'} \{i_k^{\blacktriangleright}, i_k^{\blacktriangleright}-1, \dots, j_k^{\blacktriangleright}\}\) and call \(\calM = \calM^{\blacktriangleleft}\cup \calM^{\blacktriangleright}\) the set of \emph{marked points}.
	By construction, one has the inclusion \(\Bad_{t,\delta}\subset \calM\).
	\begin{claim}\label{claim1}
		Let \(a\geq 0\). Suppose \(|\calM|\geq a \nu(x)/K\). Then, \(\surcharge_{t}(V)\geq \frac{\delta a}{32} \nu(x)\).
	\end{claim}
	Let us first use Claim~\ref{claim1} to conclude the proof of Lemma~\ref{lem:skeleton_CP_controlled_by_surcharge}.
	On the one hand, if \(|V|\geq (1-\frac{\epsilon}{2})\nu(x)/K\) and \(|\calM|\leq \frac{\epsilon}{2} \nu(x)/K\), then
	\[
		|\CPts_{t,\delta}(V)| = |V|-|\Bad_{t,\delta}(V)| \geq |V|-|\calM| \geq (1-\epsilon)\nu(x)/K.
	\]
	On the other hand, if \(|\calM| > \frac{\epsilon}{2} \nu(x)/K\), then \(\surcharge_{t}(V)\geq \frac{\delta \epsilon}{64} \nu(x)\) by Claim~\ref{claim1}.
	
	\begin{proof}[Proof of Claim~\ref{claim1}]
		We conduct the proof with \(\calM^{\blacktriangleleft}\) replacing \(\calM\). The same proof applies for \(\calM^{\blacktriangleright}\) and the claim for \(\calM\) follows from the combination of the two. Let \(a\geq 0\) and suppose \(|\calM^{\blacktriangleleft}|\geq a \nu(x)/2K\). We then claim that, for every \(k = 1, \dots, m\),
		\begin{equation}
			\label{eq:lem:skeleton_CP_controlled_by_surcharge:eq1}
			\sum_{i=i_k^{\blacktriangleleft}}^{j_k^{\blacktriangleleft}-1} \surcharge_t(v_{i+1}-v_i) \geq \frac{\delta}{8}K(j_k^{\blacktriangleleft} - i_k^{\blacktriangleleft}).
		\end{equation}
		Let us first observe that~\eqref{eq:lem:skeleton_CP_controlled_by_surcharge:eq1} implies the claim:
		\[
			\surcharge_{t}(V)\geq \sum_{k=1}^{m} \sum_{i=i_{k}^{\blacktriangleleft}}^{j_{k}^{\blacktriangleleft}-1}\surcharge_{t}(v_{i+1}-v_i) \geq \frac{\delta}{8}K\sum_{k=1}^m (j_k^{\blacktriangleleft}- i_k^{\blacktriangleleft}) \geq \frac{\delta}{16}K|\calM^{\blacktriangleleft}| \geq \frac{\delta a}{32} \nu(x).
		\]

		\begin{figure}[t]
			\includegraphics{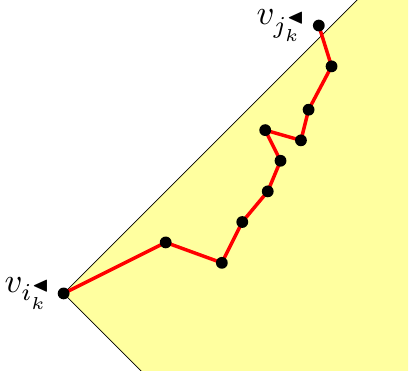}
			\hspace{1cm}
			\includegraphics{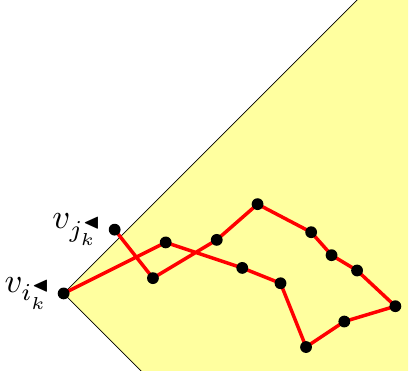}
			\hspace{1cm}
			\includegraphics{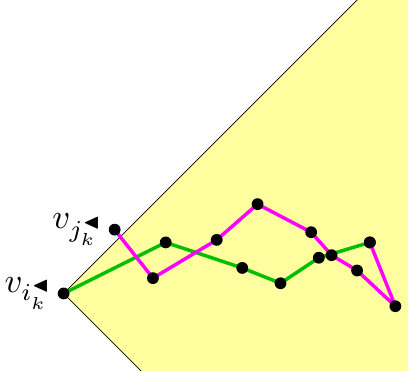}
			\caption{Proof of~\eqref{eq:lem:skeleton_CP_controlled_by_surcharge:eq1}. Left: \(t\cdot (v_{j_k^{\blacktriangleleft}}-v_{i_k^{\blacktriangleleft}}) \geq \frac{K}{8}(j_k^{\blacktriangleleft}- i_k^{\blacktriangleleft})\). Middle: \(t\cdot (v_{j_k^{\blacktriangleleft}}-v_{i_k^{\blacktriangleleft}}) < \frac{K}{8}(j_k^{\blacktriangleleft}- i_k^{\blacktriangleleft})\). Right: same as before, but with the increments reordered.}
			\label{fig:MarkedCase}
		\end{figure}
	
		Let us now prove~\eqref{eq:lem:skeleton_CP_controlled_by_surcharge:eq1} (see Fig.~\ref{fig:MarkedCase}). Consider first the case \(t\cdot (v_{j_k^{\blacktriangleleft}}-v_{i_k^{\blacktriangleleft}}) \geq \frac{K}{8}(j_k^{\blacktriangleleft}- i_k^{\blacktriangleleft})\). Then,
		\[
			\sum_{i=i_k^{\blacktriangleleft}}^{j_k^{\blacktriangleleft}-1} \surcharge_t(v_{i+1}-v_i)
			\geq \surcharge_t(v_{j_k^{\blacktriangleleft}} - v_{i_k^{\blacktriangleleft}}) 
			\geq \delta\nu(v_{j_k^{\blacktriangleleft}}-v_{i_k^{\blacktriangleleft}})
			\geq \delta t\cdot (v_{j_k^{\blacktriangleleft}} - v_{i_k^{\blacktriangleleft}})
			\geq \frac{\delta}{8}K(j_k^{\blacktriangleleft}- i_k^{\blacktriangleleft}),
		\]
		where the first inequality follows from the triangular inequality, the second one from \(v_{j_k^{\blacktriangleleft}}\notin \fcone_{t,\delta}(v_{i_k^{\blacktriangleleft}})\), the third one from the definition of \(\calW_{\nu}\), and the last one from our initial assumption.
		
		Let us now turn to the case \(t\cdot (v_{j_k^{\blacktriangleleft}}-v_{i_k^{\blacktriangleleft}}) < \frac{K}{8}(j_k^{\blacktriangleleft}- i_k^{\blacktriangleleft})\).
		Let us write \(y_i = v_{i+1}-v_i\). Clearly,
		\[
			\sum_{i= i_k^{\blacktriangleleft} }^{j_k^{\blacktriangleleft}-1} \surcharge_t(v_{i+1}-v_i) = \sum_{i= i_k^{\blacktriangleleft} }^{j_k^{\blacktriangleleft}-1} \surcharge_t(y_i)
		\]
		is invariant under permutations of the \(y_i\)s. Up to a re-ordering and a re-labeling, we can therefore assume that the sequence is \((y_1, \dots, y_{l_1}, y_{l_1+1}, \dots, y_{l_1+l_2})\) with \(l_1+l_2 = j_k^{\blacktriangleleft} -  i_k^{\blacktriangleleft}\), \(y_{1}, \dots, y_{l_1}\in \fcone_{t,\delta}\) and \(y_{l_1+1}, \dots, y_{l_1+l_2}\notin \fcone_{t,\delta} \).
		Consider first the case \(l_1 < l_2 \). In that case, there are \( l_2 \geq \frac{j_k^{\blacktriangleleft} -  i_k^{\blacktriangleleft}}{2}\) increments not in \(\fcone_{t,\delta}\). Their contribution to the sum is then at least \(\sum_{i=1}^{l_2} \delta \nu(y_{l_1+i}) \geq \frac{\delta K}{2} (j_k^{\blacktriangleleft} -  i_k^{\blacktriangleleft})\).
		Finally, consider \(l_1\geq l_2\). Set \(z_1 = \sum_{i=1}^{l_1} y_i\) and \(z_2 = \sum_{i=l_1+1}^{l_1+l_2} y_i\). Then, by definition of \(\fcone_{t,\delta}\) and the condition \(\delta < \frac12\),
		\[
			t\cdot z_1 \geq \sum_{i=1}^{l_1} (1-\delta) \nu(y_i) \geq \frac{1}{2} K l_1.
		\]
		In particular,
		\[
			- t\cdot z_2 = t\cdot z_1 - t\cdot (z_1+z_2) \geq \frac{1}{2}Kl_1 - \frac{K}{8}(l_1+l_2) \geq \frac{K}{8}(j_k^{\blacktriangleleft}- i_k^{\blacktriangleleft}),
		\]
		where we have used our initial assumption \(t\cdot (z_1+z_2) = t\cdot (v_{j_k^{\blacktriangleleft}} - v_{i_k^{\blacktriangleleft}}) < \frac{K}{8}(j_k^{\blacktriangleleft} - i_k^{\blacktriangleleft})\) and \(l_1\geq l_2\). So, using the triangular inequality and \(\surcharge_t\geq 0\),
		\[
			\sum_{i = i_k^{\blacktriangleleft}}^{j_k^{\blacktriangleleft}-1} \surcharge_t(y_i) \geq \surcharge_t(z_1) + \surcharge_t(z_2) \geq \surcharge_t(z_2) \geq - t\cdot z_2 \geq \frac{K}{8}(j_k^{\blacktriangleleft}- i_k^{\blacktriangleleft}). \qedhere
		\]
	\end{proof}
\end{proof}

\subsection{Proof of the main Theorems}

We conclude by deducing Theorems~\ref{thm:skeleton_CP} and~\ref{thm:skeleton_CP_well_distributed}.
\begin{proof}[Proof of Theorem~\ref{thm:skeleton_CP}]
	Note that \(\CPts_{t,\delta}(V)\) is non-decreasing in \(\delta\). Lemma~\ref{lem:skeleton_CP_controlled_by_surcharge} together with the estimates on \(|V|\) (Lemma~\ref{lem:skeleton_size_lower_bound}) and on the surcharge (Lemma~\ref{lem:skeleton_small_surcharge}) thus yield the desired claim.
\end{proof}

\begin{proof}[Proof of Theorem~\ref{thm:skeleton_CP_well_distributed}]
	Fix \(U,\delta,\epsilon,l\) as in the statement of Theorem~\ref{thm:skeleton_CP_well_distributed}. Let \(K\geq 0\) and
	\[
		H_k = \slab_{(k-1)lr_*K,klr_*K}(\hat{x}).
	\]
	Define
	\begin{gather*}
		A_1 = \setof{1\leq k \leq \nu(x)/lr_*K}{V\cap H_k \neq\emptyset},\\
		A_2 = \setof{1\leq k \leq \nu(x)/lr_*K}{\CPts_{t,\delta}(V)\cap H_k \neq \emptyset}.
	\end{gather*}
	Now, we claim that the occurrence of \(|A_2|\leq (1-\epsilon)\nu(x)/lr_*K\) implies that at least one of the following statements is true:
	\begin{itemize}
		\item \(|A_1|\leq (1-\frac{\epsilon}{2})\nu(x)/lr_*K\),
		\item \(|V|\geq (1+\frac{\epsilon}{4lr_*})\nu(x)/K\),
		\item \(|\CPts_{t,\delta}(V)|\leq (1-\frac{\epsilon}{4lr_*})\nu(x)/K\).
	\end{itemize}
	Indeed, suppose that the first two points do not hold. Then
	\begin{align*}
	|\CPts_{t,\delta}(V)|\leq |V|-|A_1 \setminus A_2|
			&\leq (1+\frac{\epsilon}{4lr_*})\nu(x)/K - (1-\frac{\epsilon}{2})\nu(x)/lr_*K + (1-\epsilon)\nu(x)/lr_*K\\
			&= (1-\frac{\epsilon}{4lr_*})\nu(x)/K.
	\end{align*}
	To conclude the proof, we bound the probability of the statement in the second bullet using Lemma~\ref{lem:skeleton_size_upper_bound} and the probability of the statement in the third bullet using Theorem~\ref{thm:skeleton_CP}.
	
	To bound the probability of the first statement, observe that all slabs \(H_k\) containing no vertices of \(V\) must be crossed by at least one edge of the skeleton. Of course, several consecutive such slabs can be crossed simultaneously by a single sufficiently long edge, but in any case, one must have a contribution to \(\rmr(V)\) of at least \((l-1) r_*K\) for each \(k\) such that \(V\cap H_k = \emptyset\). So, \(|A_1|\leq (1-\frac{\epsilon}{2})\nu(x)/lr_*K\) implies \(\rmr(V)\geq \epsilon \frac{(l-1)}{2l}\nu(x)\). 
	
	An application of Lemma~\ref{lem:long_edges_have_small_contribution} thus finishes the proof.
\end{proof}

%%%%%%%%%%%%%%%%%%%%%%%%%%%%%%%%%%%%%%%%%%%%%%%%%%%%%%%%%%%%%%%%%%%%%%%%%%%%%%%%%%%%%%%%%%%%%%%%%%%%%%%%
\section{Path analysis}
\label{sec:path_analysis}
%%%%%%%%%%%%%%%%%%%%%%%%%%%%%%%%%%%%%%%%%%%%%%%%%%%%%%%%%%%%%%%%%%%%%%%%%%%%%%%%%%%%%%%%%%%%%%%%%%%%%%%%

\subsection{Density of cone-points}

Our first goal is to use our control on typical skeletons to deduce a similar one on paths. Namely, this subsection is devoted to the proof of the following result.
\begin{theorem}[Typical paths contain many cone-points]
	\label{thm:path_cone_points}
	Let \(s_0\in\bbS^{d-1}\) be such that NSA holds. Let \(t_0\) be \(\nu\)-dual to \(s_0\). Then, there exists \(\epsilon_0>0\) such that
	\begin{itemize}
		\item \(W_0=\bigcup_{s\in U_0} T_s\) is a nice set, where \(U_0= \bbS^{d-1}\cap \bbB_{\epsilon_0}(s_0)\);
		\item there exist \(a>0, \delta_0 < 1, n_0\geq 0, p>0\) such that, for any \(t\in W_0\), \(x\in\fcone_{t_0,\delta_{0}}\) with \(\norm{x}\geq n_0\) and \(\tilde{\gamma}\) with \(\tilde{\gamma}_{|\tilde{\gamma}|} = 0\),
		\[
			\sfe^{t\cdot x} \sum_{\gamma:\, 0\to x} q(\gamma\given \tilde{\gamma}) \mathds{1}_{\{|\CPts_{t_0,\delta_0}(\gamma)|\leq p\norm{x}\}}
			\leq \sfe^{-a \nu(x)}.
		\]
	\end{itemize}
	Moreover, \(U_0\subset \fcone_{t_0,\delta_0}\).
\end{theorem}
\begin{proof}
	Let \(s_0, t_0\) be as in the statement. Let \(R_0, R_1\) be given by Lemma~\ref{app:irr_interaction:lem:connections_local}. Choose \(K>0\) large enough to be able to apply Theorems~\ref{thm:skeleton_CP} and~\ref{thm:skeleton_CP_well_distributed}. The coarse-graining algorithm \(\Skeleton\) is performed using that scale \(K\).
%	To simplify notations, we write \(\slab_{a,b}\equiv \slab_{a,b}(\hat{x})\).
	By Lemma~\ref{lem:skeleton_energy_bound}, the push-forward of \(q(\cdot \given \tilde{\gamma})\) by the coarse-graining procedure \(\Skeleton\) is a skeleton measure satisfying~\eqref{eq:skeleton_energy_bound} (to simplify notations, we present the proof with \(\tilde{\gamma} = (0)\)).

	Let \(1>\delta_1>0\) be such that there exists \(z\in\fcone_{t_0,\delta_1}\) with \(J_z^{R_0}>0\) (such \(\delta_1\) exists by Lemma~\ref{app:irr_interaction:lem:angular_opening_connected}). Let \(\delta_2 = (1-\delta_1)/4\). The first part of the statement follows from the discussion at the end of Section~\ref{sec:SaturationConditions}.%
	We shall assume that \(\epsilon_0\) is small enough and that \(\delta_1\) is close enough to \(1\) to have \(U_0\subset \fcone_{t_0,\delta_1}\).
		
%	We say that \(z\in \Zd\) is a \emph{\(L\)-pre-cone point of \(\gamma\)} if \(z\in \gamma\) and \(\gamma\subset \bcone_{t_0,\delta_1/2}(z)\cup \bbB_{L}(z) \cup\fcone_{t_0,\delta_1/2}(z) \). Note that for any \(C\) fixed, there exists \(C'\) large enough (independent of \(K\)), such that any point of \(\gamma\) at Euclidean distance at most \(CK\) from a \((t_0,\delta_1/2)\)-cone-point of \(\Skeleton(\gamma)\) is a \(C'K\)-pre-cone point of \(\gamma\).
	
	We first introduce the notion of a \emph{canonical path} (see Fig.~\ref{fig:CanonicalPath}). Let \(u\in \fcone_{t_0,\delta_1}\) be such that \(\normsup{u}<R_0\) and \(J_u>0\). Let \(l_u\) be the smallest integer such that, for any \(z,z'\in \fcone_{t_0,\delta_1}\) with \(\norm{z}, \norm{z'}\geq l_u\norm{u}\), one has the inclusion (see Fig.~\ref{fig:Bzzprime})
	\[
		B_{z,z'} = \bigcup_{\lambda\in [0, 1]} \Bigl(\lambda z + (1-\lambda)z' + [-1-R_1, 1+R_1]^d \Bigr) \subset \fcone_{t_0,\delta_1+\delta_2}.
	\]
	\begin{figure}[ht]
		\includegraphics{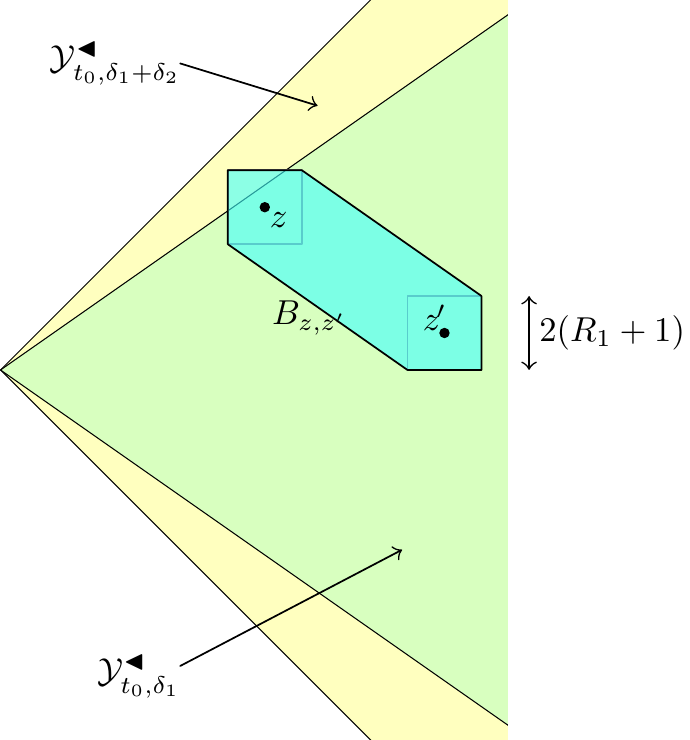}
		\caption{The condition on the points \(z\) and \(z'\).}
		\label{fig:Bzzprime}
	\end{figure}
	By Lemma~\ref{app:irr_interaction:lem:connections_local}, for any \(z,z'\) as above, there exists a path \(\eta'_{z,z'}\in\pathSet(z,z')\) using only edges \(\{i, j\}\) with \(\normsup{i-j} < R_0\) and \(i, j\in B_{z,z'}\cap\Zd\). We fix such a path for each \(z, z'\) (denoted \(\eta'_{z,z'}\)).
	\begin{figure}[ht]
		\includegraphics{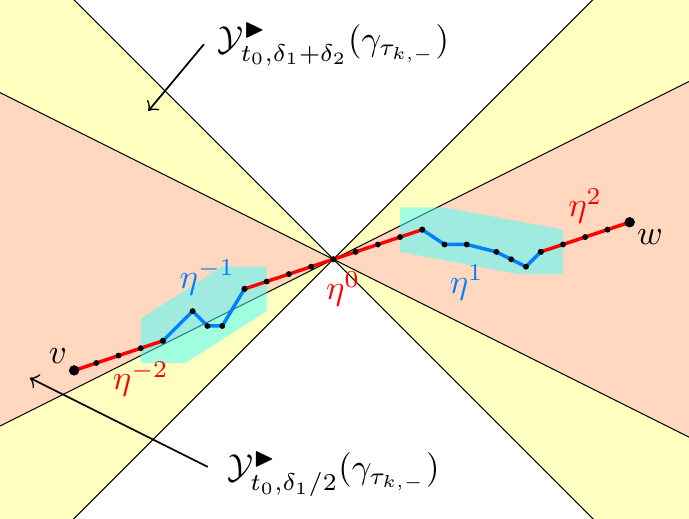}
		\caption{Construction of the canonical path between \(v\) and \(w\).}
		\label{fig:CanonicalPath}
	\end{figure}
	
	For \(v\in\bcone_{t_0,\delta_1}\), \(w\in\fcone_{t_0,\delta_1}\) satisfying \(\norm{v}, \norm{w} \geq 2 l_u \norm{u}\), define the \emph{canonical path from \(v\) to \(w\)}, \(\eta(v,w)\) as follows (See Fig.~\ref{fig:CanonicalPath}):
	\begin{itemize}
		\item Let \(\eta^{-2} = (v, v+u, \dots, v + l_u u)\).
		\item Let \(\eta^{-1}: v + l_u u \to -l_u u\) be the central reflection (and time reversal) of \(\eta'_{l_u u, -v - l_u u}\).
		\item Let \(\eta^0 = (-l_u u, (-l_u+1)u, \dots, (l_u-1)u, l_u u)\).
		\item Let \(\eta^1: l_u u\to w - l_u u\) be equal to \(\eta'_{l_u u, w - l_u u}\).
		\item Let \(\eta^{2} = (w - l_u u, w - (l_u-1)u, \dots, w)\).
		\item Define \(\eta(v,w)\) as the loop erasure of \(\eta^{-2}\concatenate\eta^{-1}\concatenate\eta^{0}\concatenate\eta^1\concatenate\eta^2\).
	\end{itemize}
\pagebreak % <<<<<<<<<<<<<<<<<<<<<<<<<<<<<<<<<<< PAGE BREAK
	By construction, the canonical path \(\eta_{v,w}\) enjoys the following properties:
	\begin{itemize}
		\item \(\eta(v,w)\) is self-avoiding and included in \(\diam_{t_0,\delta_1+\delta_2}(v,w)\);
		\item \(\eta(v,w)\) uses only edges with sup-norm \(\leq R_0\);
		\item \(0\) is a \((t_0,\delta_1+\delta_2)\) cone-point of \(\eta(v,w)\);
		\item \(|\eta(v,w)|\leq 4l_u + (2R_1+1)^d(\norm{v}+\norm{w})\leq C\norm{v-w}\).
	\end{itemize}
	
	Let \(C_1<C_2\) be large integers (independent of \(K\); the precise requirement will appear in a few lines). Let us introduce the slabs (see Fig.~\ref{fig:preGood})
	\begin{gather*}
		H_k = \slab_{(k-1) C_3, k C_3 }(\hat{t}_0/r_*K),\\
		H_{k,1} = \slab_{(k-1) C_3 + C_2, (k-1) C_3 + C_2 + 2}(\hat{t}_0/r_*K),\\
		H_{k,2} = \slab_{(k-1) C_3 + 2C_2+2, (k-1) C_3 + 2C_2 + 4}(\hat{t}_0/r_*K),\\
		H_{k,3} = \slab_{(k-1) C_3 + 3C_2+4, (k-1) C_3 + 3C_2 + 6}(\hat{t}_0/r_*K),
	\end{gather*}
	where \(C_3 = 6+4C_2\).
	
	For a path \(\gamma:0\to x\), introduce
	\[
		\tau_{k,-} = \min\setof{r\geq 0}{\gamma_r\in H_{k,1}},\quad \tau_{k,+} = \max\setof{r\geq 0}{\gamma_r\in H_{k,3}}.
	\]
	Both are equal to \(\infty\) if there is no point in \(\gamma\) meeting the requirement. We say that \(k\) is \emph{pre-good} if all the following conditions occur (see Fig.~\ref{fig:preGood})
	\begin{enumerate}[label=(\roman*)]
		\item\label{pregood-i} \(\tau_{k,-} < \tau_{k,+} < \infty\).
		\item\label{pregood-ii} \(\gamma_0^{\tau_{k,-}}\subset \slab_{-\infty, t_0\cdot\gamma_{\tau_{k,-}} }(t_0)\) and \(\gamma_{\tau_{k,+}}^{|\gamma|}\subset \slab_{t_0\cdot\gamma_{\tau_{k,+}}, \infty }(t_0)\).
		\item\label{pregood-iii} \(\displaystyle\bigcup_{l<k}\;\bigcup_{\substack{\mathclap{r<\tau_{k,-}:}\\\gamma_r\in H_l}} \bbB_{C_1K}(\gamma_r)\subset \bcone_{t_0,\delta_1+\delta_2}(\gamma_{\tau_{k,-}})\) \,and\; \(\displaystyle\bigcup_{l>k}\;\bigcup_{\substack{\mathclap{r>\tau_{k,+}:}\\ \gamma_r\in H_l}} \bbB_{C_1K}(\gamma_r)\subset \fcone_{t_0,\delta_1+\delta_2}(\gamma_{\tau_{k,+}})\).
		\item\label{pregood-iv} \(\gamma_0^{\tau_{k,-}}\subset \bcone_{t_0,\delta_1+\delta_2}(\gamma_{\tau_{k,-}})\cup \bbB_{C_1K}(\gamma_{\tau_{k,-}}) \), and \(\gamma_{\tau_{k,+}}^{|\gamma|}\subset \fcone_{t_0,\delta_1+\delta_2}(\gamma_{\tau_{k,+}})\cup \bbB_{C_1K}(\gamma_{\tau_{k,+}})\).
		\item\label{pregood-v}\label{PreGood_ExistZ} There exists \(z\in H_{k,2}\) such that \(\bbB_{C_1K}(\gamma_{\tau_{k,+}})\subset \fcone_{t_0,\delta_1}(z)\), \(\bbB_{C_1K}(\gamma_{\tau_{k,-}})\subset \bcone_{t_0,\delta_1}(z)\).
		\item\label{pregood-vi} \(\gamma_{\tau_{k,-}}^{\tau_{k,+}}\subset \diam_{t_0,\delta_1+\delta_2}(\gamma_{\tau_{k,-}}, \gamma_{\tau_{k,+}})\cup \bbB_{C_1K}(\gamma_{\tau_{k,-}})\cup \bbB_{C_1K}(\gamma_{\tau_{k,+}})\).
	\end{enumerate}
	Note that, for \(C_1\) large enough and \(C_2 r_* > C_1\) large enough (depending on \(C_1\)), \(k\) being pre-good is implied by \(H_{k,1}, H_{k,2}, H_{k,3}\) all containing a \((t_0,\delta_1/2)\) cone-point of \(\Skeleton(\gamma)\) (see Fig.~\ref{fig:preGoodFromConePt}).
	
	We say that \(k\) is \emph{good} if \(k\) is pre-good and \(\gamma_{\tau_{k,-}}^{\tau_{k,+}} = z+\eta(\gamma_{\tau_{k,-}}-z,\gamma_{\tau_{k,+}}-z)\) where \(z\) is the smallest (for the alphabetical order) point in \(H_{k,2}\) satisfying Condition~\ref{PreGood_ExistZ} above. Note that \(k\) being good implies that \(z\in H_k\) is a \((t_0,\delta_1+\delta_2)\)-cone-point of \(\gamma\).
	
	\begin{figure}[ht]
		\includegraphics{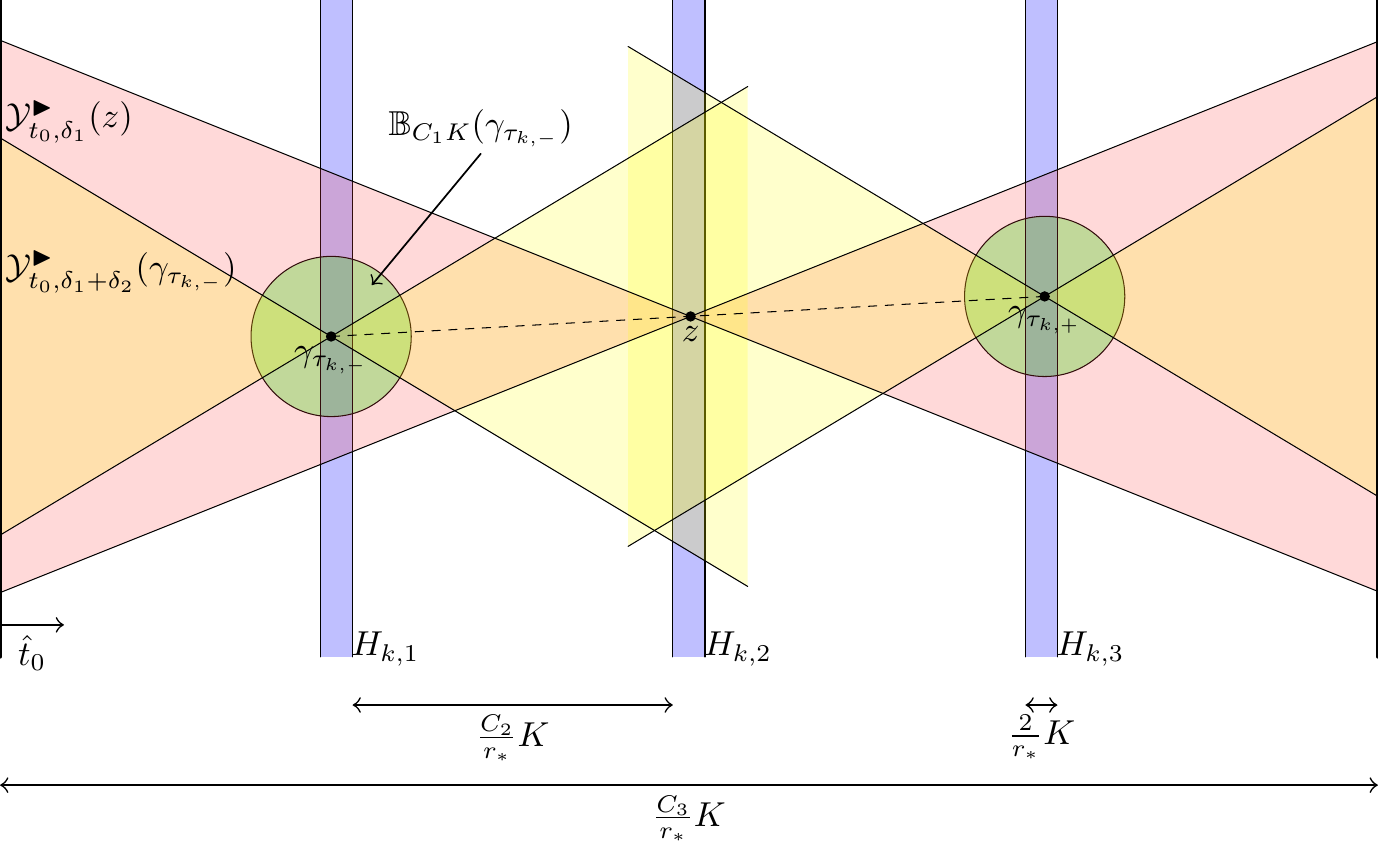}	
		\caption{Setting in the definition of pre-good slabs.}
		\label{fig:preGood}
	\end{figure}
	
	We then define
	\begin{gather*}
		I'(\gamma)=\setof{k}{k \text{ is pre-good}},\\
		I(\gamma)=\setof{k}{k \text{ is good}}.
	\end{gather*}
	Note that for \(C_2\) large enough (depending on \(C_1,\delta_1,\delta_2\)), for \(k\in I'(\gamma)\), \(I'(\gamma)\) does not depend on the particular realization of \(\gamma_{\tau_{k,-}}^{\tau_{k,+}}\) meeting the last requirement in the definition of pre-good \(k\)s. Theorem~\ref{thm:path_cone_points} with \(\delta_0=\delta_1+\delta_2\) follows from the fact that good slabs contain cone-points and the following two claims.
	\begin{claim}
		\label{claim:thm_path_cp_claim1}
		There exist \(\epsilon>0, K_0\geq 0, a>0\) such that, for any \(K\geq K_0\), \(t\in W_0\) and \(\norm{x}\) large enough,
		\[
			\sfe^{t\cdot x} \sum_{\gamma:\, 0\to x} q(\gamma) \mathds{1}_{\{|I'(\gamma)|\leq \epsilon\norm{x}/K\}}
			\leq \sfe^{-a \nu(x)}.
		\]
	\end{claim}
	\begin{claim}
		\label{claim:thm_path_cp_claim2}
		For any \(K\) fixed, there exist \(\epsilon'>0, a>0\) such that, for any \(\epsilon>0\),
		\[
			\sum_{\gamma:\, 0\to x} q(\gamma) \mathds{1}_{\{|I'(\gamma)|\geq \epsilon\norm{x}\}} \mathds{1}_{\{|I|\leq \epsilon'\epsilon\norm{x}\}}
			\leq \sfe^{-a \norm{x}} \sum_{\gamma:\, 0\to x} q(\gamma) \mathds{1}_{\{|I'(\gamma)|\geq \epsilon\norm{x}\}}.
		\]
	\end{claim}

	\begin{proof}[Proof of Claim~\ref{claim:thm_path_cp_claim1}]
		As being pre-good is implied by \(H_{k,1}, H_{k,2}, H_{k,3}\) all containing a \((t_0,\delta_1/2)\)-cone-point of \(\Skeleton(\gamma)\), the claim follows from Theorem~\ref{thm:skeleton_CP_well_distributed} (see Fig.~\ref{fig:preGoodFromConePt}).
		(For~\ref{pregood-ii}, observe that \(\gamma_0^{\tau_k ^-}\) cannot jump over  \(H_{k,1}\) before entering the latter, since this would create a vertex of the skeleton preventing the existence of a cone-point in \(H_{k,1}\).)		
	\end{proof}
	\begin{figure}[ht]
		\includegraphics{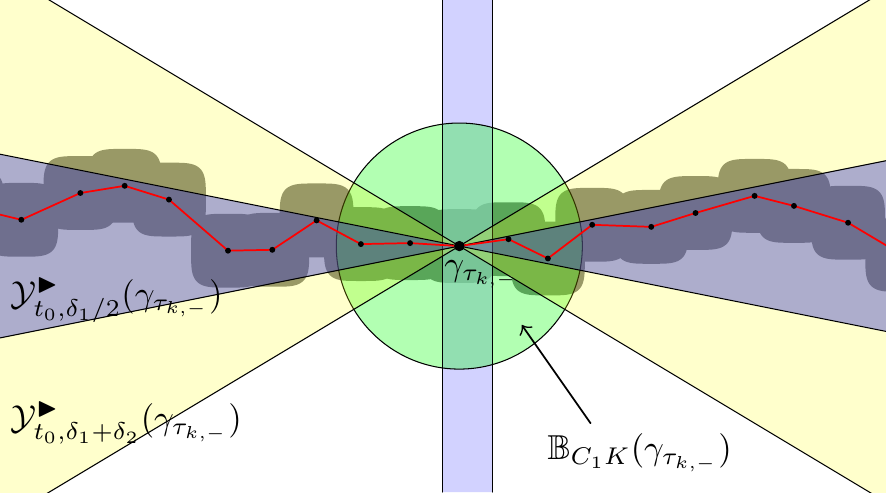}
		\caption{The existence of a (\(t_0, \delta_1/2\))-cone-point of the skeleton in the slab 	\(H_{k,i}\), \(i=1,2,3\), ensures the validity of properties \ref{pregood-iii}, \ref{pregood-iv} and~\ref{pregood-vi}. Note that the vertices of any path \(\gamma\) associated to the skeleton drawn in red are contained inside the shaded area and thus contained inside \(\fcone_{t_0,\delta_1+\delta_2}(\gamma_{\tau_{k,-}}) \cup \mathbb{B}_{C_1K}(\gamma_{\tau_{k,-}})\) once \(C_1\) is chosen large enough, which yields~\ref{pregood-iv}. A similar argument applies for~\ref{pregood-v} and~\ref{pregood-vi}.}
		\label{fig:preGoodFromConePt}
	\end{figure}
	\begin{proof}[Proof of Claim~\ref{claim:thm_path_cp_claim2}]
		For \(A = \{i_1, \dots, i_n\}\) a fixed realization of \(I'\) (with \(i_k\leq i_{k+1}\) for all \(1\leq k < n\)), let us write \(\gamma = \gamma^{0}\concatenate \eta^1 \concatenate \gamma^1\concatenate \eta^2 \concatenate\cdots \concatenate \gamma^n\), where
		\begin{gather*}
			\gamma^0 = \gamma_0^{\tau_{i_1, -}},\quad \gamma^{n} = \gamma_{\tau_{i_n, +}}^{|\gamma|} \\
			\gamma^{k} = \gamma_{\tau_{i_k,+}}^{\tau_{i_{k+1}, -}},\  1\leq k <n,\quad 
			\eta^k = \gamma_{\tau_{i_k,-}}^{\tau_{i_k, +}},\  1\leq k \leq n .
		\end{gather*}
		Then, still for \(A = \{i_1, \dots, i_n\}\) a fixed realization of \(I'\) and for \(A' = \{j_1, \dots, j_l\} \subset \{i_1, \dots, i_n\}\),
		\[
			\sum_{\gamma:\, 0\to x} q(\gamma) \mathds{1}_{\{I(\gamma)= A\setminus A'\}} \mathds{1}_{\{I'(\gamma)= A\}}  = \!\!\sum_{\gamma^0, \dots, \gamma^n} \sum_{\eta^1, \dots, \eta^n} q(\gamma^0\concatenate\eta^1\concatenate \cdots\concatenate \gamma^n) \mathds{1}_{\{I(\gamma^0\concatenate\eta^1\concatenate \cdots\concatenate \gamma^n) =  A\setminus A'\}}, 
		\]
		where the sums in the right-hand side are over the decomposition of paths \(\gamma\) as above such that \(I'(\gamma^0\concatenate\eta^1\concatenate \cdots\concatenate \gamma^n) = \{i_1, \dots, i_n\}\).
		Now, for fixed \(\gamma^0, \gamma^1, \dots, \gamma^n\),
		\begin{equation*}
			\sum_{\eta^1, \dots, \eta^n} q(\gamma^0\concatenate\eta^1\concatenate \cdots\concatenate \gamma^n) \mathds{1}_{\{I(\gamma^0\concatenate\eta^1\concatenate \cdots\concatenate \gamma^n) = A\setminus A'\}}\\
			\leq
			\sum_{\eta^1, \dots, \eta^n} q(\gamma^0\concatenate\eta^1\concatenate \cdots\concatenate \gamma^n) \prod_{k=1}^{l} \mathds{1}_{\{\eta^k \text{ is not canonical}\}}.
		\end{equation*}
	
		We will then have proved our claim if we can show that there is \(a'>0\), uniform over \(n\), \(\gamma^1, \dots, \gamma^n\), \(k \leq n\) and \(\eta^1, \dots,\eta^{k-1}, \eta^{k+1}, \dots, \eta^{n}\), such that
		\[
			\sum_{\eta^k} q(\gamma^0\concatenate\eta^1\concatenate \cdots\concatenate \gamma^n) \mathds{1}_{\{\eta^k \text{ is not canonical}\}}
			\leq
			\sfe^{-a'} \sum_{\eta^k} q(\gamma^0\concatenate\eta^1\concatenate \cdots\concatenate \gamma^n) \mathds{1}_{\left\{I'(\gamma^0\concatenate\eta^1\concatenate \cdots\concatenate \gamma^n) = \{i_1,\cdots, i_n\}\right\}}.
		\]
		The existence of such an \(a'\) is a consequence of the following observations: on the one hand, if we let \(\gamma'=\gamma^0\concatenate\cdots\concatenate\gamma^{k}\) and \(\gamma''=\gamma^{k+1}\concatenate\cdots\concatenate\gamma^{n}\), then~\ref{weight_property:monotonicity} implies that
		\[
			q(\gamma'\concatenate\eta^k_c\concatenate\gamma'') \geq q(\gamma') q(\eta^k_c) q(\gamma'').
		\]
		On the other hand, using~\ref{weight_property:finite_energy} twice implies the existence of \(C=C(\beta)\) such that
		\[
			\sum_{\eta^k} q(\gamma'\concatenate\eta^k\concatenate\gamma'') \leq C_{0} q(\gamma') q(\gamma'') \sum_{\eta^k} q(\eta^k) \leq C q(\gamma') q(\gamma'').
		\]
		Combining these inequalities with~\ref{weight_property:lower_bound} yields
		\begin{equation*}
		q(\gamma'\concatenate\eta^k_c\concatenate\gamma'')\geq C^{-1}
		\biggl(\prod_{i=1}^{|\eta_{c}^{k}|}\tilde{C}J_{\eta^{k}_{c,i-1},\eta^{k}_{c,i}}\biggr)\sum_{\eta^k} q(\gamma'\concatenate\eta^k_c\concatenate\gamma'')>0,
		\end{equation*}
		where the last inequality follows from the definition of the canonical path. This means that the weight associated to the canonical path provides a positive fraction of the sum over all possible \(\eta^{k}\), which leads to the desired claim.	
	\end{proof}
\end{proof}

\subsection{Irreducible decomposition}
\label{subsec:irreducible_decomposition}

Let \(s_0\in\bbS^{d-1}\) be fixed such that NSA holds. We also fix \(t_0\) \(\nu\)-dual to \(s_0\) and let \(\epsilon_0, U_0, W_0, \delta_0, n_0\) be given by Theorem~\ref{thm:path_cone_points}. We shall use the shorthand notation
\[
	\fcone \equiv \fcone_{t_0,\delta_0},\quad \CPts(\gamma) \equiv \CPts_{t_0,\delta_0}(\gamma).
\]
In the same line of ideas, we will say cone-point instead of \((t_0,\delta_0)\)-cone-point.

Given a path \(\gamma: 0\to x\) with \(x\in \fcone\), and containing at least one cone-point, we can uniquely decompose it as
\[
	\gamma= \gamma_L'\concatenate \gamma_1'\concatenate \cdots \concatenate \gamma_M'\concatenate \gamma_R',
\]
with \(\CPts(\gamma) = \{x_1, \dots, x_{M+1}\}\) the cone-points of \(\gamma\), \(\gamma_L': 0\to x_1\), \(\gamma_R': x_{M+1}\to x\), and \(\gamma_k': x_k\to x_{k+1}\). We allow empty paths for \(\gamma_L', \gamma_R'\) and allow \(M=0\). By construction,
\begin{itemize}
	\item \(\gamma_L = \gamma_L'\in\SetRootMarkBackCont^{\irreducible}\),
	\item \(\gamma_R = \gamma_R'-x_{M+1}\in\SetRootMarkForwCont^{\irreducible}\),
	\item \(\gamma_k = \gamma_k'-x_k\in \SetRootDiaCont^{\irreducible}\).
\end{itemize}
We can therefore write
\begin{equation}
	\label{eq:irred_decomposition}
	\gamma = \gamma_L\concatenate \gamma_1\concatenate \cdots \concatenate \gamma_M\concatenate \gamma_R.
\end{equation}
This decomposition is called the \emph{irreducible decomposition} of \(\gamma\) (see Fig.~\ref{fig:IrreducibleDecomposition}).
\begin{figure}[ht]
	\includegraphics{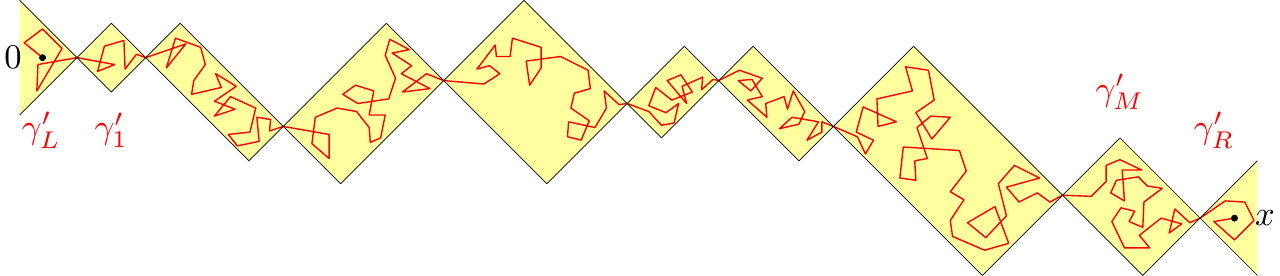}
	\caption{The irreducible decomposition \(\gamma = \gamma'_L\concatenate\gamma'_1\concatenate\cdots\concatenate\gamma'_M\concatenate\gamma'_R\) of a path \(\gamma:\,0\to x\).}
	\label{fig:IrreducibleDecomposition}
\end{figure}

From the previous section we deduce the following decay bound:
\begin{lemma}
	\label{lem:irred_decomp:unif_exp_dec}
	There exist \(C, a' > 0\) such that, for any \(t\in W_0\) and any path \(\gamma'\),
	\begin{gather*}
		\sum_{\gamma\in\SetRootDiaCont^{\irreducible}} \sfe^{t\cdot \displace(\gamma)} q(\gamma\given \gamma') \mathds{1}_{\{\norm{\displace(\gamma)} \geq l\}}
		\leq C\sfe^{-a'l},\\
		\sum_{\gamma_L\in\SetRootMarkBackCont^{\irreducible}} \sfe^{t\cdot \displace(\gamma_L)} q(\gamma_L\given \gamma') \mathds{1}_{\{\norm{\displace(\gamma_L)} \geq l\}}
		\leq C\sfe^{-a'l},\\
		\sum_{\gamma_R\in\SetRootMarkForwCont^{\irreducible}} \sfe^{t\cdot \displace(\gamma_R)} q(\gamma_R\given \gamma') \mathds{1}_{\{\norm{\displace(\gamma_R)} \geq l\}}
		\leq C\sfe^{-a'l}.
	\end{gather*}
\end{lemma}
\begin{proof}
	Being irreducible implies not having cone-points, except possibly at the endpoints.
	The exponential decay thus follows from Theorem~\ref{thm:path_cone_points}.
\end{proof}

%%%%%%%%%%%%%%%%%%%%%%%%%%%%%%%%%%%%%%%%%%%%%%%%%%%%%%%%%%%%%%%%%%%%%%%%%%%%%%%%%%%%%%%%%%%%%%%%%%%%%%%%
\section{Factorization of measure}
\label{sec:facto_meas}
%%%%%%%%%%%%%%%%%%%%%%%%%%%%%%%%%%%%%%%%%%%%%%%%%%%%%%%%%%%%%%%%%%%%%%%%%%%%%%%%%%%%%%%%%%%%%%%%%%%%%%%%

We now start from the output of Section~\ref{sec:path_analysis}. We therefore use the same quantities as in Subsection~\ref{subsec:irreducible_decomposition}. The goal of the present section is to pass from a pre-renewal structure (the irreducible decomposition of paths) to a \emph{true renewal structure} (that is, with factorized weights). A similar argument can be found in~\cite{Ioffe+Ott+Shlosman+Velenik-2021}, and a more general (and more technical) argument in~\cite{Ott+Velenik-2018}.

We will work with paths \(\gamma: 0\to y\), other starting points are treated using translation invariance.
%We will prove
%\begin{theorem}
%	\label{thm:random_walk_coupling}
%	Let \(s_0\in\bbS^{d-1}\) be a direction satisfying NSA. Then, there exist \(\epsilon_0>0, a_1>0, C_1<\infty\), such that for any \(t\in W_0\), there exist a constant \(C_{LR}\), and probability measures \(\rmp_L\), \(\rmp_R\), and \(\rmp\) on \(\SetRootMarkBackCont\), \(\SetRootMarkForwCont\), and \(\SetRootDiaCont\) respectively, such that for any \(f:\pathSet\to \bbC\),
%	\begin{multline}
%		\label{eq:thm:random_walk_coupling:expectation}
%		\Big|\sum_{\gamma:0\to y}e^{t\cdot y}q(\gamma)f(\gamma) - C_{LR}\sum_{\substack{\gamma_L\in\SetRootMarkBackCont\\\gamma_R\in\SetRootMarkForwCont}}\sum_{n\geq 0}\sum_{\gamma_1,\cdots,\gamma_n\in\SetRootDiaCont} \rmp_L(\gamma_L)\rmp_{R}(\gamma_R)\prod_{k=1}^{n} \rmp(\gamma_k) \mathds{1}_{\displace(\bar{\gamma}) = y} f(\bar{\gamma}) \Big|\\ \leq \normsup{f}C_1e^{-a_1 \norm{y}},
%	\end{multline}where \(\bar{\gamma} = \gamma_L\concatenate\gamma_1\concatenate\cdots\concatenate\gamma_n\concatenate \gamma_R\). Moreover,
%	\begin{itemize}
%		\item There exist \(C_2,a_2>0\) such that
%		\[
%			\max\Big(\rmp\big(\norm{\displace(\gamma_1)}\geq l\big), \rmp_L\big(\norm{\displace(\gamma_L)}\geq l\big), \rmp_R\big(\norm{\displace(\gamma_R)}\geq l\big)\Big)\leq C_2e^{-a_2l}.
%		\]
%		\item finite energy
%	\end{itemize}
%\end{theorem}

\subsection{Memory-percolation picture}\label{subsection:memory-percolation}

%We start by a blocking (grouping of irreducible pieces). Fix \(m\in \Z_+\), to be taken large (but independent of \(y\), the target point). In fact, \(m\) will only depend on the constants appearing in the mixing property of weights (Property~\ref{weight_property:mixing}). The requirement on \(m\) appears in the proof of Lemma~\ref{lem:exp_dec_connected_weights}.
Let \(\gamma: 0\to y\). From the irreducible decomposition~\eqref{eq:irred_decomposition}
%we can define the \emph{\(m\)-blocked irreducible decomposition}
%\[
%	\gamma = \gamma_L\concatenate \gamma_1'\concatenate\cdots\concatenate\gamma_{M}'\concatenate\gamma_R'
%\]where \(\gamma_i'\) are concatenations of exactly \(m\) irreducible paths, \(\gamma_L\in \SetRootMarkBackCont^{\irreducible}\), and \(\gamma_R'\) is the concatenation of \(<m\) irreducible paths and of a path in \(\SetRootMarkForwCont^{\irreducible}\). We denote the sets to which \(\gamma_i',\gamma_R'\) belong \(\SetRootDiaCont^{\irreducible,m},\SetRootMarkForwCont^{\irreducible,m}\).
and the definition of conditional weights, one can write
\begin{equation}
	\label{eq:weight_facto_eq1}
	q(\gamma) = q(\gamma_L) \prod_{k=1}^{M+1} q(\gamma_k\given \gamma_L\concatenate \cdots \concatenate \gamma_{k-1}),
\end{equation}
where we use the convention \(\gamma_0 = \gamma_L\) and \(\gamma_{M+1} = \gamma_R\).
For \(n+1\geq k\geq 2\), we can then define
\begin{gather*}
	p_k(\gamma \given \gamma_0, \dots, \gamma_{n}) = q(\gamma\given \gamma_{n+1-k}\concatenate \cdots\concatenate \gamma_{n}) - q(\gamma\given \gamma_{n+2-k}\concatenate \cdots\concatenate \gamma_{n}),\\
	p_1(\gamma \given \gamma_0, \dots, \gamma_{n}) = q(\gamma\given \gamma_{n}) - q(\gamma),\quad p_0(\gamma) = q(\gamma).
\end{gather*}
By the monotonicity Property~\ref{weight_property:monotonicity}, the weights \(p_k\), \(0\leq k\leq n+1\), are non-negative. Note that \(p_k(\cdot\given\gamma_0, \dots, \gamma_{n})\) does not depend on the full sequence \(\gamma_0, \dots, \gamma_{n}\) of irreducible pieces, but only on the \(k\) pieces \(\gamma_{n+1-k}, \dots, \gamma_{n}\). We then further express \(q(\gamma_k\given \gamma_L\concatenate \cdots \concatenate \gamma_{k-1})\) as a telescoping sum,
\begin{equation}
	\label{eq:weight_facto_eq2}
	q(\gamma_k\given \gamma_L\concatenate \cdots \concatenate \gamma_{k-1}) = \sum_{l=0}^{k}p_l(\gamma_k\given \gamma_L,\cdots , \gamma_{k-1}).
\end{equation}
We define the \emph{connected weights} by \(\qco(\gamma) = q(\gamma)\) and, for \(k>0\),
\[
	\qco(\gamma_0, \dots, \gamma_k) = q(\gamma_0) \sum_{l_1, \dots, l_k \geq 0 \text{ connected}} \prod_{i=1}^{k} p_{l_i}(\gamma_i\given \gamma_0, \dots, \gamma_{i-1}),
\]
where the sum is over sequences of integers \(l_i\geq 0\), \(1\leq i\leq k\), such that \(l_i\leq i\) and \(\bigcup_{i=1}^{k} \{i-l_i, i-l_i+1, \dots, i-1\} = \{0, \dots, k-1\}\) (see Fig.~\ref{Figure:memory_percolation_picture}).

\begin{figure}[ht]
	\centering
	\begin{tikzpicture}[scale=0.8]
		\draw[thick, red] (10,0) arc (0:180:2);
		\draw[thick, red] (9,0) arc (0:180:0.5);
		\draw[thick, red] (7,0) arc (0:180:1);
		\draw[thick, red] (4,0) arc (0:180:1);
		\draw[thick, red] (3,0) arc (0:180:0.5);
		\draw[thick, red] (2,0) arc (0:180:1);
		\draw[thick, red] (1,0) arc (0:180:0.5);
		\draw[ultra thick, blue](-0.3,-0.3)--(-0.3,-0.6)--(4.3,-0.6)--(4.3,-0.3);
		\draw[ultra thick, blue](4.7,-0.3)--(4.7,-0.6)--(10.3,-0.6)--(10.3,-0.3);
		\foreach \x in {0,...,10}{
			\draw (\x,0) node{\(\large\bullet\)} node[below]{\(\x\)};
		}
	\end{tikzpicture}
	\caption{Each index \(i\) receives a ``memory arc'' of range \(l_i\) (depicted in red). Connected components are depicted by thick blue lines.}
	\label{Figure:memory_percolation_picture}
\end{figure}
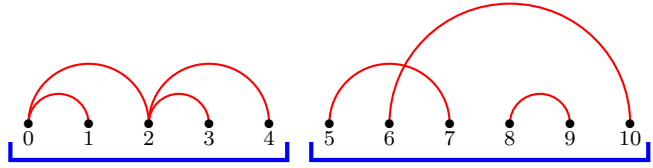

Plugging~\eqref{eq:weight_facto_eq2} into~\eqref{eq:weight_facto_eq1}, expanding and regrouping terms, one obtains
\begin{equation}
	\label{eq:weight_to_connected_weight}
	q(\gamma_0\concatenate \gamma_1\concatenate\cdots\concatenate \gamma_{M+1})
	= \sum_{n=1}^{M+2}\sum_{\substack{\ell_1, \dots, \ell_n \geq 1\\ \sum \ell_i = M+2}} \prod_{k=1}^{n} \qco(\gamma_{L_{k}}, \dots, \gamma_{L_k+\ell_k-1}),
\end{equation}
where \(L_1=0\) and, for \(k\geq 2\), \(L_{k} = \sum_{i=1}^{k-1} \ell_i\). We are now ready to define the factorized weights. For a path \(\gamma\) having irreducible decomposition \(\gamma_0\concatenate\gamma_1\concatenate\cdots\concatenate\gamma_{n+1}\), we define
\begin{gather}
	\label{eq:def:facto_weight}
	\tilde{q}(\gamma) = \qco(\gamma_0, \gamma_1, \dots, \gamma_{n+1})
	\quad\text{ and }\quad
	\tilde{q}_t(\gamma) = \sfe^{t\cdot \displace(\gamma)} \tilde{q}(\gamma).
\end{gather}
Note that, by the uniqueness of the irreducible decomposition, \(\tilde{q}\) is a well-defined weight over paths having an irreducible decomposition. Straightforward algebra yields, for any \(f:\pathSet\to\bbC\),
\begin{multline}
	\label{eq:expectation_weights_to_connected_weights}
	\sum_{\gamma:\, 0\to y} q(\gamma) f(\gamma) \mathds{1}_{\{\CPts(\gamma)\setminus\{0,y\} \neq \emptyset\}}\\
	= \sum_{\gamma_L\in \SetRootMarkBackCont^{\irreducible}, \gamma_R\in\SetRootMarkForwCont^{\irreducible}} \sum_{n\geq 0} \sum_{\gamma_1, \dots, \gamma_n\in \SetRootDiaCont^{\irreducible}} f(\bar{\gamma}) \mathds{1}_{\{\displace(\bar{\gamma}) = y\}} \qco(\gamma_L, \gamma_1, \dots, \gamma_R) \\
	+ \sum_{\gamma_L\in \SetRootMarkBackCont, \gamma_R\in\SetRootMarkForwCont} \sum_{n\geq 0} \sum_{\gamma_1, \dots, \gamma_n\in\SetRootDiaCont} f(\bar{\gamma}) \mathds{1}_{\{\displace(\bar{\gamma}) = y\}} \tilde{q}(\gamma_L) \tilde{q}(\gamma_R) \prod_{k=1}^{n} \tilde{q}(\gamma_k),
\end{multline}
where \(\bar{\gamma} = \gamma_L\concatenate\gamma_1\concatenate\cdots\concatenate\gamma_R\).
%(and the \(\rmp,\tilde{\rmp}\) weights are understood to be \(0\) on paths which are not obtainable via concatenation of blocked irreducible pieces)

\begin{lemma}
	\label{lem:exp_dec_connected_weights}
	There exist \(C, c > 0\) such that, for any \(\gamma_0\in \SetRootMarkBackCont^{\irreducible}\), \(n\geq 0\), \(\gamma_1, \dots, \gamma_n\in \SetRootDiaCont^{\irreducible}\) and \(\gamma_{n+1}\in \SetRootMarkForwCont^{\irreducible}\),
	\[
		\frac{\qco(\gamma_0, \gamma_1, \dots, \gamma_n, \gamma_{n+1})}{q(\gamma_0\concatenate \gamma_1\concatenate\cdots\concatenate \gamma_n\concatenate\gamma_{n+1})} \leq C\sfe^{-cn}.
	\]
\end{lemma}
\begin{proof}
	First, observe that the (backward) cone containment property of \(\gamma_0\concatenate\gamma_1\concatenate\cdots\concatenate\gamma_{i-1}\), the fact that irreducible paths have a strictly positive projection on \(t_0\) and the mixing upper bound of Property~\ref{weight_property:mixing} imply the existence of \(C_{\rmm}, c_{\rmm} > 0\) such that, for any \(l\geq 0\) \emph{and any sequence} \(\gamma_0, \gamma_1, \dots, \gamma_n, \gamma_{n+1}\) as in the statement,
	\begin{equation}
		\label{eq:lem:exp_dec_connected_weights_eq1}
		\nonumber p_{l}(\gamma_i\given \gamma_{0}, \dots, \gamma_{i-1})
		\leq C_{\rmm} \sfe^{-c_{\rmm}l} q(\gamma_i\given \gamma_{i+1-l}\concatenate \cdots\concatenate \gamma_{i-1})
		\leq C_{\rmm} \sfe^{-c_{\rmm}l} q(\gamma_i\given \gamma_{0}\concatenate \cdots\concatenate \gamma_{i-1}),
	\end{equation}
	where we use monotonicity (Property~\ref{weight_property:monotonicity}) in the last line. Then, for \(\gamma_0, \gamma_1, \dots, \gamma_n, \gamma_{n+1}\) fixed as in the statement, we can define a probability measure \(P\) on sequences \((l_0, l_1, \dots,\penalty0 l_{n+1})\) satisfying \(l_i\leq i\) by
	\[
		P(l_0, \dots, l_{n+1}) = \frac{\prod_{i=0}^{n+1} p_{l_i}(\gamma_i\given \gamma_0, \dots, \gamma_{i-1})}{q(\gamma_0\concatenate \cdots\concatenate \gamma_{n+1})}.
	\]
	From the sequence \(\bar{l} = (l_0, l_1, \dots, l_{n+1})\), one obtains a stick percolation model by looking at \(I = \bigcup_{i=0}^{n+1} [i-l_i, i]\). We can then write
	\begin{equation}
		\label{eq:lem:exp_dec_connected_weights_eq2}
		\frac{\qco(\gamma_0, \dots, \gamma_{n+1})}{q(\gamma_0\concatenate\cdots\concatenate\gamma_{n+1})}
		= P(I\text{ is connected}).
	\end{equation}
	Moreover, by construction, \(\bar l\) is a family of independent random variables with distribution \(P(l_i=l) = \frac{p_l(\gamma_i\given \gamma_0, \dots, \gamma_{i-1})}{q(\gamma_i\given \gamma_0\concatenate\cdots\concatenate\gamma_{i-1})}\). Moreover, Property~\ref{weight_property:finite_energy} and~\eqref{eq:lem:exp_dec_connected_weights_eq1} imply that
	\begin{equation}\label{ineq:geometric_domination}
		P(l_i=0) \geq C_{\rmc}^{-1} > 0,\quad P(l_i=l) \leq C_{\rmm} \sfe^{-c_{\rmm} l}.
	\end{equation}
	Therefore, there exists \(c>0\) (depending on \(c_{\rmm}, C_{\rmm}\) and \(C_{\rmc}\)) such that \(P(l_i\geq l)\leq \sfe^{-cl}\). The sequence \(\bar{l}\) is thus stochastically dominated by an i.i.d.\ sequence of geometric random variables of parameter \(1-\sfe^{-c}\), \emph{uniformly over the sequence} \(\gamma_0, \gamma_1, \dots, \gamma_n, \gamma_{n+1}\).

	The claimed exponential decay follows from this stochastic domination,~\eqref{eq:lem:exp_dec_connected_weights_eq2} and standard one-dimensional percolation arguments. The details of the latter percolation estimates can be found, for example, in~\cite[Claim C.3]{Ott+Velenik-2018}.
\end{proof}

\begin{corollary}
	\label{cor:renewal_density}
	For any \(\epsilon>0\), there exist \(C, c > 0\) such that, for any \(t\in W_0\),
	\[
		\sfe^{t\cdot y}\sum_{\substack{\gamma_L\in \SetRootMarkBackCont^{\irreducible}\\ \gamma_R\in\SetRootMarkForwCont^{\irreducible}}} \sum_{n\geq \epsilon \norm{y}} \sum_{\gamma_1, \dots, \gamma_n\in \SetRootDiaCont^{\irreducible}} f(\bar{\gamma}) \mathds{1}_{\{\displace(\bar{\gamma}) = y\}} \qco(\gamma_L, \gamma_1, \dots, \gamma_R) \leq C\normsup{f} \sfe^{-c\norm{y}}.
	\]
\end{corollary}
\begin{proof}
	The result follows directly from the existence (and the definition) of \(\nu\), \(t\in\partial \calW_{\nu}\) and Lemma~\ref{lem:exp_dec_connected_weights}.
\end{proof}

\subsection{Gathering the pieces}

We now combine the discussion of the previous subsection and the output of Section~\ref{sec:path_analysis}. The following result follows from~\eqref{eq:expectation_weights_to_connected_weights} and a direct combination of Theorem~\ref{thm:path_cone_points} and Corollary~\ref{cor:renewal_density}.
\begin{lemma}
	\label{lem:approx_by_factorized_weights}
	There exist \(C, c > 0\) such that, for any \(f:\pathSet\to \bbC\) with \(\normsup{f}\leq 1\), any \(y\in \fcone\) and any \(t\in W_0\),
	\[
		\Bigl|\sfe^{t\cdot y} \sum_{\gamma:\,0\to y} \!q(\gamma) f(\gamma) - \sum_{\substack{\gamma_L\in\SetRootMarkBackCont\\\gamma_R\in\SetRootMarkForwCont}} \sum_{n\geq 0} \sum_{\gamma_1, \dots, \gamma_n\in\SetRootDiaCont}\!\!\! f(\bar{\gamma}) \mathds{1}_{\{\displace(\bar{\gamma}) = y\}} \tilde{q}_t(\gamma_L) \tilde{q}_t(\gamma_R) \prod_{i=1}^n \tilde{q}_t(\gamma_i) \Bigr|
		\leq C\sfe^{-c\norm{y}},
	\]
where \(\bar{\gamma} = \gamma_L \concatenate \gamma_1 \concatenate \cdots \concatenate \gamma_n \concatenate \gamma_R\).
\end{lemma}
Let us now prove the required properties of the weights \(\tilde{q}_t\).
\begin{lemma}
	\label{lem:exp_dec_renorm_connected_weights}
	There exist \(C, c > 0\) such that, for any \(t\in W_0\),
	\[
		\sum_{y\in\fcone} \sum_{\gamma:\,0\to y} \tilde{q}_t(\gamma) \mathds{1}_{\{\norm{y}\geq l\}} \leq C\sfe^{-cl},
	\]
	where the sum is over \(\gamma\) admitting an irreducible decomposition.
\end{lemma}
\begin{proof}
	Recall the definition of \(\tilde{q}_t\) (equation~\eqref{eq:def:facto_weight}). We first partition on whether \(\gamma\) has more that \(\epsilon l\) cone-points (for some \(\epsilon>0\) small) or not. In the first case, the claim follows from an application of Lemma~\ref{lem:exp_dec_connected_weights}, the definition of \(\nu\) and \(t\in\calW_{\nu}\). In the second case, it follows from Theorem~\ref{thm:path_cone_points}. 
\end{proof}

\begin{lemma}
	\label{lem:normalization_renorm_connected_weights}
	For any \(t\in W_0\),
	\begin{equation}
		\label{eq:normalization_renorm_connected_weights}
		\sum_{\gamma\in\SetRootDiaCont} \tilde{q}_t(\gamma) = 1.
	\end{equation}
\end{lemma}
\begin{proof}
	Consider the generating function 
	\[
		\bbG(h) = \sum_{y\in \fcone} \sfe^{y\cdot h} \sum_{n\geq 1} \sum_{\gamma_1, \dots, \gamma_n\in\SetRootDiaCont} \mathds{1}_{\{\displace(\bar{\gamma}) = y\}} \prod_{i=1}^{n} \tilde{q}(\gamma_i).
	\]
	We claim that, for any \(t\in W_0\), \(\bbG(\lambda\hat{t}) < \infty \) for any \(\lambda<\norm{t}\) and \(\bbG(\lambda\hat{t})=\infty \) for any \(\lambda>\norm{t}\). Since \(t\in\calW_{\nu}\), the case \(\lambda<\norm{t}\) follows from
	\[
		\bbG(h) \leq \sum_{y\in\Zd} \sfe^{h\cdot y} \sum_{\gamma:\,0\to y} q(\gamma),
	\]
	the latter series converging on the interior of \(\calW_{\nu}\).
	
	The case \(\lambda>\norm{t}\) follows from \(U_0\subset \fcone\) (recall the statement of Theorem~\ref{thm:path_cone_points}), Lemma~\ref{lem:approx_by_factorized_weights} and the bound of Lemma~\ref{lem:exp_dec_renorm_connected_weights} applied to the pieces \(\gamma_L,\) and \(\gamma_R\). Indeed, restricting the summation to vertices \(y\) close to the half-line \(\setof{\alpha s}{\alpha\in\R_{\geq 0}}\), with \(s\in U_0\) dual to \(t\), is sufficient to obtain a diverging sum.
	
	Now, consider
	\[
		\bbQ(h) = \sum_{\gamma\in \SetRootDiaCont} \tilde{q}(\gamma) \sfe^{h\cdot \displace(\gamma)}.
	\]
	Observe that convergence of \(\bbG(h)\) implies convergence of \(\bbQ(h)\) and the following renewal equation holds:
	\[
		\bbG(h) = \frac{\bbQ(h)}{1-\bbQ(h)}.
	\]
	From Lemma~\ref{lem:exp_dec_renorm_connected_weights}, there exists \(\epsilon > 0\) such that \(\bbQ\big(\hat{t} (\norm{t}+ \epsilon)\big) <\infty\) for any \(t\in W_0\). In particular, for any \(t\in W_{0}\), the fact that \(\norm{t}\) is the radius of convergence of \(\lambda\mapsto \bbG(\lambda\hat{t})\) implies that \(1-\bbQ(t) = 0\), which is equivalent to~\eqref{eq:normalization_renorm_connected_weights}.
\end{proof}

\begin{lemma}
	\label{lem:centering_steps}
	For any \(t\in W_0\), there is a unique \(s\in\bbS^{d-1}\) which is \(\nu\)-dual to \(t\). Moreover, there exists \(\zeta>0\) such that
	\begin{equation}
		\sum_{\gamma\in\SetRootDiaCont} \tilde{q}_{t}(\gamma) \displace(\gamma) = \zeta s.
	\end{equation}
\end{lemma}
\begin{proof}
	Let
	\[
		v_{t} =  \sum_{\gamma\in\SetRootDiaCont} \tilde{q}_{t}(\gamma) \displace(\gamma).
	\]
	Note that \(v_t\) is non-zero, since \(\gamma\in\SetRootDiaCont\) implies \(\displace(\gamma)\cdot t_0>0\). By Lemmas~\ref{lem:normalization_renorm_connected_weights} and~\ref{lem:exp_dec_renorm_connected_weights}, the push-forward of \(\tilde{q}_{t}\) by \(\displace\) is a probability distribution on \(\Zd\) with exponential tails. Therefore, by standard large deviation estimates,
	\[
		\sum_{M\geq 1} \sum_{\gamma_1, \dots, \gamma_M\in\SetRootDiaCont} \mathds{1}_{\{\displace(\bar{\gamma}) = ns'\}} \prod_{i=1}^{M} \tilde{q}_{t}(\gamma_i) \leq C_{s'} \sfe^{-c_{s'} n},
	\]
	if \(s'\neq \lambda v_{t}\) for all \(\lambda>0\).
	But, by Lemmas~\ref{lem:approx_by_factorized_weights} and~\ref{lem:exp_dec_renorm_connected_weights} and the definition of \(\nu\), when \(t\) is \(\nu\)-dual to \(s\), we necessarily have
	\[
		\sum_{M\geq 1} \sum_{\gamma_1, \dots, \gamma_M\in\SetRootDiaCont} \mathds{1}_{\{\displace(\bar{\gamma}) = ns\}} \prod_{i=1}^{M} \tilde{q}_{t}(\gamma_i) = \sfe^{\sfo(n)}.
	\]
	We thus conclude that \(s\) \(\nu\)-dual to \(t\) implies \(s= v_t/\norm{v_t}\).
\end{proof}

\begin{lemma}
	\label{lem:analyticity_W}
	Let \(W\subset W_0\) be open. Then, \(W\) is analytic.
\end{lemma}
\begin{proof}
	Let \(t\in W_0\) be such that \(t\) has an open neighborhood in \(W_0\). Let \(s\) be the (unique by Lemma~\ref{lem:centering_steps}) direction \(\nu\)-dual to \(t\). It follows from~\eqref{eq:normalization_renorm_connected_weights} that, in a neighborhood of \(t\),
	\[
		t'\in \partial\calW \quad\Leftrightarrow\quad
		\bbQ(t') = \sum_{\gamma\in\SetRootDiaCont} \tilde{q}(\gamma)\sfe^{t'\cdot \displace(\gamma)} =1.
	\]
	We want to use the (analytic version of the) Implicit Function Theorem. Let \(b_1, \dots, b_{d-1}\) be such that \(\{b_1, \dots, b_{d-1}, s\}\) is an orthonormal basis of \(\Rd\).
	Introduce
	\begin{align*}
		F: \bbR^{d-1}\times \R &\to \R\\
		F(h,\lambda) &= \bbQ\Big( \sum_{i=1}^{d-1} h_ib_i + \lambda s \Big)-1.
	\end{align*}
	Define \(\tilde{h}\) by \(\tilde{h}_i = t\cdot b_i\) and let \(\tilde{\lambda} = t\cdot s = \nu(s)\). From Lemma~\ref{lem:exp_dec_renorm_connected_weights}, \(F\) is analytic in a neighborhood of \((\tilde{h},\tilde{\lambda})\). By Lemma~\ref{lem:normalization_renorm_connected_weights}, \(F(\tilde{h},\tilde{\lambda}) = 0\). Moreover, by Lemma~\ref{lem:centering_steps},
	\[
		\frac{\partial}{\partial \lambda} F\Big|_{(\tilde{h}, \tilde{\lambda})}
		=
		\sum_{\gamma\in \SetRootDiaCont} \tilde{q}_{t} \displace(\gamma)\cdot s > 0.
	\]
	Now, in a neighborhood of \(t\), \(\partial\calW\) coincides with the set
	\[
		\Bsetof{\sum_{i=1}^{d-1} h_ib_i + \lambda s}{F(h,\lambda) = 0, (h,\lambda) \text{ in a neighborhood of }(\tilde{h},\tilde{\lambda})}.
	\]
	The claim now follows from an application of the Implicit Function Theorem.
\end{proof}

\subsection{Coupling with random walk}
\label{subsec:coupling_with_RW}

We end the section by proving Theorem~\ref{thm:main_coupling_with_RW}. Using the notations and objects introduced of this section, we simply set
\begin{gather*}
	C_L = \sum_{\gamma_L\in\SetRootMarkBackCont } \tilde{q}_{t_0}(\gamma_L),\quad
	C_R = \sum_{\gamma_R\in\SetRootMarkForwCont } \tilde{q}_{t_0}(\gamma_R),\quad
	C_{LR} = C_L C_R,\\
	\rmp_L(\gamma_L) = \frac{1}{C_L} \tilde{q}_{t_0}(\gamma_L),\ \gamma_L\in\SetRootMarkBackCont,\quad
	\rmp_R(\gamma_R) = \frac{1}{C_R} \tilde{q}_{t_0}(\gamma_R),\ \gamma_R\in\SetRootMarkBackCont,\\
	\rmp(\gamma) = \tilde{q}_{t_0}(\gamma),\ \gamma\in\SetRootDiaCont.
\end{gather*}
Lemmas~\ref{lem:approx_by_factorized_weights},~\ref{lem:exp_dec_renorm_connected_weights},~\ref{lem:normalization_renorm_connected_weights} and~\ref{lem:centering_steps} provide all the required properties except for ``finite energy''. To obtain the latter, simply observe that, when the irreducible decomposition of \(\gamma\) is itself (that is, \(\gamma\) is irreducible),
\[
\tilde{q}_t(\gamma) = \sfe^{t\cdot \displace (\gamma)}q(\gamma),
\]
which implies the desired finite energy property.

%%%%%%%%%%%%%%%%%%%%%%%%%%%%%%%%%%%%%%%%%%%%%%%%%%%%%%%%%%%%%%%%%%%%%%%%%%%%%%%%%%%%%%%%%%%%%%%%%%%%%%%%
\section{Proofs of the main Theorems}\label{sec:proof_main_thm}
%%%%%%%%%%%%%%%%%%%%%%%%%%%%%%%%%%%%%%%%%%%%%%%%%%%%%%%%%%%%%%%%%%%%%%%%%%%%%%%%%%%%%%%%%%%%%%%%%%%%%%%%

\subsection{Local Limit Theorem and OZ asymptotics}\label{section:OZ_asymptotics}

We want to prove Theorem~\ref{thm:main_OZ_asymp}. Fix \(s\in\bbS^{d-1}\) such that NSA holds. Let \(t\) be \(\nu\)-dual to \(s\). Let \(\delta_0\) be given by Theorem~\ref{thm:path_cone_points}. Let \(\rmp\), \(\rmp_L\), \(\rmp_R\), \(C_L\) and \(C_R\) be as in Subsection~\ref{subsec:coupling_with_RW}. We shall use the same notation for their push-forwards by \(\displace\). Recall that \(\rmp\), \(\rmp_L\) and \(\rmp_R\) have exponential tails. We denote by \(\RWP{u}\) the law of the directed random walk \((S_{n})_{n\geq 0}\) on \(\Zd\), starting at \(u\in\Zd\) and with increments of law \(\rmp\), by \(E_{u}\) the corresponding expectation. By Lemma~\ref{lem:centering_steps}, \(E_{0}(S_{1})=\mu=\zeta_{0} s\). We want to study the associated Green function, i.e.
\begin{equation*}
	\RWG(0,y)=\sum_{k\geq 1}\RWP{0}(S_{k}=y).
\end{equation*}
We will derive sharp asymptotics of \(G\) for \(y\) ``close'' to \(n\mu\). Fix \(n\) large and \(y\in\Zd\) such that \(\norm{y-n\mu}<\log(n)^{2}\). Fix \(\varepsilon>0\) small. Since \(p\) has an exponential tail, it follows from standard large deviation bounds that 
\begin{equation}\label{eq:OZ_prefactor_intemediate_1}
	\sum_{k: \abs{k-n}\geq n^{1/2}\log(n)^2}\RWP{0}(S_{k}=y)\leq \sfe^{-c\log(n)^2},
\end{equation}
for some \(c>0\).

For other values of \(k\), we use the following refinement of the local limit theorem. Fix \(y_{k}=\frac{y-k\mu}{\sqrt{k}}\). There exists \(p_{1}\) a polynomial depending on the first three cumulants of \(S_{1}\) such that
\begin{equation*}
	\RWP{0}(S_{k}=y)=\frac{1}{k^{d/2}}\left(f_{D}(y_{k})+\frac{1}{\sqrt{k}}p_{1}(y_{k})f_{D}(y_{k})\right)+\sfo(k^{-(d+2)/2}).
\end{equation*}
where \(f_{D}\) the density function of multidimensional normal law with mean \(0\) and covariance matrix \(D\), where \(D\) is the covariance matrix of \(S_{1}\) (see, for instance, \cite[Corollary~22.3]{Bhattacharya-1979}).
We get 
\begin{equation}\label{eq:OZ_prefactor_intemediate_2}
\sum_{k:\, \abs{k-n}\leq n^{1/2}\log(n)^2}\!\!\!\!\!\!\!\!\!\!\!\! \RWP{0}(S_{k}=y)
=
\frac{(1+\sfO_{n}(n^{-\frac{1}{2}+\varepsilon}))}{n^{d/2}}
\sum_{k}
\left(f_{D}(y_{k})+f_{D}(y_{k})p_{1}(y_{k})n^{-1/2}\right)+\sfo(n^{-d/2}),
\end{equation}
where the condition over \(k\) is the same on the RHS.

Let us estimate the last sum over \(k\). Using the definition of \(f_{D}\), we have
\begin{multline*}
\sum_{k}
\left(f_{D}(y_{k})+f_{D}(y_{k})p_{1}(y_{k})n^{-1/2}\right)
\\
\begin{split}
&=
\frac{1}{\sqrt{(2\pi)^{d}\det(D)}}\sum_{k}\exp(-y_{k}\cdot D^{-1}\cdot y_{k})(1+p_{1}(y_{k})n^{-1/2})
\\
&=
\frac{\sqrt{n}(1+\sfO(n^{3\varepsilon-1/2}))}{
\sqrt{(2\pi)^{d}\det(D)n^{-1/2}}}
\sum_{k}n^{-1/2}\exp\bigl(- (\frac{n-k}{\sqrt{n}})^{2}  \mu\cdot D^{-1}\cdot\mu\bigr)
\bigl(1+\sfO_{n}(n^{-1/2})p_{1}(\frac{n-k}{\sqrt{n}}\mu)\bigr)
\\
&=
\frac{\sqrt{n}(1+\sfO(n^{3\varepsilon-1/2}))}
{\sqrt{(2\pi)^{d}\det(D)\mu\cdot D^{-1}\cdot\mu}}\biggl(\int_{-\infty}^{\infty}\exp(-1/2 s^{2}) \,\mathrm{d}s + \sfO(n^{-1/2})\biggr)
\\
&=
\frac{\sqrt{n}(1+\sfO(n^{3\varepsilon-1/2}))}
{\sqrt{(2\pi)^{d-1}\det(D)\mu\cdot D^{-1}\cdot\mu}}.
\end{split}
\end{multline*}
By replacing \(3\varepsilon\) by \(\varepsilon\), this last estimate combined with~\eqref{eq:OZ_prefactor_intemediate_1} and~\eqref{eq:OZ_prefactor_intemediate_2} implies
\begin{equation}\label{eq:Green_estimate}
	\RWG(0,y) = \tilde{C}n^{-(d-1)/2}(1+\sfO(n^{\varepsilon-1/2})),
\end{equation}
with \(\tilde{C}=((2\pi)^{d-1}\det(D)\mu\cdot D^{-1}\cdot\mu)^{-1/2}\) and uniformly in \(y\) satisfying \(\norm{y-n\mu}\leq \log(n)^{2}\).

To conclude, it follows directly from Lemma~\ref{lem:approx_by_factorized_weights} that 
\[
	\sfe^{\tilde{n}\nu(s)}\, G(0,\tilde{n}s) = (1+\sfO(\sfe^{-c\tilde{n}})) C_{LR} \sum_{u,v\in\Zd} \rmp_L(u) \rmp_R(v) \RWG(u,\tilde{n}s-v).
\]
On the one hand, since \(p_{L}\) and \(p_{R}\) have exponential tails, the contribution of \(u\), \(v\) such that \(\max\{\norm{u}, \norm{\tilde{n}s-v}\}\geq\log(\tilde{n}/\zeta_{0})^{2}/2\) is of order \(\sfe^{-\sfO(\log(\tilde{n})^{2})}\) (which is in particular asymptotically smaller than any polynomial). On the other hand, it follows from~\eqref{eq:Green_estimate} with \(y=v-u\) that for any \(u\), \(v\) such that \(\max\{\norm{u}, \norm{\tilde{n}s-v}\} < \log(\tilde{n}/\zeta_{0})^{2}/2\), we have
\begin{equation*}
	\RWG(u,v) = \tilde{C}\zeta_0^{(d-1)/2}\tilde{n}^{-(d-1)/2}(1+\sfO(\tilde{n}^{\varepsilon-1/2})).
\end{equation*}
In particular, we obtain the Ornstein--Zernike asymptotics
\begin{equation*}
	G(0,\tilde{n}s) = \dfrac{c_{s}}{\tilde{n}^{(d-1)/2}} \sfe^{-\tilde{n}\nu(s)} (1+\sfO(\tilde{n}^{\varepsilon-1/2})),
\end{equation*}
with \(c_{s}=\tilde{C}\zeta_0^{(d-1)/2}\sum_{u,v\in\mathbb{Z}^{d}}\rmp_{L}(u)\rmp_{R}(v)\).

\subsection{Analyticity of \(\calU_{\nu}\)}\label{sec:ProofAnalyticity}
We only sketch the argument; see~\cite[Chapter~2]{schneider_2013} for additional information.
Let \(s_0\in\bbS^{d-1}\) be such that NSA holds. Let \(\epsilon_0 > 0\) and \(U_0, W_0\) be given by Theorem~\ref{thm:path_cone_points}. Consider \(W = W_0\setminus \partial W_0\). This is an open subset of \(W_0\). In particular, by Lemma~\ref{lem:analyticity_W}, it is analytic. So, dual directions are uniquely defined for \(t\in W\). Moreover, the set of directions dual to some point in \(W\) is non-trivial, so \(W\) does not contain affine parts. In particular, still by analyticity of \(W\), the convex duality lifts to an analytic bijection between \(W\) and \(U\), an open neighborhood of \(s_0\). The local analyticity of \(\calU_{\nu}\) around \(s_0/\nu(s_0)\) follows then from the fact that \(s\mapsto 1/\nu(s)\) is analytic in a neighborhood of \(s_0\) (being the inverse of a non-zero analytic function).

\section*{Acknowledgments}
YA and YV are supported by the Swiss NSF grant 200021\_200422. SO is supported by the Swiss NSF grant 200021\_182237. All authors are members of the NCCR SwissMAP.

\appendix

\section{Properties of the weights of HT Ising paths}
\label{sec:IsingHTProperties}

In this section, we sketch the proof of Lemma~\ref{lem:IsingQ}.
\begin{lemma}
	\label{app:lem:Ising_Properties}
	Suppose \(\beta<\betac\). Then, \(q\in \calQ\).
\end{lemma}
\begin{proof}
	Property~\ref{weight_property:ICL} follows from~\cite{Aizenman+Barsky+Fernandez-1987}.
	All the other claims are proved (or follow from results) in~\cite{Pfister+Velenik-1999} (note that, although the focus in the latter paper is on the planar model, the arguments apply to the model on general graphs).
	We briefly sketch the argument with links to the relevant statements in~\cite{Pfister+Velenik-1999}.
	
	Remember that, given a high-temperature configuration \(\omega\subset\bbE_d\) with \(\partial\omega = \{x,y\}\), Algorithm~\ref{algo:path_extract} outputs both the path \(\Gamma(\omega) = \gamma\) and a set \(\bar\gamma\).
	The latter set can be obtained from \(\gamma\) (and the chosen ordering of the vertices of \(\Zd\)) in the following way: if \(\gamma=(\gamma_0, \dots, \gamma_n)\), then
	\[
		\bar\gamma = \setof{\{\gamma_i, j\}}{1\leq i\leq n-1,\, j\leq \gamma_{i+1}}.
	\]
	The paths \(\gamma':\,x\to y\) and \(\gamma:\,y\to z\) are compatible if and only if \(\gamma\) contains no edge of \(\bar\gamma'\). In that case, it is shown in~\cite{Pfister+Velenik-1999} that
	\[
		q(\gamma'\concatenate\gamma) = q(\gamma') q_{\calE(\gamma')}(\gamma),
	\]
	where \(\calE(\gamma') = \bbE^d \setminus \bar\gamma'\).
	In particular, \(q(\gamma \given \gamma') = q_{\calE(\gamma')}(\gamma)\).
	
	Let us see how~\ref{weight_property:repulsion} follows. For any admissible \(\gamma': z\to x\),
	\[
		\sum_{\gamma:\, x\to y} q(\gamma\given \gamma') = \sum_{\gamma:\, x\to y} 	q_{\calE(\gamma')}(\gamma) = \lrangle{\sigma_x\sigma_y}^0_{\calE(\gamma');\beta} \leq \lrangle{\sigma_x\sigma_y}^0_{\beta} = \sum_{\gamma:\, x\to y} q(\gamma),
	\]
	where the inequality is GKS2.
	
	It is proved in~\cite[Lemma~6.3]{Pfister+Velenik-1999} that the following volume-monotonicity property holds: \(q_E(\gamma) \geq q_{E'}(\gamma)\) for all \(E\subset E'\subset\bbE_d\) and \(\gamma\) admissible in \(E\). In particular, for all compatible \(\gamma'\) and \(\gamma\), 
	\begin{equation}\label{eq:LowerBndFactorWeight}
		q(\gamma \given\gamma') = q_{\calE(\gamma')}(\gamma) \geq q(\gamma) .
	\end{equation}
	
	A second consequence of volume-monotonicity is~\ref{weight_property:weight_growth}:
	\[
		q(\gamma\given\gamma') = q_{\calE(\gamma')}(\gamma) \leq \prod_{i=1}^{|\gamma|} \tanh(\beta J_{\gamma_i-\gamma_{i-1}}),
	\]
	since the right-most expression coincides with the weight of \(\gamma\) in a graph composed only of the edges of \(\gamma\).
	A third consequence of the volume-monotonicity property is that
	\[
		q(\gamma \given \gamma'\concatenate\gamma'')
		= q_{\calE(\gamma'\concatenate\gamma'')}(\gamma)
		\geq q_{\calE(\gamma'')}(\gamma)
		= q(\gamma \given \gamma''),
	\]
	which is~\ref{weight_property:monotonicity}.
	
	\medskip
	The proofs of the other claims rely on the following more explicit expression for the weight \(q(\cdot)\) derived in~\cite[Lemma~6.35]{Pfister+Velenik-1999}:
	\begin{equation}\label{eq:ExplicitWeight}
		q_E(\gamma) = \Bigl( \prod_{i=1}^{\gamma} \tanh(\beta J_{\gamma_i - \gamma_{i-1}}) \Bigr)
		\prod_{\{u,v\}\in\bar\gamma} \cosh(\beta J_{v-u}) \exp\Bigl( -\beta J_{v-u} \int_0^1 \lrangle{\sigma_u\sigma_v}^0_{E;\beta, J^s} \dd s\Bigr),
	\end{equation}
	where \(\lrangle{\cdot}^0_{E;\beta, J^s}\) denotes expectation with respect to the Gibbs measure with coupling constants \(J^s\) given by
	\[
		J^s_{u,v} =
		\begin{cases}
			J_{u,v} 	& \text{if }\{u, v\}\notin\bar\gamma,\\
			sJ_{u,v}	& \text{if }\{u, v\}\in\bar\gamma.
		\end{cases}
	\]
	Let \(\gamma', \gamma''\) and \(\gamma\) be contours such that \(\bend(\gamma') = \bend(\gamma'') = \fend(\gamma) = 0\) and such that \(\gamma\) contains no edge of either \(\bar\gamma'\) or \(\bar\gamma''\). Then, it follows from the previous expression that
	\[
		\frac{q(\gamma\given\gamma')}{q(\gamma\given\gamma'')} = \prod_{\{u,v\}\in\bar\gamma} \exp\Bigl( -\beta J_{v-u} \int_0^1 \bigl( \lrangle{\sigma_u\sigma_v}^0_{\calE(\gamma');\beta, J^s} - \lrangle{\sigma_u\sigma_v}^0_{\calE(\gamma'');\beta, J^s} \bigr) \dd s\Bigr).
	\]
	Exponential decay of correlations then immediately implies the mixing property~\ref{weight_property:mixing} (see~\cite[Lemma~6.8]{Pfister+Velenik-1999} or~\cite[Lemma~3.1]{Campanino+Ioffe+Velenik-2003}).
	
	Another direct consequence of~\eqref{eq:ExplicitWeight} is the lower bound
	\[
		q(\gamma)
		\geq 
		\Bigl( \prod_{i=1}^{\gamma} \tanh(\beta J_{\gamma_i - \gamma_{i-1}}) \Bigr)
		\Bigl(\prod_{i=0}^{|\gamma|-1} \frac{1+\sfe^{-2\beta J_{\gamma_i - \gamma_{i-1}}}}{2}\Bigr) 
		\geq
		\prod_{i=1}^{\gamma} C_\beta\tanh(\beta J_{\gamma_i - \gamma_{i-1}}),
	\]
	which yields~\ref{weight_property:lower_bound}.
	\ref{weight_property:finite_energy} also follows easily from the above identity, once one takes into account the geometrical constraints imposed upon the paths. Indeed,
	\begin{align}
		\frac{q(\gamma\given\gamma')}{q(\gamma)}
		&=
		\prod_{e=\{u,v\}\in\bar\gamma} \exp\Bigl( -\beta J_{v-u} \int_0^1 \bigl( \lrangle{\sigma_u\sigma_v}^0_{\calE(\gamma');\beta, J^s} - \lrangle{\sigma_u\sigma_v}^0_{\beta, J^s} \bigr) \dd s\Bigr) \notag\\
		&\leq 
		\exp\Bigl(\beta \sum_{\{u,v\}\in\bar\gamma} J_{v-u} \sum_{w\in\gamma'} \sfe^{-c(1+\norm{w-v})} \Bigr).
		\label{eq:BoundOnInteraction}
	\end{align}
	For the last inequality, see~\cite[proof of Lemma~3.1]{Campanino+Ioffe+Velenik-2003} for instance.
	Our assumptions on the paths imply that the sums over \(\{u,v\}\) and \(w\) in the right-hand side are uniformly bounded (see Fig.~\ref{fig:ConstrainedPaths}).
	\begin{figure}[ht]
		\includegraphics{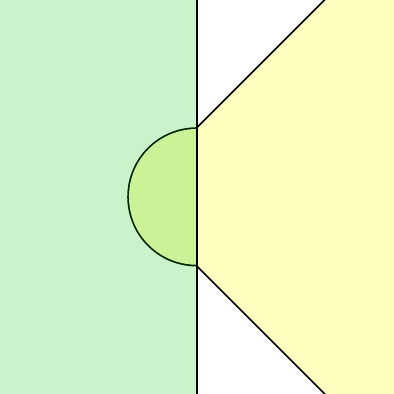}
		\caption{The path \(\gamma\) is constrained to lie inside the yellow region, while the path \(\gamma'\) is constrained to lie inside the green region. This implies that the sums over \(\{u,v\}\) and \(w\) in~\eqref{eq:BoundOnInteraction} are bounded above, uniformly in such paths \(\gamma\) and \(\gamma'\), by the corresponding sums over pairs of vertices lying one in the yellow region and the other one in the green region.}
		\label{fig:ConstrainedPaths}
	\end{figure}

\section{Irreducible interactions}

Let \(J\) be as in Section~\ref{subsec:interaction}. By the translation invariance and central symmetry of \(J\), it is equivalently defined as a function \(J:\Zd\to\R_{\geq 0}\).

In this appendix, we establish the properties about connectivities in the graph \((\Zd, \bbE_d)\) that are used in the local surgery of Section~\ref{sec:path_analysis}. Let \(\Lambda_N = \{-N, \dots, N\}^d\). First, for any \(0\leq R < \infty\), let us introduce
\[
	J^R_{ij} =
	\begin{cases}
		J_{ij} 	&	\text{ if } \normsup{i-j} \leq R,\\
		0 		&	\text{ otherwise}.
	\end{cases}
\]
Let us also define \(E^R_A = \bsetof{\{i,j\}\subset A}{J_{ij}^R > 0}\).

\begin{lemma}
	\label{app:irr_interaction:lem:connections_local}
	There exist \(R_0, R_1 < \infty\) such that, for any \(N\) and any \(x\in \Lambda_{N}\), \(x\) is connected to \(0\) using only edges in \(E_{\Lambda_{N+R_1}}^{R_0}\).
\end{lemma}
\begin{proof}
	By irreducibility, for each \(i\in\{1, \dots, d\}\) there exists a finite path \(\gamma^i: 0\to \rme_i\) (where \((\rme_i)_k = \mathds{1}_{\{i=k\}}, k=1,\cdots, d\)) with \(\{\gamma_k^i, \gamma_{k+1}^i\}\in\bbE_d\). Denote by \(-\gamma^i\) the image of \(\gamma^i\) under the central symmetry. Let \(R_1 = \max_{i\in\{1, \dots, d\}} \max_{x\in\gamma^i} \normsup{x}\). Let now \(N\geq 0\) and \(x=(x_1, \dots, x_d)\in\Lambda_N\). We can connect \(0\) and \(x\) via any path following \(|x_i|\) times the sequence of edges in \(\frac{x_i}{|x_i|}\gamma^i\) for \(i=1, \dots, d\). Any such path does not exit \(\Lambda_{N+R_1}\) by definition of \(R_1\). We therefore proved the claim with \(R_0 = \max_{i\in\{1, \dots, d\}} \max_{k\in\{1, \dots, |\gamma^i|\}} \normsup{\gamma^i_k-\gamma^i_{k-1}}\).
\end{proof}

\begin{lemma}
	\label{app:irr_interaction:lem:angular_opening_connected}
	There exists \(\epsilon>0\) such that, for any \(s\in \bbS^{d-1}\),
	\[
		\setof{x\in\Zd}{x\cdot s \geq \epsilon\norm{x}} \cap \setof{x\in\Zd}{J_x^{R_0}>0} \neq \emptyset,
	\]
	where \(R_0\) is given by Lemma~\ref{app:irr_interaction:lem:connections_local}.
\end{lemma}
\begin{proof}
	Let \(R=R_0\). Let \(A = \setof{\hat{x}}{x\in\Zd,\, J_x^{R}>0}\). It is a finite set. For \(s\in\bbS^{d-1}\), define
	\[
		f(s) = \max_{s'\in A} s'\cdot s.
	\]
	The claim is equivalent to the fact that \(f(s)\) is uniformly bounded away from \(0\). First, observe that, for any \(s\in \bbS^{d-1}\), \(f(s) > 0\), since \(A\) is a centrally symmetric set that spans \(\Rd\) (since \(\setof{x\in\Zd}{J_x^{R}>0}\) is centrally symmetric and generates \(\Zd\)). Since \(f\) is continuous over \(\bbS^{d-1}\), the conclusion follows from compactness of \(\bbS^{d-1}\).
\end{proof}

\bibliographystyle{plain}
\bibliography{BIGbib}

\end{document}